\documentclass[journal,twoside,web]{ieeecolor}
\usepackage{generic}
\usepackage{cite}
\usepackage{amsmath,amssymb,amsfonts}

\usepackage{amsthm}
\usepackage{graphicx,booktabs}
\usepackage{tikz}
\definecolor{nblue}{rgb}{0,0.263,0.576}\definecolor{mblue}{rgb}{0.075,0.541,0.855}
\definecolor{OIblue}{HTML}{0072B2}\definecolor{OIred}{HTML}{D55E00}
\definecolor{subsectioncolor}{rgb}{0,0,0}
\usepackage{textcomp}
\makeatletter
\let\NAT@parse\undefined   % 拆掉 ieeecolor 骗 hyperref 的假 natbib
\makeatother
\usepackage{hyperref}
\hypersetup{
  colorlinks=false,
  pdfborder={0 0 1},
  linkbordercolor={1 0 0},
  citebordercolor={0 1 0},
  urlbordercolor={0 1 1}
}
\let\labelindent\relax      % 让 ieeecolor 的 \labelindent 让位给 enumitem
\usepackage{enumitem}

\newcommand{\pot}{\mathrm{pot}}\newcommand{\harm}{\mathrm{harm}}\newcommand{\pos}{\mathrm{pos}}
\newcommand{\dv}{\operatorname{div}}\newcommand{\E}{\mathbb{E}}\newcommand{\R}{\mathbb{R}}
\newcommand{\calV}{\mathcal{V}}\newcommand{\calG}{\mathcal{G}}\newcommand{\calA}{\mathcal{A}}
\newcommand{\calN}{\mathcal{N}}\newcommand{\calD}{\mathcal{D}}
\newcommand{\Var}{\operatorname{Var}}\newcommand{\Cov}{\operatorname{Cov}}\newcommand{\Bin}{\mathrm{Bin}}
\newcommand{\numop}{\mathsf{N}}\newcommand{\one}{\mathbf{1}}\newcommand{\sa}{\mathsf{a}}
\newcommand{\Epot}{\mathcal{E}_{\pot}}\newcommand{\Eharm}{\mathcal{E}_{\harm}}
\newcommand{\fB}{\mathfrak{B}}
\newtheorem{theorem}{Theorem}\newtheorem{lemma}{Lemma}
\newtheorem{proposition}{Proposition}\newtheorem{corollary}{Corollary}
\theoremstyle{definition}\newtheorem{example}{Example}
\theoremstyle{remark}\newtheorem{remark}{Remark}
\usepackage{float,array,wrapfig}
\theoremstyle{plain}
\newtheorem{lemmaS}{Lemma}
\newtheorem{propS}[lemmaS]{Proposition}\newtheorem{thmS}[lemmaS]{Theorem}\newtheorem{corS}[lemmaS]{Corollary}
\newtheorem*{mainthm}{Theorem}\newtheorem*{mainlem}{Lemma}
\theoremstyle{definition}\newtheorem*{definition*}{Definition}
\theoremstyle{remark}\newtheorem{remarkS}[lemmaS]{Remark}
\newcommand{\dk}{\varkappa}\newcommand{\calB}{\mathcal{B}}
\newcommand{\ff}[2]{#1^{\underline{#2}}}\newcommand{\sgn}{\operatorname{sgn}}

\makeatletter
\def\section{\@startsection{section}{1}{\z@}{3.0ex plus 1.5ex minus 1.5ex}%
{0.7ex plus 1ex minus 0ex}{\centering\fontencoding{T1}\fontfamily{phv}\fontseries{m}\fontshape{n}%
\fontsize{10}{12}\selectfont\scshape\noindent\color{nblue}}}%
\def\subsubsection{\@startsection{subsubsection}{3}{\parindent}{0ex plus 0.1ex minus 0.1ex}%
{0ex}{\normalfont\fontsize{9}{11}\selectfont\everymath={\sf}\sf\itshape\color{nblue}}}%
\makeatother

\begin{document}
\title{Priority Coordination Games: Hodge Decomposition and a Sharp Design Limit}

\author{Zhihao Lin, Jianglin Lan, Anh-Tu~Nguyen,~\IEEEmembership{Senior Member,~IEEE,} and Yoshinobu Kawahara% <-this % stops a space
\thanks{This work was supported in part by the Japan Society for the Promotion of Science (JSPS)
through the JSPS Postdoctoral Fellowship for Research in Japan (Fellow ID P26385), in part by
JSPS KAKENHI Grant Number JP26H02498, and in part by the Japan Science and Technology Agency (JST)
Moonshot R\&D Program under Grant No. JPMJMS263K.
(Corresponding author: Yoshinobu Kawahara.)}
\thanks{Zhihao Lin is with the Graduate School of Information Science and Technology, The University of Osaka,
Suita, Osaka 565-0871, Japan (e-mail: zhihao.lin@ist.osaka-u.ac.jp).}%
\thanks{Jianglin Lan is with the James Watt School of Engineering, University of Glasgow, Glasgow G12 8QQ, United Kingdom (e-mail: Jianglin.Lan@glasgow.ac.uk).}%
\thanks{Anh-Tu Nguyen is with the LAMIH UMR CNRS 8201 Laboratory, Universit\'e Polytechnique Hauts-de-France, 59313 Valenciennes, France, and also with INSA Hauts-de-France, 59313 Valenciennes, France (e-mail: tnguyen@uphf.fr).}%
\thanks{Yoshinobu Kawahara is with the Graduate School of Information Science and Technology, The University of Osaka, Suita, Osaka 565-0871, Japan, and also with the RIKEN Center for Advanced Intelligence Project, Tokyo 103-0027,
Japan (e-mail: kawahara@ist.osaka-u.ac.jp).}
}

\maketitle

\begin{abstract}
In decentralised priority coordination, agents announce priority levels and a shared resource serves them in
decreasing order, as at an unsignalised intersection; the levels form the decision layer of a hierarchical controller.
Such interactions are routinely replaced by a potential game, i.e.\ by a common objective, for analysis and
design. This paper determines what that surrogate misses, using the Hodge decomposition of the incentives into a
potential component, which a common objective can represent, and a harmonic component, which it cannot. For
the linear payoff, both components are obtained in closed form on every conflict graph and for every deterministic
tie-breaking protocol: in common units, the harmonic energy is the number of conflicts and the potential energy adds
the number of adjacent pairs of conflicts. Consequently, for every rationality parameter, the best common-objective model
of the agents' choice log-odds, weighted uniformly over unilateral moves, has a relative squared error of at least
$1/(d_{\max}+1)$, where $d_{\max}$ is the largest number of conflicts of one agent; for an eight-vehicle intersection
it is exactly one fifth, for any number of priority levels. Invisible to strict-improvement dynamics, the missed component is, under
low-rationality log-linear learning with uniform revision and to leading order, the stationary probability current,
and its energy sets the entropy-production rate. Payoff design cannot remove it: on
the complete conflict graph of $N$ agents, under a total-order protocol and with at least three priority levels, every
nonconstant rank-based payoff leaves a relative error of at least $1/N$, with equality exactly for affine payoffs.
\end{abstract}

%======================================================================
\section{Introduction}\label{sec:intro}
\IEEEPARstart{M}{any} multi-agent systems resolve conflicts by priority. Vehicles at an unsignalised intersection must
decide who crosses first \cite{dresner2008,levin2017,xu2018,zanardi2021,li2023,TAVV6081,11130528}; robots meeting at
a corridor junction and transmitters contending for a channel face the same problem. In \emph{decentralised priority
coordination}, each of $N$ agents announces one of $K$ priority levels, and the shared resource serves the agents in decreasing order of
level, ties being broken by a fixed protocol. Every agent prefers to be served early, so the agents play a game.

Priority selection is a control problem, not only a static game. At an intersection, priority levels fix the order in
which vehicles enter the conflict zone, and trajectory controllers then realise that order \cite{levin2017,xu2018}.
Read this way, the priority levels form the decision layer of a hierarchical controller, the conflict graph encodes
the coupling of the agents through shared resources, and the learning rule governs how decisions are updated over
time. We study this decision layer and assume that the execution-layer controllers execute the service order it produces.
This is an abstraction, since an ordering of the agents restricts what the execution layer could do; here we take it
as given and analyse the incentives inside it.

A standard way to analyse, predict and design such decision layers is to model them as potential games
\cite{monderer1996,marden2009}. In a potential game, the change in any agent's payoff caused by its own deviation
equals the change of one common function, the potential. This brings two guarantees. First, the potential is a
Lyapunov function for asynchronous best-response dynamics with inertia, so these dynamics converge
\cite{monderer1996}. Second, under log-linear learning the stationary distribution is explicit: it is the Gibbs
distribution of the potential \cite{blume1993,marden2012}. For this reason, utility design for self-interested agents
has become a standard problem in control \cite{marden2014efficiency,paccagnan2020utilityI,paccagnan2022utilityII},
and for interacting vehicles, potential-game models are used for prediction and planning \cite{liu2023}, a game that is
not potential being projected onto the nearest potential game and the residual discarded \cite{liu2025}. Whether
such a surrogate is adequate depends on how much it misses, and that quantity is rarely computed.

\begin{figure}[t]\centering
\begin{tikzpicture}[x=1.42cm,y=1.42cm,>=stealth,
  st/.style={circle,draw=black,text=black,inner sep=0.8pt,minimum size=15.5pt,font=\footnotesize},
  lab/.style={font=\footnotesize,fill=white,inner sep=1pt},
  every edge/.style={draw,->,thick}]
 % (a) game flow
 \begin{scope}
  \node[st] (a00) at (0,0) {00}; \node[st] (a10) at (1,0) {10};
  \node[st] (a01) at (0,1) {01}; \node[st] (a11) at (1,1) {11};
  \path (a00) edge node[lab,below]{$0$} (a10)
        (a01) edge node[lab,above]{$1$} (a11)
        (a00) edge node[lab,left]{$1$} (a01)
        (a10) edge node[lab,right]{$0$} (a11);
  \node[font=\footnotesize] at (0.5,-0.58) {(a) game flow $u$};
 \end{scope}
 \node at (1.6,0.5) {\large$=$};
 % (b) potential part
 \begin{scope}[shift={(2.2,0)},every edge/.append style={OIblue}]
  \node[st] (b00) at (0,0) {00}; \node[st] (b10) at (1,0) {10};
  \node[st] (b01) at (0,1) {01}; \node[st] (b11) at (1,1) {11};
  \path (b00) edge node[lab,below]{$\frac12$} (b10)
        (b01) edge node[lab,above]{$\frac12$} (b11)
        (b00) edge node[lab,left]{$\frac12$} (b01)
        (b10) edge node[lab,right]{$\frac12$} (b11);
  \node[font=\footnotesize] at (0.5,-0.58) {(b) potential $\delta\phi^\star$};
 \end{scope}
 \node at (3.8,0.5) {\large$+$};
 % (c) harmonic part
 \begin{scope}[shift={(4.4,0)},every edge/.append style={OIred}]
  \node[st] (c00) at (0,0) {00}; \node[st] (c10) at (1,0) {10};
  \node[st] (c01) at (0,1) {01}; \node[st] (c11) at (1,1) {11};
  \path (c00) edge node[lab,left]{$\frac12$} (c01)
        (c01) edge node[lab,above]{$\frac12$} (c11)
        (c11) edge node[lab,right]{$\frac12$} (c10)
        (c10) edge node[lab,below]{$\frac12$} (c00);
  \draw[<-,OIred!60] (0.5,0.28) arc[start angle=-90,end angle=200,radius=0.24];
  \node[font=\footnotesize] at (0.5,-0.58) {(c) harmonic $h$};
 \end{scope}
\end{tikzpicture}
\vspace{-3mm}
\caption{The smallest priority game: two conflicting agents, levels $\{0,1\}$, agent~$0$ served first on a tie.
Vertices are profiles $s_0s_1$; arrows are unilateral raises, labelled with the mover's payoff gain.
(a)~The game flow; (b)~its potential part (blue), the common incentive to raise one's level; (c)~its harmonic part (red),
a circulation recording who goes first.}\label{fig:toy}
\end{figure}
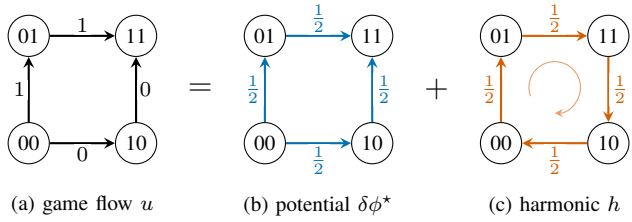

Priority coordination is not a potential game as soon as two agents are in conflict
(Propositions~\ref{prop:pairwise}(d) and~\ref{prop:energies}). The reason is elementary: when two conflicting agents
choose the same level, only one of them can go first. Fig.~\ref{fig:toy}(a) shows the smallest case. Agent~$1$ gains
by raising its level when both are low, agent~$0$ then gains by raising its level too, and nobody gains by lowering.
In a potential game, the gains collected around any closed cycle of profiles add up to zero \cite{monderer1996}; here
they add up to $2$. This paper asks three questions about the part of priority coordination that a potential surrogate
misses: \emph{How large is it? What does it do to the agents' behaviour? Can payoff design remove it?} Classical
utility design does not address the last question: there, utilities are designed mainly to improve welfare or the
price of anarchy \cite{marden2014efficiency,paccagnan2020utilityI}, whereas in priority coordination the total number
of waits is the same in every profile and for every payoff, because on every conflict exactly one agent waits. What a
payoff can change is how far the game is from a potential game, and we determine a limit on this distance that no
payoff can overcome.

Our tool is the Hodge decomposition of finite games \cite{candogan2011}. It splits the incentives of every game,
viewed as a flow on the graph of profiles, into two orthogonal components (Fig.~\ref{fig:toy}(b)--(c)). The
\emph{potential} component is the part that a common objective can represent; the \emph{harmonic} component is the
remainder, which circulates. The squared norms of the two components, the \emph{potential} and \emph{harmonic
energies}, measure how much of the incentives each carries, and the harmonic energy is exactly the smallest squared
error of any potential surrogate. The decomposition can always be computed numerically; our aim is to obtain it in
closed form and to determine what it implies.

Our answers to the three questions are as follows.

\emph{Size.} The harmonic component of a pairwise game is the \emph{antisymmetric} part of its interactions
(Proposition~\ref{prop:pairwise}). For the linear payoff $f(r)=-r$, nothing in an interaction is reciprocal, because
exactly one of two conflicting agents goes first, so both energies are counts on the conflict graph
(Proposition~\ref{prop:energies}). Since the harmonic energy is the smallest squared error of any potential surrogate,
the best common-objective model of the agents' choice log-odds misses at least $1/(d_{\max}+1)$ of the incentive
energy for every rationality $\beta>0$ (Corollary~\ref{cor:surrogate}); the floor is set by the conflicts per agent, not by
the number of agents. For the eight-vehicle intersection of Fig.~\ref{fig:intersection}, with $16$ conflicts and $48$
adjacent pairs, it is one fifth.

\emph{Behaviour.} Under log-linear learning at low rationality, the missed component is what makes learning
irreversible: the stationary law is still Gibbs in the Hodge potential up to second-order terms, but the
leading-order stationary current is proportional to the harmonic flow and the leading-order entropy-production rate
to the harmonic energy (Proposition~\ref{prop:learning}). It is, in contrast, invisible to ordinal analysis: paths of
strictly improving moves always terminate (Proposition~\ref{prop:fip}), and, under the index protocol, all
pure-strategy Nash equilibria induce the same service order (Proposition~\ref{prop:ne}); in Fig.~\ref{fig:toy} both
equilibria, $10$ and $11$, serve agent~$0$ first. Discarding the harmonic component is therefore harmless for
best-response reasoning and costly for any model that uses the magnitudes of the incentives.

\emph{Design.} The harmonic component cannot be designed away. On the complete conflict graph, under a total-order
protocol and with at least three priority levels, no payoff that depends only on the rank brings the harmonic share,
and hence the surrogate error, below $1/N$, and affine payoffs attain this value (Theorem~\ref{thm:finiteK}); for
continuous priorities the full spectrum of the problem is explicit (Theorem~\ref{thm:spectrum}). This is a limit on
potential-game approximation, not a recommendation of the affine payoff.

The main contributions of this paper are the following.
\begin{enumerate}
\item \emph{The error of a potential surrogate, in closed form} (Proposition~\ref{prop:energies},
Corollary~\ref{cor:surrogate}): for the linear payoff, on every conflict graph, for every $K\ge2$ and every
deterministic protocol. Prior work assumes a bound on this distance \cite{candogan2013,candogan2013geb} or discards it
\cite{liu2025}.
\item \emph{Learning identities and ordinal invisibility} (Propositions~\ref{prop:learning}, \ref{prop:fip}
and~\ref{prop:ne}): exact leading-order identities for arbitrary payoffs as the rationality tends to zero, where the
potential share was so far an empirical predictor \cite{bichler2025}; termination and a unique equilibrium service
order for every strictly decreasing payoff.
\item \emph{A sharp design limit} (Theorems~\ref{thm:spectrum} and~\ref{thm:finiteK}): on the complete conflict
graph, under a total-order protocol, for every $K\ge3$ and every nonconstant payoff-from-rank function; it can fail
for $K=2$ and, at finite $K$, off the complete graph (Section~\ref{sec:beyond}). It is an impossibility result on the
distance to a potential game, not a welfare optimisation \cite{marden2014efficiency,paccagnan2020utilityI}; the
computer-assisted steps of its proof are exact certificates.
\end{enumerate}

\subsection{Related Work}
\emph{Decompositions of games.} The potential--harmonic decomposition is due to Candogan \emph{et
al.}~\cite{candogan2011}; related decompositions and characterisations appear in
\cite{sandholm2010,hwang2011,arefizadeh2024,abdou2022,gilles2026}, and, in the control literature, through
the semi-tensor product \cite{cheng2014finite,cheng2015networked,cheng2016,hao2022survey,zhang2017}. For graphical
games the decomposition acts edge by edge \cite{cheng2016,arditti2020sep};
Proposition~\ref{prop:pairwise} specialises these results to centred interaction kernels, and we claim no novelty
beyond this form. Symmetrised interactions give approximate potentials of network games with continuous
actions \cite{guo2024,rokade2025}.

\emph{Potential games in control and learning.} Potential-game design is a standard tool of cooperative control
\cite{marden2009,marden2012}. The potential share has been proposed as a predictor of learning behaviour \cite{bichler2025}; near-potential
games inherit approximate guarantees \cite{candogan2013,candogan2013geb}; harmonic structure makes no-regret dynamics
recurrent \cite{legacci2024,legacci2024noregret}; bounded-rationality log-linear learning has been analysed for
symmetric coordination games \cite{boundedLLL2026}; and currents and entropy production driven by non-reciprocal
interactions appear in statistical physics and evolutionary games \cite{maes2010,szabo2015,szabo2016}. We are not aware of closed-form energies for
rank-order games, of identities linking the discarded residual to learning, or of sharp bounds on the potential share
over a family of payoffs.

\emph{Rank-order games and the affine case.} Rank-order games are studied computationally in
\cite{brandt2009ranking}. In weighted congestion games, affine costs are singled out by an exact
potential \cite{harks2011}; here affine payoffs minimise, rather than remove, the harmonic share. The spectral analysis
uses the Hoeffding decomposition of functions of ranks \cite{helmers1982} and discrete Chebyshev polynomials
\cite{koekoek2010}.

\subsection{Organisation and Notation}
\emph{Organisation.} Section~\ref{sec:setup} states the model and the three problems; Sections~\ref{sec:pairwise}--\ref{sec:finiteK} answer them in turn, and Section~\ref{sec:numerics} reports the numerical study. Short proofs are given in full in the text. The proofs of Theorems~\ref{thm:spectrum} and~\ref{thm:finiteK}
are developed in the Appendix, with lengthy algebraic reductions and exact computer-assisted certificates deferred to
the supplementary material (SM), which also contains further results. For these,
the text also explains why the result holds.

\emph{Notation.} $[N]=\{0,\dots,N-1\}$; $\#A$ is the number of elements of a finite set $A$; $\one$ is the all-ones
vector; $x^{\underline k}=x(x-1)\cdots(x-k+1)$ is the falling factorial; $\Delta f(r)=f(r+1)-f(r)$ is the forward
difference and $\Delta^m$ its $m$-th iterate; $\Bin(n,p)$ is the binomial law; $M\succeq0$ means that $M$ is positive
semidefinite. For a finite set $S$, functions on $S$ carry the uniform probability measure, $\|g\|^2=\E[g^2]$, and
$\Var$, $\Cov$ refer to it. Flows (functions on edges) carry the counting inner product $\langle
x,y\rangle=\sum_e x(e)y(e)$; the factor $K^N$ between the two conventions is always written out. $A_i$ averages over
the $i$-th coordinate, $g_T$ is the Hoeffding component of $g$ on $T\subseteq[N]$ (Section~\ref{sec:setup}), and
$[u_i]_T$ is that of the payoff $u_i$.

%======================================================================
\section{Model and Problem Statement}\label{sec:setup}
This section defines pairwise games, for which the Hodge decomposition has a simple form, and the priority games of
interest, which are pairwise when the payoff is affine; it then states the three problems and the lemma from which the
closed forms are derived.

\subsection{Pairwise Games on an Interconnection Graph}
Let $C=(\calV,E)$ be a finite undirected graph on the agent set $\calV=[N]$, the \emph{interconnection} (or
\emph{conflict}) \emph{graph}; $\calN(i)$ is the neighbourhood of $i$ and $d_i=|\calN(i)|$ its degree. Each agent
chooses a \emph{priority level} (\emph{level} for short) $s_i\in\{0,\dots,K-1\}$, $K\ge2$; the profile is $s\in S:=\{0,\dots,K-1\}^N$. A \emph{pairwise game}
on $C$ has payoffs
\begin{equation}\label{eq:pairwise}
u_i(s)=\gamma_i(s_i)+\sum_{j\in\calN(i)}\psi_{ij}(s_i,s_j),
\end{equation}
with arbitrary own terms $\gamma_i$ and arbitrary kernels $\psi_{ij}$, one for each \emph{ordered} pair of neighbours.
Terms that do not depend on $s_i$ are strategically irrelevant and are omitted. Decomposing each kernel by the
two-variable ANOVA under the uniform measure gives
\begin{equation}\label{eq:anova}
\psi_{ij}(x,y)=\overline{\psi}_{ij}+a_{ij}(x)+b_{ij}(y)+\omega_{ij}(x,y),
\end{equation}
where $a_{ij}$, $b_{ij}$ are centred and the \emph{interaction} $\omega_{ij}$ is centred in each argument. The
\emph{own effect} of agent $i$ is
\begin{equation}\label{eq:own}
\sa_i(x):=\gamma_i(x)-\bar\gamma_i+\sum_{j\in\calN(i)}a_{ij}(x),
\end{equation}
where $\bar\gamma_i$ is the mean of $\gamma_i$. Each interaction splits into its symmetric and antisymmetric parts under the exchange of the two agents,
\begin{equation}\label{eq:pm}
\omega^{\pm}_{ij}(x,y):=\tfrac12(\omega_{ij}(x,y)\pm\omega_{ji}(y,x)),
\end{equation}
so that $\|\omega^{\pm}_{ij}\|$ depends only on the edge $\{i,j\}$.

\subsection{Priority (Rank-Order) Games}
Agents are served in decreasing order of level, ties being broken by a fixed \emph{protocol}. Since only neighbours
compete, we model a deterministic protocol as an orientation of the conflict edges: we write $j\to i$ if, on a tie,
$j$ is served before $i$. Throughout the paper, ``any deterministic protocol'' means any such orientation. Unless stated
otherwise we use the \emph{index protocol}, $j\to i$ iff $j<i$, i.e.\ the acyclic orientation induced by a total order
of the agents; after relabelling, every total order is of this form. The \emph{conflict position} of agent $i$ is the
number of its neighbours served before it,
\begin{equation}\label{eq:pos}
\!\!\pos_C(i,s)=\#\{j\in\calN(i): s_j>s_i \ \text{or} \ (s_j=s_i,\ j\to i)\}, \!\!
\end{equation}
and $u_i(s)=f(\pos_C(i,s))$ for a \emph{payoff-from-rank} function
$f$, which is strictly decreasing unless stated otherwise. The \emph{linear} game
$\calG(C,K)$ has $f(r)=-r$; it is the pairwise game with $\gamma_i\equiv0$ and
\begin{equation}\label{eq:rankkernel}
\psi_{ij}(x,y)=-\one[y>x]-\one[y=x]\,\one[j\to i].
\end{equation}
For general $f$ we write $\calG_f(C,K)$, which is pairwise iff $f$ is affine on $\{0,\dots,d_{\max}\}$, with $d_{\max}$ being the
maximum degree of $C$ (an agent of degree $d$ only ever sees ranks $0,\dots,d$). Since every conflict edge
contributes one unit to the position of whichever endpoint is served second,
\begin{equation}\label{eq:cons}
\sum_{i}\pos_C(i,s)=|E|,\qquad \forall s\in S.
\end{equation}
Thus the total number of waits, i.e.\ of conflicting agents served before each agent, is the same in every profile.
The aggregate payoff $\sum_if(\pos_C(i,s))$ is also constant when $f$ is affine, but for a nonlinear $f$ it need not
be. An orientation with directed cycles still
defines $\pos_C$ and the game, but only as a system of local comparisons (three mutually tied agents on a directed
triangle admit no global service order); statements about pure-strategy Nash equilibria (Proposition~\ref{prop:ne}) are made for the
index protocol, and, for the linear payoff, the energies are the same for every orientation.

\begin{example}[Intersection]\label{ex:intersection}
At a two-lane four-arm intersection with eight vehicles, four turning on the inner lanes and four going straight on
the outer ones, the conflict graph has an edge for every pairwise crossing of the nominal paths
(Fig.~\ref{fig:intersection}): it is $4$-regular with $16$ edges, the turning vehicles forming a $4$-cycle, the
straight ones another, and each turning vehicle crossing two straight lanes. A vehicle's delay is, to first order,
proportional to $\pos_C$, the number of conflicting vehicles it lets through. This graph is our running example.
\end{example}

\begin{figure}[t]\centering
\includegraphics[width=\columnwidth]{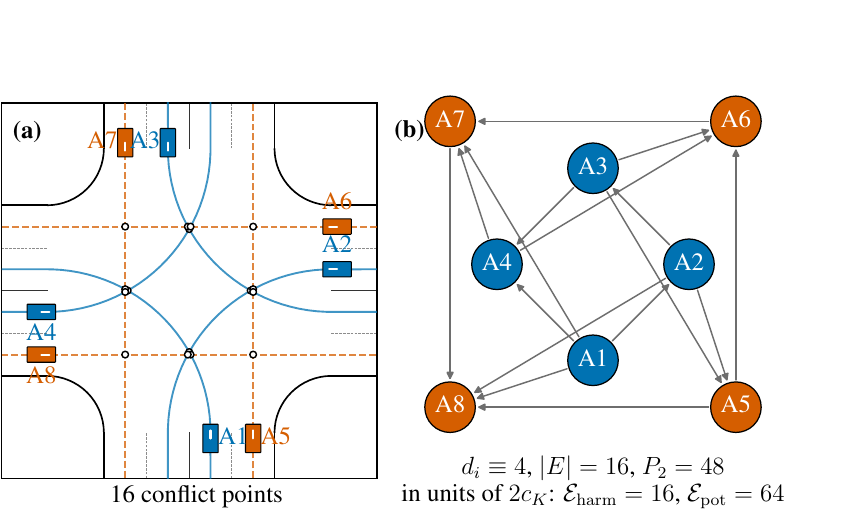}
\vspace{-2em}
\caption{Running example. (a)~Eight vehicles at a two-lane intersection: A1--A4 turn left from the inner lanes
(solid), A5--A8 go straight on the outer lanes (dashed); circles mark the $16$ pairwise crossings of the nominal paths.
(b)~The resulting conflict graph; arrows show the index protocol (on a tie, tail first).}\label{fig:intersection}
\end{figure}

\subsection{Flows and Hodge Decomposition}
We use the flow representation of \cite{candogan2011}. The \emph{game graph} has vertex set $S$ and an edge
$e=(i,s\to s')$ for every agent $i$, profile $s$, and level $s'_i>s_i$ (with $s'_{-i}=s_{-i}$). The game's \emph{flow}
is $u(i,s\to s')=u_i(s')-u_i(s)$; the gradient of $\phi:S\to\R$ is $(\delta\phi)(i,s\to s')=\phi(s')-\phi(s)$; and
$(\dv x)(s)=\sum_{e\,\text{into}\,s}x(e)-\sum_{e\,\text{out of}\,s}x(e)$, so that
$\langle\delta\phi,x\rangle=\sum_s\phi(s)(\dv x)(s)$. Every flow decomposes uniquely and orthogonally as
\begin{equation}\label{eq:hodge}
u=\delta\phi^\star+h,\qquad \dv h=0,
\end{equation}
with $\delta\phi^\star$ the projection of $u$ onto $\operatorname{im}\delta$; the game is an exact potential game iff
$h=0$. We call
\begin{equation}\label{eq:energies}
\Epot:=\|\delta\phi^\star\|^2,\qquad \Eharm:=\|h\|^2
\end{equation}
the \emph{potential} and \emph{harmonic energies}, respectively; ``energy'' means this squared norm, not physical energy. The \emph{potential share} $\rho_\pot:=\Epot/(\Epot+\Eharm)$ and
$\rho_\harm:=1-\rho_\pot$ are normalised summaries (the discussion after Corollary~\ref{cor:rigid} states when they can be compared across designs).

\subsection{Problem Statement}\label{subsec:problem}
Once a conflict graph $C$, a number $K$ of levels, a protocol and a payoff-from-rank function $f$ are fixed, the
game $\calG_f(C,K)$ and the pair $(\Epot,\Eharm)$ of \eqref{eq:energies} are determined. We
address three problems. (i)~\emph{Closed form}: determine $\Epot$ and $\Eharm$ as functions of $C$, $K$ and $f$
(Proposition~\ref{prop:energies} for the linear game, for every protocol). (ii)~\emph{Meaning}: determine what the two
energies imply for the behaviour of the agents, both for the convergence of improvement dynamics
(Propositions~\ref{prop:fip}--\ref{prop:ne}) and for the stationary behaviour of log-linear learning
(Proposition~\ref{prop:learning}). (iii)~\emph{Design}: under a total-order protocol on the complete conflict graph,
determine the smallest achievable harmonic share,
\begin{equation}\label{eq:problem}
\rho_\harm^{\min}(N,K):=\inf_{f\ \text{nonconstant}}\rho_\harm\big(\calG_f(K_N,K)\big).
\end{equation}
The infimum is over all nonconstant $f$, not only decreasing ones (Theorems~\ref{thm:spectrum}
and~\ref{thm:finiteK}).
Lemma~\ref{lem:numop} below reduces~(i) and~(iii), and the learning part of~(ii), to computations with the
Hoeffding components of a single function.

\subsection{The Number-Operator Lemma}

Let $A_i$ denote averaging over the $i$-th coordinate. Every $g:S\to\R$ has the Hoeffding (ANOVA) decomposition
$g=\sum_{T\subseteq[N]}g_T$, $g_T=\prod_{j\in T}(I-A_j)\prod_{j\notin T}A_j\,g$; the components are mutually orthogonal,
$g_T$ depends on $s_T$ only and is centred in each of its variables. Put $g_m=\sum_{|T|=m}g_T$. The \emph{number
operator} $\numop:=\sum_i(I-A_i)$ multiplies $g_T$ by $|T|$. Lemma~\ref{lem:numop} yields all closed forms: since the
game graph is a product of complete graphs, its Laplacian multiplies an $m$-agent Hoeffding component by $m$, so the
Hodge potential divides each component by its order. Its content is implicit in \cite[Sec.~4.1]{candogan2011}, and \eqref{eq:phistar}
follows from the inverse of $\numop$ in \cite{zhang2017} and is stated in \cite{gilles2026}
(Section~\ref{sec:intro}); we state it in the graded form in which it is used.

\begin{lemma}[Number-operator lemma]\label{lem:numop}
For any payoffs $(u_i)$ on $S$, let $g:=\sum_i(u_i-A_iu_i)$. Then, up to an additive constant,
\begin{equation}\label{eq:phistar}
\phi^\star=\numop^{-1}g=\sum_{T\neq\emptyset}\frac{g_T}{|T|},
\end{equation}
and
\begin{equation}\label{eq:twoenergies}
\begin{aligned}
\Epot&=K^{N+1}\sum_{m\ge1}\frac{\|g_m\|^2}{m},\\
\Epot+\Eharm&=K^{N+1}\sum_i\|u_i-A_iu_i\|^2 .
\end{aligned}
\end{equation}
\end{lemma}

\begin{proof}
The game graph is the Cartesian product of \(N\) copies of the complete graph on \(K\) vertices, whose Laplacian is \(K(I-A)\); hence \(\dv\delta=K\sum_i(I-A_i)=K\numop\). At a profile \(s\),
\begin{align*}
(\dv u)(s)&=\sum_i\sum_t[u_i(s)-u_i(t,s_{-i})]\\
&=K\sum_i(u_i-A_i u_i)(s)=Kg(s).
\end{align*}
Thus the normal equation \(\dv\delta\phi^\star=\dv u\) reduces to \(\numop\phi^\star=g\).
Since each \(A_i\) preserves the uniform mean, \(g\) has zero mean. Moreover, for the orthogonal decomposition \(g=\sum_T g_T\), one has \(\numop g_T=|T|g_T\). Therefore $\phi^\star=\numop^{-1}g$ up to an additive constant, which is \eqref{eq:phistar}.
By adjointness and the normal equation,
\begin{align*}
\Epot&=\langle\delta\phi^\star,\delta\phi^\star\rangle_{\rm edge}
=\sum_s\phi^\star(s)(\dv u)(s)\\
&=K^{N+1}\langle\numop^{-1}g,g\rangle.
\end{align*}
Writing \(g_m=\sum_{|T|=m}g_T\) and using orthogonality gives \(\langle\numop^{-1}g,g\rangle=\sum_{m\ge1}\|g_m\|^2/m\), hence the first identity in \eqref{eq:twoenergies}.
For the second, fix \(i\) and \(s_{-i}\). Since \(\sum_{x<y}(f(y)-f(x))^2=K^2\Var(f)\),
\[
\sum_{x<y}\bigl(u_i(y,s_{-i})-u_i(x,s_{-i})\bigr)^2
=K^2\Var_{s_i}u_i(\cdot,s_{-i}).
\]
Summing over \(s_{-i}\) gives \(\|u\|_{\rm edge}^2=K^{N+1}\sum_i\|u_i-A_i u_i\|^2\). Since the potential and harmonic components are orthogonal, \(\Epot+\Eharm=\|u\|_{\rm edge}^2\), proving the second identity.
\end{proof}

%======================================================================
\section{Closed-Form Energies of Priority Coordination}\label{sec:pairwise}
This section solves the first problem described in Section~\ref{subsec:problem}: computing the potential and harmonic
components of priority coordination in closed form. The main result, Proposition~\ref{prop:energies}, shows that for
the linear payoff the harmonic energy counts the conflicts and the potential energy adds the adjacent pairs of
conflicts; Section~\ref{sec:reading} then reads the two components for a single move.

\subsection{Symmetric and Antisymmetric Interactions}\label{sec:pairwise-gen}
Consider two neighbours $i$ and $j$. Their interaction $\omega_{ij}$ splits, by \eqref{eq:pm}, into a symmetric part
$\omega^+_{ij}$, which affects both agents in the same way, and an antisymmetric part $\omega^-_{ij}$, in which what one
agent gains the other loses. Intuitively, only the symmetric part can be represented by a common objective that changes in the same way for both agents, whereas under the antisymmetric part their interests are opposed.
Proposition~\ref{prop:pairwise} makes this exact and gives the two energies.

\begin{proposition}[Pairwise games, kernel form]\label{prop:pairwise}
Let $\calG$ be the pairwise game \eqref{eq:pairwise} on a graph $C$ with at least one edge, $K\ge2$. Then:
\begin{enumerate}
\item[(a)] the Hodge potential is, up to a constant,
\begin{equation}\label{eq:polyphi}
\phi^\star(s)=\sum_i\sa_i(s_i)+\sum_{\{i,j\}\in E}\omega^{+}_{ij}(s_i,s_j).
\end{equation}
\item[(b)] the harmonic flow $h$ is the flow of the payoffs $h_i(s)=\sum_{j\in\calN(i)}\omega^{-}_{ij}(s_i,s_j)$.
\item[(c)] the energies are
\begin{equation}\label{eq:polyenergies}
\begin{aligned}
\Epot&=K^{N+1}\Big(\sum_i\|\sa_i\|^2+2\!\!\sum_{\{i,j\}\in E}\!\!\|\omega^{+}_{ij}\|^2\Big),\\
\Eharm&=2K^{N+1}\!\!\sum_{\{i,j\}\in E}\!\!\|\omega^{-}_{ij}\|^2.
\end{aligned}
\end{equation}
\item[(d)] $\calG$ is an exact potential game iff every interaction is symmetric,
$\omega_{ij}(x,y)=\omega_{ji}(y,x)$, for every edge.
\end{enumerate}
\end{proposition}
\begin{proof}
\emph{Step 1.} $I-A_i$ annihilates every term of \eqref{eq:pairwise}--\eqref{eq:anova} not depending on $s_i$ and
fixes $\sa_i(s_i)$ and $\omega_{ij}(s_i,s_j)$, which are centred in $s_i$. Hence
$v_i:=u_i-A_iu_i=\sa_i(s_i)+\sum_{j\in\calN(i)}\omega_{ij}(s_i,s_j)$.
\emph{Step 2.} In $g=\sum_iv_i$, the term $\sa_i$ is a Hoeffding component on $T=\{i\}$, and $\omega_{ij}$ is a
Hoeffding component on $T=\{i,j\}$. Collecting by index set gives
\[
g_{\{i\}}=\sa_i,\quad g_{\{i,j\}}=\omega_{ij}(s_i,s_j)+\omega_{ji}(s_j,s_i)=2\omega^{+}_{ij}(s_i,s_j)
\]
for $\{i,j\}\in E$, and $g_T=0$ otherwise.
\emph{Step 3.} Lemma~\ref{lem:numop} gives $\phi^\star=\sum_Tg_T/|T|$, which is \eqref{eq:polyphi}. The payoffs
$\hat u_i:=\sa_i+\sum_j\omega^{+}_{ij}$ differ from $\phi^\star$ by a function of $s_{-i}$, so their flow is
$\delta\phi^\star$; subtracting from $v_i$ gives (b).
\emph{Step 4.} By \eqref{eq:twoenergies} and orthogonality,
$\Epot/K^{N+1}=\sum_i\|\sa_i\|^2+\sum_E\|2\omega^{+}_{ij}\|^2/2$ and
$(\Epot+\Eharm)/K^{N+1}=\sum_i\|\sa_i\|^2+\sum_E(\|\omega_{ij}\|^2+\|\omega_{ji}\|^2)$. Since the uniform measure is
exchange-invariant, $\|\omega_{ij}\|^2+\|\omega_{ji}\|^2=2\|\omega^{+}_{ij}\|^2+2\|\omega^{-}_{ij}\|^2$, and (c)
follows by subtraction; (d) is immediate from (c).
\end{proof}

What an agent gains from its own level, and what two neighbours gain \emph{reciprocally}, admits a common objective;
what one agent gains exactly at the other's expense does not. The proposition is not new in substance
(\cite[Sec.~4.3]{candogan2011} for two players; \cite{cheng2016,arditti2020sep} on graphs; \cite{gilles2026,sandholm2010}
for (a) and (d)); we include the short proof because the constants in~(c) are used throughout.

\subsection{Priority Games: Every Interaction is Antisymmetric}\label{sec:prio-energies}
In a priority game, the interaction between two conflicting agents is strictly competitive: exactly one of them is
served first, so whatever one gains the other loses. For the linear payoff, which makes the game pairwise,
Proposition~\ref{prop:pairwise}(d) therefore shows that the game is never an exact potential game once there is a
conflict. The result in Proposition~\ref{prop:energies} makes this quantitative: both energies are determined by the
degrees of the conflict graph.

\begin{proposition}[Energies of priority coordination]\label{prop:energies}
For the linear rank-order game $\calG(C,K)$, every interaction is purely antisymmetric, $\omega^{+}_{ij}\equiv0$, and, up to an additive constant, the Hodge potential is $\phi^\star(s)=K^{-1}\sum_i d_i s_i$. With $c_K:=K^{N-1}(K^2-1)/12$,
\begin{equation}\label{eq:main}
\Epot=c_K\sum_i d_i^2,\qquad \Eharm=c_K\sum_i d_i=2c_K|E|.
\end{equation}
Consequently, $\rho_\pot=\sum_i d_i^2/\sum_i d_i(d_i+1)$, independently of $K$, while $\rho_\harm\ge 1/(d_{\max}+1)$ on every non-empty graph of maximum degree $d_{\max}$. Thus the energy split is controlled by the degree structure rather than explicitly by $K$ or $N$. In particular, $\rho_\pot=d/(d+1)$ on a $d$-regular graph and $\rho_\pot=(N-1)/N$ on the complete graph. All these conclusions hold for every protocol, i.e.\ every orientation of the conflict edges, acyclic or not.
\end{proposition}

\begin{proof}
For any protocol, $\psi_{ij}(x,y)+\psi_{ji}(y,x)\equiv-1$, since exactly one of two conflicting agents is served first. Double centring removes the constant, so $\omega_{ij}(x,y)+\omega_{ji}(y,x)=0$ and hence $\omega^+_{ij}=0$.
Next, $\E_y\psi_{ij}(x,y)=-(K-1-x)/K$ if the protocol favours $i$ on the edge and $-(K-x)/K$ otherwise. Both equal $(x-\bar s)/K$ up to a constant, where $\bar s=(K-1)/2$. Hence $\sa_i(x)=d_i(x-\bar s)/K$, so Proposition~\ref{prop:pairwise}(a) gives $\phi^\star(s)=K^{-1}\sum_i d_i s_i$ up to an additive constant, and $\|\sa_i\|^2=d_i^2(K^2-1)/(12K^2)$.
For either orientation, $\Var\psi_{ij}=(K^2-1)/(4K^2)$ and $\|a_{ij}\|^2=\|b_{ij}\|^2=(K^2-1)/(12K^2)$. By \eqref{eq:anova} and orthogonality, $\|\omega_{ij}\|^2=(K^2-1)/(12K^2)$; since $\omega^+_{ij}=0$, also $\|\omega^-_{ij}\|^2=\|\omega_{ij}\|^2$. Substituting into \eqref{eq:polyenergies} gives \eqref{eq:main}.
Finally, $\rho_\pot=\sum_i d_i^2/\sum_i d_i(d_i+1)$ and $\rho_\harm=\sum_i d_i/\sum_i d_i(d_i+1)$. Since $d_i+1\le d_{\max}+1$, $\rho_\harm\ge1/(d_{\max}+1)$. The $d$-regular and complete-graph formulas follow by setting $d_i=d$ and $d_i=N-1$, respectively.
\end{proof}

\begin{remark}[Reading the two energies]\label{rem:quantum}
Let $E_2$ be the set of pairs of conflicts that share an agent, so that $|E_2|=\sum_i\binom{d_i}2$. Since
$\sum_id_i^2=2|E|+2|E_2|$, in units of $2c_K$,
\begin{equation}\label{eq:quanta}
\Eharm=|E|,\;\; \Epot=|E|+|E_2|,\;\; \rho_\pot=\frac{|E|+|E_2|}{2|E|+|E_2|}.
\end{equation}
Every conflict carries one unit of harmonic energy, whatever the rest of the graph; potential energy is generated, in
addition, by \emph{adjacent} conflicts, because what is potential is the shared incentive to escalate, which is
an own effect, and not the interaction. Removing the conflict $\{i,j\}$ eliminates one unit of harmonic energy and $d_i+d_j-1$
units of potential energy. For the intersection of Example~\ref{ex:intersection}, $|E|=16$ and $|E_2|=48$,
so $\Eharm:\Epot=16:64$ and one fifth of the incentive energy is harmonic.
\end{remark}

\emph{The error of a potential surrogate.} Let $\Phi:S\to\R$ be a common objective that replaces the agents
\cite{liu2023,liu2025}. Under log-linear learning with rationality $\beta>0$, a revising agent $i$ at
profile $s$ chooses between $s_i$ and $s'_i$ with log-odds $\beta(u_i(s')-u_i(s))=\beta u(e)$, where $e$ is the
corresponding edge of the game graph; the common-objective model predicts $\beta\,\delta\Phi(e)$ instead.

\begin{corollary}[Best common-objective model of choice odds]\label{cor:surrogate}
For every game, every $\beta>0$ and every $\Phi:S\to\R$, the squared prediction error of the log-odds, summed over
all $NK^N(K-1)/2$ edges of the game graph, satisfies
\begin{equation}\label{eq:surrogate}
\sum_e\beta^2\big(u(e)-\delta\Phi(e)\big)^2\ \ge\ \beta^2\Eharm ,
\end{equation}
with equality iff $\Phi-\phi^\star$ is constant. Hence, for every $\beta>0$, the relative squared error of the
best common-objective model is $\rho_\harm$. For the linear game $\calG(C,K)$ it is at least $1/(d_{\max}+1)$, on every
conflict graph with at least one conflict and for every deterministic protocol.
\end{corollary}
\begin{proof}
By \eqref{eq:hodge}, $u-\delta\Phi=h+\delta(\phi^\star-\Phi)$ with $h\perp\operatorname{im}\delta$, so
$\|u-\delta\Phi\|^2=\Eharm+\|\delta(\phi^\star-\Phi)\|^2$. Dividing by $\sum_eu(e)^2=\Epot+\Eharm$ gives the relative
error; the bound for $\calG(C,K)$ is Proposition~\ref{prop:energies}.
\end{proof}
The error weights every unilateral move equally, as low-rationality learning does to first order
(Proposition~\ref{prop:learning}(a)); at high rationality play concentrates near the equilibria of
Proposition~\ref{prop:ne}, and an error weighted by the visited profiles can be smaller. The corollary measures what a
surrogate misses about the incentives, not how often it mispredicts along a trajectory.

\emph{Pricing and the limits of the two energies.} For the linear game, a price on the levels changes the
potential energy but not the harmonic one. This follows from a general property of pairwise games described in Corollary~\ref{cor:rigid}.

\begin{corollary}[Rigidity of the harmonic energy]\label{cor:rigid}
$\Eharm$ is unchanged by any modification of the own terms $\gamma_i$, the marginals $a_{ij},b_{ij}$, the constants $\overline{\psi}_{ij}$, or the symmetric parts $\omega^+_{ij}$. In particular, no separable pricing $\gamma_i(s_i)$ changes $\Eharm$. If $\Eharm>0$, pricing alone can move $\rho_\pot$ to any value in $[\rho_{\min},1)$, where $\rho_{\min}=\sum_E\|\omega^+_{ij}\|^2/\sum_E(\|\omega^+_{ij}\|^2+\|\omega^-_{ij}\|^2)$.
\end{corollary}

\begin{proof}
By \eqref{eq:polyenergies}, $\Eharm$ depends only on $\omega^-$. Hence changes to $\gamma_i,a_{ij},b_{ij},\overline{\psi}_{ij}$, or $\omega^+_{ij}$ leave it unchanged. Pricing changes only $\sa_i$: choosing $\gamma_i=-\sum_j a_{ij}$ gives $\sa_i\equiv0$ and hence $\rho_\pot=\rho_{\min}$, while scaling any nonconstant pricing makes $\sum_i\|\sa_i\|^2\to\infty$ and thus $\rho_\pot\to1$. Since $\rho_\pot$ varies continuously with the scaling parameter, every value in $[\rho_{\min},1)$ is attained when $\Eharm>0$.
\end{proof}

Corollary~\ref{cor:rigid} reveals the reason why we treat the \emph{pair} $(\Epot,\Eharm)$, rather than the ratio, as the fundamental
quantity: $\rho_\pot$ is comparable across designs only at a fixed normalisation of the own effects. It also
separates pricing from transfers. A transfer is a payment $\tau_i:S\to\R$ added to each payoff, giving the game
with payoffs $u_i+\tau_i$. Transfers that make this game an exact potential game must cancel the
harmonic flow, so the energy of their flow is at least $\Eharm$; for a pairwise game this minimum is attained by the
budget-balanced side payment $\tau_i=-\sum_{j\in\calN(i)}\omega^{-}_{ij}(s_i,s_j)$ (Proposition~\ref{prop:pairwise}(b)),
which no pricing of the levels can achieve. Pricing, in turn, can change the behaviour
drastically while leaving $\Eharm$ fixed: with payoffs $-\pos_C(i,s)-c\,s_i$ under the index protocol, an arbitrarily
small price $c\in(0,1)$ destroys every pure-strategy Nash equilibrium on any graph with a conflict
(Proposition~\ref{prop:priced-SM} in Appendix~\ref{sec:priced} of the SM). Neither energy
predicts the existence of pure-strategy Nash equilibria, and we make no such claim. Nor
do they measure efficiency: by \eqref{eq:cons}, the total number of waits is constant across profiles.

\subsection{Reading a Single Move: Mean and Fluctuation}\label{sec:reading}
Proposition~\ref{prop:energies} gives the energies of the whole game. The results in Proposition~\ref{prop:hflow} explain the decomposition for
a single move: when an agent raises its level, it identifies which part of the gain is potential and which
part is harmonic.

\begin{proposition}[Mean and fluctuation]\label{prop:hflow}
For $\calG(C,K)$, along an edge $(i,s\to s')$ with jump $\ell=s'_i-s_i\ge1$, a neighbour $j$ is \emph{overtaken} iff
\begin{equation}\label{eq:overtake}
s_i<s_j<s'_i,\ \text{or}\ (s_j=s_i,\ j\to i),\ \text{or}\ (s_j=s'_i,\ i\to j), \!\!
\end{equation}
and the game flow is the number $W$ of overtaken neighbours. Then $h(i,s\to s')=W-d_i\ell/K$: the potential part of a
move is the \emph{expected} number of neighbours overtaken under uniformly random levels, and the harmonic part is the
deviation from it.
\end{proposition}

\begin{proof}
By definition of the rank-order payoff, the flow along $(i,s\to s')$ equals the number $W$ of neighbours satisfying \eqref{eq:overtake}, so $u=W$. By Proposition~\ref{prop:energies}, $\delta\phi^\star=d_i\ell/K$, hence $h=u-\delta\phi^\star=W-d_i\ell/K$.
For each neighbour $j$, exactly $\ell$ of the $K$ possible values of $s_j$ satisfy \eqref{eq:overtake}: the $\ell-1$ interior levels $s_i<s_j<s_i'$, together with exactly one endpoint selected by the orientation of $\{i,j\}$. Thus, under a uniform level for $s_j$, the overtaking probability is $\ell/K$. By linearity of expectation, $\E W=d_i\ell/K$, proving the claim.
\end{proof}

Proposition~\ref{prop:hflow} identifies the potential part of a move with the mean of $W$. In fact, the whole
distribution of $W$ is explicit. Given the move, the $d_i$ overtaking events are independent under the uniform
reference measure, each with probability $\ell/K$, so $W\sim\Bin(d_i,\ell/K)$. The potential part of the move is the
mean $d_i\ell/K$, and the harmonic part is the fluctuation around it, with variance
$d_i\frac{\ell}{K}\big(1-\frac{\ell}{K}\big)$. The squared mean grows like $d_i^2$ and the variance like $d_i$, which
is the scaling $\Epot\propto\sum_id_i^2$ and $\Eharm\propto\sum_id_i$ of Proposition~\ref{prop:energies}. The
constants also match. Summing over all moves weights each jump length $\ell$ by the number $K-\ell$ of level pairs 
at distance $\ell$, and
\[
\sum_{\ell=1}^{K-1}(K-\ell)\Big(\frac{\ell}{K}\Big)^2=\sum_{\ell=1}^{K-1}(K-\ell)\,\frac{\ell}{K}\Big(1-\frac{\ell}{K}\Big)=\frac{K^2-1}{12}.
\]
The two sums are not equal term by term; rather, the substitution $\ell\mapsto K-\ell$ maps one onto the other. This is why
the potential share of the linear game does not depend on $K$.

Thus the Hodge potential collects what an agent can expect from its own level against uniformly random
opponents, and the harmonic flow is what depends on who the opponents actually are (cf.\ the tournament reading in
\cite{strang2022}); its dependence on ties is computed in Theorem~\ref{thm:pattern-SM} of the SM. For the linear payoff, if ties are broken uniformly at random and payoffs are taken in expectation, $\Epot$ is
unchanged and $\Eharm$ is multiplied by $(K-2)/(K+1)$ (Proposition~\ref{prop:randtie-SM}). For $K\ge3$ a
deterministic protocol therefore amplifies the harmonic component without creating it, whereas for $K=2$ randomised
tie-breaking yields an exact potential game.

%======================================================================
\section{Dynamics and Learning}\label{sec:dyn}
This section addresses the second problem posed in Section~\ref{subsec:problem}. It first shows that a nonzero harmonic component need not induce cycles in improvement dynamics, and then identifies what the harmonic component does affect: the stationary behaviour of log-linear learning.  In Propositions~\ref{prop:fip} and~\ref{prop:ne}, $C$ is an arbitrary
conflict graph and $f$ an arbitrary strictly decreasing payoff-from-rank function.

\begin{proposition}[Finite improvement]\label{prop:fip}
For every conflict graph $C$, every protocol and every strictly decreasing $f$, both $V(s)=\sum_is_i$ and the Hodge
potential of the linear game, $\phi^\star(s)=\frac1K\sum_id_is_i$, are generalised ordinal potentials of
$\calG_f(C,K)$. Every improvement path has length at most $N(K-1)$.
\end{proposition}

\begin{proof}
Fix $s_{-i}$. If $s'_i>s_i$ and $j$ is served before $i$ at $s'$, then $s_j\ge s'_i>s_i$, so $j$ is also served before $i$ at $s$. Hence $\pos_C(i,s)$ is non-increasing in $s_i$, and since $f$ is strictly decreasing, every strict improvement must have $s'_i>s_i$. The mover cannot be isolated, so $d_i>0$; therefore both $V(s)=\sum_i s_i$ and $\phi^\star(s)=K^{-1}\sum_i d_i s_i$ strictly increase. Since each improvement raises the integer-valued $V$ by at least one and $0\le V\le N(K-1)$, every improvement path has length at most $N(K-1)$.
\end{proof}

Proposition~\ref{prop:fip} concerns \emph{strictly} improving moves. If a revising agent may switch between equally good levels,
play can cycle: for $N=K=2$ with agent~$0$ favoured, $(0,0)\to(0,1)\to(1,1)\to(1,0)\to(0,0)$ consists of best
responses, the last two being indifferent. Hence best-response dynamics converge with inertia (an agent
moves only when it strictly gains) and not without it. Thus $\phi^\star$ is a Lyapunov function for asynchronous
strict better-response dynamics of every $\calG_f(C,K)$, although $\calG(C,K)$ is not an exact potential game. Since every improving move $(i,s\to s')$ raises the level of $i$ by some $\ell\ge1$, $\delta\phi^\star=d_i\ell/K>0$ along it: the harmonic flow changes the size of the gain but cannot create an improving move against the potential.

\begin{proposition}[Pure-strategy Nash equilibria]\label{prop:ne}
Under the index protocol, let $\calB(s)=\{i:s_i<K-1\}$ be the set of agents below the top level. A profile $s$ is a pure-strategy Nash equilibrium of $\calG_f(C,K)$ iff every $i\in\calB(s)$ has a larger index than each of its
neighbours; in particular $\calB(s)$ is then an independent set of $C$. Consequently, in every pure-strategy Nash equilibrium and on
every conflict edge, the agent favoured by the protocol is served first: all pure-strategy Nash equilibria induce the
same service order. Mixed equilibria are not considered.
\end{proposition}

\begin{proof}
By monotonicity of $\pos_C(i,\cdot)$ in $s_i$, $i$'s best deviation is to $K-1$, after which its predecessors are exactly $\{j\in\calN(i):s_j=K-1,\ j<i\}$. Hence $s$ is a Nash equilibrium iff every $i\in\calB(s)$ has no neighbour in $\calB(s)$ and no neighbour $j>i$, i.e.\ iff every such $i$ has larger index than all its neighbours. Thus $\calB(s)$ is independent.
Consequently every conflict edge has at least one endpoint at level $K-1$; if exactly one does, it is the protocol-favoured endpoint, while if both do, the index protocol itself determines the order. Hence all pure-strategy Nash equilibria induce the same service order.
\end{proof}

\subsection{Log-linear Learning: Stationary Law and Current}\label{sec:learning}
Best response ignores magnitudes. Log-linear learning is sensitive to them, and at low rationality it gives both
energies a dynamical meaning to leading order. At each step an agent $i$ is selected with probability $\alpha_i>0$, where $\sum_i\alpha_i=1$, and resamples
its level from $\Pr(s_i=x)\propto\exp(\beta u_i(x,s_{-i}))$. This is an irreducible Markov chain $P_\beta$ on $S$
with stationary law $\pi_\beta$; for a potential game with potential $\phi$ it is reversible and
$\pi_\beta\propto e^{\beta\phi}$ \cite{blume1993}. For a general game, let $\phi_\alpha$ and $h_\alpha:=u-\delta\phi_\alpha$
be the Hodge decomposition in the weighted norm $\|x\|_\alpha^2=\sum_e\alpha_{i(e)}x(e)^2$, where $i(e)$ is the agent
moving along $e$; for uniform selection, $\alpha_i=1/N$, these are $\phi^\star$ and $h$. Write
$\mathcal J_\beta(e)=\pi_\beta(s)P_\beta(s,s')-\pi_\beta(s')P_\beta(s',s)$ for the stationary probability current along
$e=(i,s\to s')$, and $\sigma_\beta=\sum_e\mathcal J_\beta(e)\log\frac{\pi_\beta(s)P_\beta(s,s')}{\pi_\beta(s')P_\beta(s',s)}\ge0$
for the entropy production rate, which vanishes iff the chain is reversible.

\begin{proposition}[Low-rationality learning]\label{prop:learning}
Let $N\ge2$ agents each have the $K\ge2$ actions $\{0,\dots,K-1\}$, let $(u_i)_{i\in[N]}$ be arbitrary payoffs on
$S=\{0,\dots,K-1\}^N$, and let log-linear learning be run with fixed selection probabilities $\alpha_i>0$,
$\sum_i\alpha_i=1$. Then, with $N$, $K$, $(u_i)$ and $(\alpha_i)$ fixed and $\beta\to0$,
\begin{enumerate}
\item[(a)] $\pi_\beta\propto e^{\beta\phi_\alpha}\,\big(1+O(\beta^2)\big)$: the occupation measure is Gibbs in the
Hodge potential;
\item[(b)] $\mathcal J_\beta(e)=\beta\,\alpha_i\,h_\alpha(e)/K^{N+1}+O(\beta^2)$ on every edge $e=(i,s\to s')$: the probability current
is the harmonic flow;
\item[(c)] $\sigma_\beta=\beta^2\,\|h_\alpha\|_\alpha^2/K^{N+1}+O(\beta^3)$.
\end{enumerate}
Under uniform selection, $\pi_\beta\propto e^{\beta\phi^\star}(1+O(\beta^2))$, $\mathcal J_\beta \!=\! \beta h/(NK^{N+1}) \!+ O(\beta^2)$ and
$\sigma_\beta \!\!=\!\! \beta^2 \Eharm/(NK^{N+1}) \!+ O(\beta^3)$.
\end{proposition}

\begin{proof}
Since $\alpha_i>0$ and every logit choice has positive probability, the finite chain is irreducible for all $\beta$, so its stationary distribution is unique and smooth near $\beta=0$. Expanding
$P_\beta=P_0+\beta P_1+O(\beta^2)$ gives $P_0=\sum_i\alpha_iA_i$ and, for a transition changing coordinate $i$,
$P_1(s,s')=\frac{\alpha_i}{K}(u_i(s')-(A_i u_i)(s))$. Since $P_0$ is symmetric and doubly stochastic, $\pi_0=K^{-N}\one$ and $\pi_\beta=\pi_0+\beta\pi_1+O(\beta^2)$ with $\pi_1\perp\one$.

(a) The $O(\beta)$ stationarity equation is $\pi_1(I-P_0)=\pi_0P_1$. Here $I-P_0=\numop_\alpha:=\sum_i\alpha_i(I-A_i)$ and $\pi_0P_1=K^{-N}g_\alpha$, where $g_\alpha:=\sum_i\alpha_i(u_i-A_i u_i)$. The weighted Hodge normal equation is $\numop_\alpha\phi_\alpha=g_\alpha$; since $\numop_\alpha$ is invertible on $\one^\perp$,
$\pi_1=K^{-N}(\phi_\alpha-\bar\phi_\alpha)$. This is exactly the first-order expansion of the normalized Gibbs measure proportional to $e^{\beta\phi_\alpha}$, proving (a).

(b) Let $\mathcal J_\beta=\mathcal J_0+\beta\mathcal J_1+O(\beta^2)$. Since $P_0$ is symmetric, $\mathcal J_0=0$. For $e=(i,s\to s')$, using $P_0(s,s')=\alpha_i/K$ and $(A_i u_i)(s)=(A_i u_i)(s')$,
\[
\mathcal J_1(e)
\!\!=\!\!\frac{\alpha_i}{K^{N+1}}
\big[(u_i(s')-u_i(s))-(\phi_\alpha(s')-\phi_\alpha(s))\big]
\!\!=\!\!\frac{\alpha_i h_\alpha(e)}{K^{N+1}}.
\]

(c) Let $q_e=\pi_\beta(s)P_\beta(s,s')$ and $q'_e=\pi_\beta(s')P_\beta(s',s)$. Then $q_e-q'_e=\mathcal J_\beta(e)=O(\beta)$ and $q_e,q'_e=\alpha_i/K^{N+1}+O(\beta)$, so
$\log(q_e/q'_e)=K^{N+1}\mathcal J_\beta(e)/\alpha_i+O(\beta^2)$. Hence
\[
\mathcal J_\beta(e)\log\frac{q_e}{q'_e}
=\frac{\beta^2\alpha_i h_\alpha(e)^2}{K^{N+1}}+O(\beta^3).
\]
Summing over one orientation of each edge yields
$\sigma_\beta=\beta^2\|h_\alpha\|_\alpha^2/K^{N+1}+O(\beta^3)$.

For uniform selection, $\alpha_i=1/N$, so $\numop_\alpha=\numop/N$ and $g_\alpha=g/N$, hence $\phi_\alpha=\phi^\star$ and $h_\alpha=h$. Also $\|h_\alpha\|_\alpha^2=\Eharm/N$, giving the stated formulas.
\end{proof}

The harmonic component is thus invisible in the occupation measure to first order and appears as the current it
sustains; its energy determines, to leading order, the entropy production rate. For the linear priority game under
uniform revision, Proposition~\ref{prop:energies} gives $\sigma_\beta=\beta^2(K^2-1)|E|/(6NK^2)+O(\beta^3)$: every
conflict contributes the same amount to the leading-order entropy production, whatever the rest of the graph. With the
$L^2$ projection the first-order discrepancy from a Gibbs law vanishes identically, whereas near-potential games are
only guaranteed a stationary law within $O(\beta\delta)$ of that of a nearby potential game \cite{candogan2013geb}.
The selection probabilities $\alpha_i$ enter the stationary law at first order, through $\phi_\alpha$, only when
$\Eharm>0$.

\begin{remark}[Range of validity]\label{rem:regimes}
Proposition~\ref{prop:learning} describes the noise-dominated regime and asserts only the orders of the remainders.
For the intersection, the law~(c) is accurate to about $10\%$ for $\beta\max_id_i\lesssim1$
(Fig.~\ref{fig:learning}(d)), an observation for one example, not a uniform bound. As $\beta$ grows, logit responses
approach best responses, which see only the ordinal structure of Propositions~\ref{prop:fip}--\ref{prop:ne}; we do not
attempt a large-$\beta$ theorem.
\end{remark}

%======================================================================
\section{General Payoffs: The Hodge Spectrum}\label{sec:spectrum}
Sections~\ref{sec:pairwise}--\ref{sec:dyn} concern the linear payoff. This section and the next address the
third problem described in Section~\ref{subsec:problem}: whether the harmonic share can be reduced by choosing a different
payoff-from-rank function. We work on the complete conflict graph $C=K_N$, where agents compete with every other, and, for finitely many
levels, with the index protocol (every total order is equivalent after relabelling). This section
treats continuous priorities, where the answer can be computed completely; Section~\ref{sec:finiteK} returns to finitely
many levels.

The first step is to turn the design question into a matrix problem. A payoff-from-rank function is determined, up to
an irrelevant constant, by its increments $F_j=f(j-1)-f(j)$, $j=1,\dots,N-1$, i.e.\ by how much an agent loses when it
drops by one rank. The game flow is linear in $F$: along an edge $e$, $u_f(e)=\sum_jF_jI_j^{(e)}$ with
$I^{(e)}_j=\one\{\pos'<j\le\pos\}$, where $\pos$ and $\pos'$ are the positions of the moving agent at the two ends of
$e$ (a move upwards can only lower the position). Both energies are therefore quadratic forms in $F$. With the Gram
matrices $[M_{\mathrm{tot}}]_{jk}=\langle I_j,I_k\rangle$ and $[M_\pot]_{jk}=\langle PI_j,PI_k\rangle$, where $P$ is the
orthogonal projection onto $\operatorname{im}\delta$, the potential share is a Rayleigh quotient,
\begin{equation}\label{eq:rayleigh}
\rho_\pot(\calG_f)=\frac{F^\top M_\pot F}{F^\top M_{\mathrm{tot}}F},\qquad
Q_K:=\tfrac{N-1}{N}M_{\mathrm{tot}}-M_\pot .
\end{equation}

The bound $\rho_\pot\le(N-1)/N$ for all $F$ is thus equivalent to $Q_K\succeq0$, and the equality cases are the kernel
of $Q_K$. The linear payoff corresponds to $F=\one$.

\begin{lemma}[Exact eigenvector]\label{lem:B}
For every $N\ge2$ and $K\ge2$, $Q_K\one=0$: the Rayleigh quotient \eqref{eq:rayleigh} is stationary at the linear
payoff, with value $(N-1)/N$.
\end{lemma}

The proof, given in Appendix~\ref{app:B-kernel}, uses only \eqref{eq:cons} and the conditional
independence of opponents given $s_i$, and the dependence on $K$ cancels. To compute the whole spectrum, we pass to continuous priorities, where ties have probability zero.
Each agent chooses $s_i\in[0,1]$, the reference measure is the uniform law on $[0,1]^N$, and
$u_i(s)=f(\pos_i(s))$ with $\pos_i(s)=\#\{j\ne i:s_j>s_i\}$. A flow assigns to every agent $i$, every pair of
levels $t,t'$ and every $s_{-i}$ a value $x_i(t\to t';s_{-i})$ that is antisymmetric in $(t,t')$. The game flow is
$u_i(t',s_{-i})-u_i(t,s_{-i})$, and the gradient of $\phi:[0,1]^N\to\R$ is $\phi(t',s_{-i})-\phi(t,s_{-i})$. Flows carry
the norm
\begin{equation}\label{eq:contnorm}
\|x\|^2=\sum_i\E_{s_{-i}}\Big[\tfrac12\iint_{[0,1]^2}x_i(t\to t';s_{-i})^2\,\mathrm{d}t\,\mathrm{d}t'\Big],
\end{equation}
the continuous analogue, up to normalisation, of the counting norm of Section~\ref{sec:setup}, and $P$ denotes the orthogonal projection
onto the closure of the space of gradients. The potential and the total energy are $\|Pu\|^2$ and $\|u\|^2$, and
$\rho_\pot=\|Pu\|^2/\|u\|^2$. With the continuum form of Lemma~\ref{lem:numop} (Appendix~\ref{app:A-spec}, Step~1), the two
energies become
\begin{equation}\label{eq:contquad}
\begin{aligned}
\|Pu\|^2&=\sum_{m=1}^{N-1}\frac{(N-m)^2}{m}J_m(f),\\
\|u\|^2&=N\sum_{m=1}^{N-1}(N-m)J_m(f),
\end{aligned}
\end{equation}
where $J_m(f)=J_m(f,f)$ is the order-$m$ form
$J_m(f,g)=\binom Nm\iint\kappa(q,q')^m d_m[f](q)d_m[g](q')\,\mathrm{d}q\,\mathrm{d}q'$, with $\kappa(q,q')=\min(q,q')-qq'$ the covariance
of the Brownian bridge and $d_m[f](q)=\E[\Delta^mf(\Bin(N-1-m,q))]$; each $J_m$ is positive semidefinite.
Intuitively, $J_m(f)$ measures how much of the incentive structure of the payoff is carried by interactions that
involve $m$ agents at once. Comparing the equations in \eqref{eq:contquad}, the order $m$ contributes to the potential
energy with a relative weight that decreases with $m$, so interactions among many agents are relatively more
harmonic.

\begin{theorem}[Hodge spectrum]\label{thm:spectrum}
Let the discrete Chebyshev polynomials be $t_p(r)=\sum_{k=0}^p(-1)^k\binom pk\binom{p+k}k\frac{r^{\underline k}}{(N-1)^{\underline k}}$, $p=0,\dots,N-1$, i.e.\ the orthogonal polynomials of the uniform distribution on
$\{0,\dots,N-1\}$.
\begin{enumerate}[
    label=(\alph*),
    leftmargin=0pt,
    itemindent=1.5em,
    labelsep=0.5em
]
\item[(a)] \emph{Every order form is diagonal.} For all $m$ and all $p\ne p'$, $J_m(t_p,t_{p'})=0$, and
\begin{equation}\label{eq:Jdiag}
J_m(t_p,t_p)=\frac{\binom Nm\,m\,(p+m)!}{p(p+1)(2p+1)(p-m)!\,[(N-1)^{\underline m}]^2}
\end{equation}
for $m\le p$, while $J_m(t_p,t_p)=0$ for $m>p$. Consequently $t_1,\dots,t_{N-1}$ diagonalise $\sum_mw_mJ_m$ for
\emph{every} choice of weights, in particular both forms in \eqref{eq:contquad}.
\item[(b)] \emph{Eigenvalues.} The generalised eigenvalues of \eqref{eq:contquad} are
\begin{equation}\label{eq:lambda}
\lambda_p=\frac{(N-p)(N+p+1)}{N\,p(p+1)}\bigg[1-\frac{\binom Np}{\binom{N+p+1}p}\bigg],
\end{equation}
for $p=1,\dots,N-1$. In particular $\lambda_1=(N-1)/N$ and $\lambda_{p+1}<\lambda_p$.
\end{enumerate}
\end{theorem}

The proof is given in Appendix~\ref{app:A-spec}; we explain here why the theorem holds.
Each order form $J_m$ pairs the smoothed differences $d_m[f]$ with the kernel $\kappa^m$, and two classical facts about
orthogonal polynomials make the two compatible. First, averaging $t_p$ over a binomial law turns it into the Legendre
polynomial of degree $p$, and taking $m$ differences then turns it into a Jacobi polynomial of degree $p-m$
\cite{koekoek2010}. Second, the integral operator with kernel $\kappa^m$ maps each of these Jacobi polynomials to a
multiple of itself times $[q(1-q)]^m$, and $[q(1-q)]^m$ is exactly the weight function for which these Jacobi
polynomials are orthogonal. Hence, $t_p$ and $t_{p'}$ are orthogonal in every order form, which is part~(a), and the
constant in \eqref{eq:Jdiag} is the norm of a Jacobi polynomial. Part~(b) follows by inserting \eqref{eq:Jdiag} into
\eqref{eq:contquad}, where the resulting finite sums can be evaluated in closed form.
A payoff of degree $p$ has Hoeffding components of order up to $p$, which is why $\lambda_p$ decreases. For fixed
$p$ and $N\to\infty$, $1-\lambda_p=\binom{p+1}2/N+O(N^{-2})$.

\begin{corollary}[Continuum optimality]\label{cor:Qinf}
For continuous priorities, $\rho_\pot(f)\le(N-1)/N$ for every nonconstant $f$, with equality iff $f$ is affine.
\end{corollary}

\begin{proof}
For nonconstant $f$, $\|u\|^2>0$. By \eqref{eq:contquad},
\[
\frac{N-1}{N}\|u\|^2-\|Pu\|^2
=\sum_{m\ge1}(N-m)J_m(f)\frac{N(m-1)}{m}\ge0,
\]
since $J_m(f)\ge0$. Hence $\rho_\pot(f)=\|Pu\|^2/\|u\|^2\le(N-1)/N$. Equality holds iff $J_m(f)=0$ for every $m\ge2$. By the representation of $J_m$ in terms of $d_m[f]$, this is equivalent to $d_m[f]\equiv0$ for all $m\ge2$, hence to $f$ being affine. Conversely, affine $f$ has $d_m[f]\equiv0$ for $m\ge2$, so equality holds.
\end{proof}

\begin{corollary}[Loss of potential share for non-affine payoffs]\label{cor:stability}
For continuous priorities and nonconstant $f$, let $R_f^2$ be the squared correlation between $f(r)$ and $r$ under the uniform law on ranks, i.e.\ the fraction of $\Var_r f$ explained by linear regression on $r$. Then
\begin{equation}\label{eq:stability}
\tfrac{2}{N+2}(1-R_f^2)
\le \lambda_1-\rho_\pot(f)
\le(\lambda_1-\lambda_{N-1})(1-R_f^2),
\end{equation}
with equality on the left iff $f$ is a polynomial of degree at most two.
\end{corollary}

\begin{proof}
By Theorem~\ref{thm:spectrum}, $\rho_\pot(f)=\sum_{p=1}^{N-1}\lambda_p w_p$, where $w_p\ge0$ is the variance share of the $t_p$ component and $\sum_p w_p=1$. Since $t_1$ spans the centred linear functions of rank, $w_1=R_f^2$. Hence
\[
\lambda_1-\rho_\pot(f)=\sum_{p\ge2}(\lambda_1-\lambda_p)w_p,\qquad \sum_{p\ge2}w_p=1-R_f^2.
\]
As $\lambda_p$ is strictly decreasing, the bounds follow from
$\lambda_1-\lambda_2=2/(N+2)$ and
$\lambda_1-\lambda_p\le\lambda_1-\lambda_{N-1}$. Equality in the lower bound holds iff $w_p=0$ for all $p\ge3$, i.e.\ iff $f\in\operatorname{span}\{t_0,t_1,t_2\}$, equivalently $f$ has degree at most two. For $N=2$, every $f$ on the two ranks is affine and the claim is immediate.
\end{proof}

Thus, the loss of potential share relative to the linear payoff is bounded on both sides by the payoff's deviation from linearity. For example, when $N=8$, the potential share decreases from $0.875$ for the linear payoff to $0.860$ for $f(r)=-r^2$, and further to $0.647$ for winner-take-all.

%======================================================================
\section{Finitely Many Priority Levels}\label{sec:finiteK}
Corollary~\ref{cor:Qinf} shows that, for continuous priorities, no payoff-from-rank function has a smaller harmonic
share than the affine one. In applications the number of priority levels is finite and often small, e.g.\ a few levels of urgency. Ties then occur with positive probability, and a deterministic protocol
resolves them in favour of some agents, which breaks the exchangeability that the continuum argument rests on. The
bound nevertheless remains valid for every $K\ge3$, and it can fail for $K=2$.

\begin{theorem}[Finite quantisation]\label{thm:finiteK}
For every $K\ge3$ and every $N\ge2$, $Q_K\succeq0$ with $\ker Q_K=\operatorname{span}(\one)$. Equivalently, among all
nonconstant payoff-from-rank functions on $K_N$ under a total-order protocol with $K\ge3$ priority levels, $\rho_\pot(f)\le(N-1)/N$ with equality iff $f$ is
affine. The hypothesis $K\ge3$ cannot be dropped: for $K=2$ and $N=7$ the strictly decreasing payoff with
increments $F=(9,8,6,4,2,1)$ has $F^\top Q_2F=-7237/480<0$, i.e.\ $\rho_\pot=53454277/62354880>6/7$.
\end{theorem}

The proof is provided in Appendix~\ref{app:B}, with the lengthy computations in the SM. We explain here its
structure and where the hypothesis $K\ge3$ enters. Lemma~\ref{lem:B} gives the kernel; positivity is established in
four stages, each carried out in the corresponding part of Appendix~\ref{app:B}.
\begin{enumerate}[label=\arabic*),leftmargin=*]
\item \emph{Splitting by interaction order} (Appendix~\ref{app:B-orders}). Decomposing the payoffs into
their Hoeffding components splits $Q_K$ into a sum of forms, one for each interaction order $m$, i.e.\ for each number
of agents whose levels jointly determine a component. Every form of order $m\ge3$ is positive semidefinite for every
payoff (Proposition~\ref{prop:A-orders}). The reason is a sign structure: of two conflicting agents exactly one is served
first, so the corresponding indicators are complementary and their cross terms are nonpositive. It remains to treat
the one- and two-agent terms.
\item \emph{Reduction to bond variables} (Appendix~\ref{app:B-bond}). The one- and two-agent terms depend
on the payoff only through $N-1$ vectors in $\R^K$, one for each pair of adjacent ranks, which we call \emph{bonds}. In
these variables the form is a sum of squares minus a single negative term (Proposition~\ref{prop:A-bond}).
\item \emph{A sharp inequality on the levels} (Appendix~\ref{app:B-level}). The negative term is dominated, bond by bond, by
an inequality on functions of the level whose constant is $K-2$ and cannot be improved (Lemma~\ref{lem:level} below). This is
the only place where $K\ge3$ is used.
\item \emph{Closing the argument} (Appendices~\ref{app:B-perron} and~\ref{app:B-assembly}). What remains is a matrix inequality on
the path of bonds, which a Perron-type argument reduces to one scalar condition per bond. This condition is proved
analytically for $K\ge5$, and with exact symbolic computation for $K=3,4$, in each case for all but finitely
many $N$; the finitely many remaining pairs $(K,N)$ are certified by exact rational matrix factorisations. No
floating-point computation enters the proof; the scripts are listed in Appendix~\ref{app:repro} of the SM.
\end{enumerate}
Since stage~3) is where the hypothesis $K\ge3$ enters, we state its inequality. For $D\in\R^K$, viewed as a function of
the level, let $G(a)=\sum_{x\le a}D(x)$ and $H(a)=\sum_{x<a}D(x)$ be its inclusive and exclusive cumulative sums, let
$\Var(g)=\frac1K\sum_{x=0}^{K-1}\big(g(x)-\bar g\big)^2$, with $\bar g=\frac1K\sum_xg(x)$, be the variance of $g\in\R^K$ under the uniform law on the levels (so $\Var(D)$, $\Var(G)$, $\Var(H)$ are defined), and let
$\calD(g)=\sum_{y=1}^{K-1}\frac{y(K-y)}K\big(g(y)-g(y-1)\big)^2$, a weighted Dirichlet form for which $\calD\ge2\Var$.

\begin{lemma}[Sharp level inequality]\label{lem:level}
For $K\ge3$, $\mu\in[0,1]$, $\nu=1-\mu$ and every $D\in\R^K$,
\begin{equation}\label{eq:level}
\mu\,(\calD-2\Var)(H)+\nu\,(\calD-2\Var)(G)\ \ge\ (K-2)\,\mu\nu\,\Var(D),
\end{equation}
and the constant $K-2$ cannot be improved.
\end{lemma}
The proof (Appendix~\ref{app:B-level}; details in Lemma~\ref{lem:Lstar} of the SM) writes the difference of the two sides as a Schur
product of the Green's function of a path with a positive semidefinite matrix.

\begin{remark}[why $K\ge3$]\label{rem:K2}
For $K=2$ the right-hand side of \eqref{eq:level} vanishes, the negative term of stage~2) is no longer compensated,
and the statement is indeed false: exact computation gives $Q_2\succeq0$ for $N\le6$, while at $N=7$ both the payoff
of Theorem~\ref{thm:finiteK} and its reversal beat the linear one. Numerically the
violation persists for $7\le N\le12$ (Table~\ref{tab:gap-SM} in the SM).
\end{remark}

Proposition~\ref{prop:recip} describes how finite quantisation perturbs the continuum picture. Let a
protocol be given by a total order $\prec$ of the agents, let $\prec^{\mathrm{rev}}$ be the reversed order, and let
$\mathsf m^{\prec}(K)$ denote either of the two energy forms of $\calG_f(C,K)$ divided by $K^{N+1}$.

\begin{proposition}[Reciprocity and expansion]\label{prop:recip}
{\rm(a)} $\mathsf m^{\prec}(K)$ is a polynomial in $1/K$ of degree at most $2N-1$, whose constant term is the
corresponding form for continuous priorities. {\rm(b)} As polynomials in $x=1/K$,
$\mathsf m^{\prec}(-x)=\mathsf m^{\prec^{\mathrm{rev}}}(x)$: changing the sign of $K$ reverses the protocol.
{\rm(c)} If an automorphism of $C$ reverses the protocol on every edge, as $i\mapsto N-1-i$ does on $K_N$, then
$\mathsf m^{\prec}(K)$ is a polynomial in $K^{-2}$ of degree at most $N-1$ vanishing at $K=1$. In particular
\begin{equation}\label{eq:Qexpansion}
K^{-(N+1)}Q_K=Q_\infty+\sum_{j=1}^{N-1}K^{-2j}\,\mathsf C_j,~ Q_\infty+\sum_j\mathsf C_j=0,
\end{equation}
where $Q_\infty:=\lim_{K\to\infty}K^{-(N+1)}Q_K$ is the continuum form of Section~\ref{sec:spectrum}, i.e.\ the form
$\frac{N-1}N\|u\|^2-\|Pu\|^2$ of \eqref{eq:contquad} written in the increments $F$, diagonal in the Chebyshev basis, and 
$\mathsf C_j$ do not depend on $K$.
\end{proposition}
The proof (Section~\ref{sec:recip} of the SM) counts the lattice points that induce a given
order of the coordinates; part~(b) is the reciprocity theorem for $P$-partitions \cite{stanley2012}, since changing the
sign of the dilation exchanges weak and strict inequalities, i.e.\ who wins a tie. Hence finite quantisation perturbs
the two forms of Theorem~\ref{thm:spectrum} by a terminating expansion in $K^{-2}$, and the generalised eigenvalues
equal $\lambda_p+O(K^{-2})$. Without an automorphism that reverses the protocol, odd powers of $1/K$ do occur.

The expansion also locates the failure at $K=2$: extended to real $K$ via \eqref{eq:Qexpansion}, the critical value at
which the smallest nontrivial eigenvalue vanishes is $1.750$, $1.875$, $1.953$, $2.006$ for $N=4,\dots,7$.

Theorem~\ref{thm:finiteK} turns Corollary~\ref{cor:surrogate} into a limit over all payoffs.

\begin{corollary}[Design limit for potential surrogates]\label{cor:logodds}
Under the settings of Theorem~\ref{thm:finiteK}, let $K\ge3$ and $\beta>0$. For every nonconstant $f$ and every
$\Phi:S\to\R$,
\begin{equation}\label{eq:logodds}
\sum_e\beta^2\big(u(e)-\delta\Phi(e)\big)^2\ \ge\ \frac1N\sum_e\beta^2u(e)^2 ,
\end{equation}
with equality iff $f$ is affine and $\Phi-\phi^\star$ is constant.
\end{corollary}
\begin{proof}
Combine \eqref{eq:surrogate} with $\Eharm\ge\frac1N(\Epot+\Eharm)$, which is Theorem~\ref{thm:finiteK}.
\end{proof}

\begin{figure}[t]\centering
\includegraphics[width=\columnwidth]{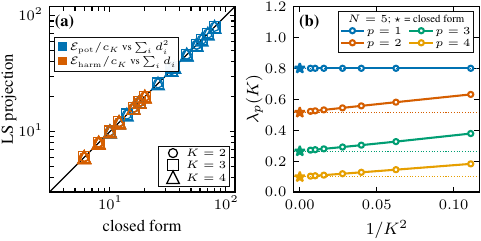}
\vspace{-8mm}
\caption{Closed forms against independent computation. (a)~Energies of Proposition~\ref{prop:energies} against a
least-squares Hodge projection. (b)~Finite-$K$ spectrum against Theorem~\ref{thm:spectrum}.}\label{fig:verify}
\end{figure}

\begin{remark}[Interpretation of Theorem~\ref{thm:finiteK}]\label{rem:meaning}
Within the family $\{\calG_f\}$ there are no own terms to rescale, so the normalisation issue of
Corollary~\ref{cor:rigid} does not arise and the shares of different payoffs are comparable. Theorem~\ref{thm:finiteK}
is therefore a limit on reward design. By Proposition~\ref{prop:learning}, under uniform revision,
$\sigma_\beta\ge\beta^2(\Epot+\Eharm)/(N^2K^{N+1})+O(\beta^3)$: relative to the total incentive energy, no payoff
from rank makes low-rationality learning less irreversible, to leading order, than the linear payoff.
Corollary~\ref{cor:logodds} gives the same limit at every rationality, without an expansion in $\beta$. The limit concerns the \emph{magnitudes} of the incentives, not their
direction, which an ordinal potential captures (Proposition~\ref{prop:fip}).
\end{remark}

\begin{figure*}[t]
\centering
\includegraphics[width=\textwidth]{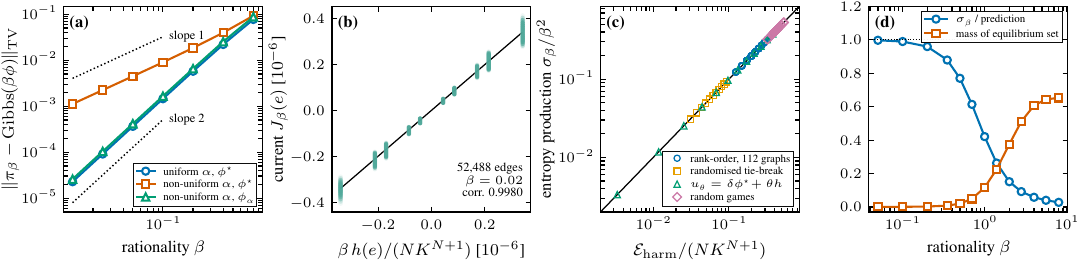}
\vspace{-5mm}
\caption{Log-linear learning (Proposition~\ref{prop:learning}), computed from the transition matrix. (a)~Occupation, on $C_5$ ($K=3$): distance between the stationary law and the Gibbs law of a Hodge potential.
(b)~Circulation: stationary current against the harmonic flow on every edge of the game graph of the intersection.
(c)~Irreversibility: entropy production against harmonic energy for four families of games; the line is the
prediction, without fitted parameters; the random games are not pairwise. (d)~Range of validity on the intersection:
the low-$\beta$ law gives way to the ordinal regime, in which mass shifts towards the set of pure-strategy Nash equilibria of
Proposition~\ref{prop:ne} without being absorbed by it.}\label{fig:learning}
\end{figure*}

\section{Numerical Study}\label{sec:numerics}
The experiments check the closed forms against computations that do not use them and show what the results mean on
concrete systems: the eight-vehicle intersection of Example~\ref{ex:intersection} with $K=3$ ($6561$ profiles,
$52\,488$ edges), families of small conflict graphs, games that are not priority games, and alternative layouts.
All quantities are computed by enumerating the profile space in double precision. The Hodge components are obtained
by a sparse least-squares projection of the flow onto $\operatorname{im}\delta$, which uses neither
Lemma~\ref{lem:numop} nor any closed form; stationary laws, currents and entropy production are computed from the
transition matrix of log-linear learning, not by simulation. None of these computations enters any proof.

\subsection{Verification of the Closed Forms}
Figure~\ref{fig:verify} checks Proposition~\ref{prop:energies} and Theorem~\ref{thm:spectrum}. Panel~(a) shows that the
least-squares energies of $\calG(C,K)$ on $24$ conflict graphs with four to six vertices and $K\in\{2,3,4\}$, divided by
$c_K$, equal $\sum_id_i^2$ and $\sum_id_i$ to within $4\times10^{-15}$, and the markers of one graph coincide across
$K$, as the independence of the potential share from $K$ requires. Panel~(b) shows the generalised eigenvalues of
$(M_\pot,M_{\mathrm{tot}})$ for $N=5$ and $K\le12$: they tend to the values \eqref{eq:lambda} at rate $O(K^{-2})$, as
Proposition~\ref{prop:recip} predicts, and $(N-1)/N$ is an eigenvalue for every $K$ (Lemma~\ref{lem:B}), the largest one in this example.

\begin{table}[t]
\centering
\caption{Four layouts for eight agents, $K=3$. Energies in units of $2c_K$; $\sigma_\beta/\beta^2$ measured at
$\beta\max_id_i=0.02$.}\label{tab:layouts}
\small
\setlength{\tabcolsep}{3pt}
\begin{tabular}{@{}lccccc@{}}\toprule
Layout & $|E|$ & $|E_2|$ & $\Epot$ & $\rho_\pot$ & $\sigma_\beta/\beta^2$\\\midrule
Fully channelised (matching) & $4$ & $0$ & $4$ & $1/2$ & $0.0741$\\
Single merge cycle $C_8$ & $8$ & $8$ & $16$ & $2/3$ & $0.1481$\\
One arbiter (star) & $7$ & $21$ & $28$ & $4/5$ & $0.1296$\\
Unprotected crossing (Example~\ref{ex:intersection}) & $16$ & $48$ & $64$ & $4/5$ & $0.2963$\\\bottomrule
\end{tabular}
\end{table}

\subsection{Log-linear Learning}
Figure~\ref{fig:learning} tests the three parts of Proposition~\ref{prop:learning}, one per panel, and delimits their
range of validity. Panel~(a) shows, for the rank-order
game on $C_5$ ($K=3$), that the stationary law is Gibbs in the Hodge potential to second order. It also shows that under
non-uniform revision rates ($\alpha_i\propto i+1$), this is true of the weighted potential $\phi_\alpha$ and false of
$\phi^\star$: the revision rates enter at first order. Panel~(b) compares, on all edges of the game graph of the
intersection ($\beta=0.02$), the stationary current with $\beta h/(NK^{N+1})$, where $h$ is evaluated from
Proposition~\ref{prop:hflow}; the harmonic flow takes finitely many values, hence the clusters. Panel~(c) is the main check. It plots
$\sigma_\beta/\beta^2$ against $\Eharm/(NK^{N+1})$, with no fitted parameter, for four families of different origin:
rank-order games on all $112$ non-isomorphic connected graphs on six vertices; the same games with randomised tie-breaking, whose
harmonic energy is smaller by the factor $(K-2)/(K+1)$ (Section~\ref{sec:reading}); the interpolation
$u_\vartheta=\delta\phi^\star+\vartheta h$ on $C_5$, $\vartheta\in[0.15,1.6]$; and thirty games on four agents with i.i.d.\
standard normal payoffs. The last family is not pairwise, yet obeys the same law, as Proposition~\ref{prop:learning} asserts for arbitrary
payoffs. For each game $\beta\max_e|u(e)|=0.02$; the largest relative deviation from the line over all $284$ points
is $1.2\times10^{-3}$.
Panel~(d) shows the two regimes of Remark~\ref{rem:regimes} on the intersection: the ratio of $\sigma_\beta$ to its
prediction is $0.99$ at $\beta=0.1$, $0.88$ at $\beta=0.35$, $0.42$ at $\beta=1$, and $0.03$ at $\beta=8$, while the
stationary mass of the three pure-strategy Nash equilibria of Proposition~\ref{prop:ne} rises from $10^{-3}$ and remains near $0.65$ up to
$\beta=8$.

\subsection{Reading a Layout}
This experiment shows how the closed forms can be used for design, before any simulation. Table~\ref{tab:layouts}
compares four ways of routing eight agents: separate channels, so that the conflicts form a perfect matching; a single
merge cycle; one agent that conflicts with all others; and the unprotected crossing of Example~\ref{ex:intersection}. By \eqref{eq:quanta} and
Proposition~\ref{prop:learning}, the leading-order entropy production counts the conflicts, and the measured
$\sigma_\beta/\beta^2$ agrees with the prediction $(K^2-1)|E|/(6NK^2)$ to five digits. This has two consequences. First, each removed conflict lowers the leading-order entropy production by the same
amount, whatever happens to the ratio: the fully channelised layout has the smallest potential share, $1/2$, and is the closest to a reversible
system. Second, the ratio does not rank layouts: the star and the unprotected crossing share $\rho_\pot=4/5$, yet the
crossing has $16/7$ times the entropy production. A hub concentrates conflicts on one agent, which raises $|E_2|$ and hence
$\Epot$ at almost no harmonic cost; whether that is desirable is for the designer to decide.

\subsection{Beyond the Complete Graph}\label{sec:beyond}
This experiment probes the hypotheses of Theorem~\ref{thm:finiteK}. The theorem is stated for $K_N$, which has the protocol-reversal symmetry of Proposition~\ref{prop:recip}(c). Off $K_N$ the affine payoff can be slightly beaten at finite $K$. Under the index protocol on $K_{3,3}$ the best payoff, found by maximising the Rayleigh quotient
\eqref{eq:rayleigh} over $F$, reaches $\rho_\pot=0.7635$ for $K=3$ and $0.7504$ for $K=7$, against $3/4$ for the
linear payoff; the prism and the octahedron (linear values $3/4$ and $4/5$) are likewise beaten; on the cycles we tested the linear payoff was optimal. The excess decreases with $K$; we do not claim a single cause. Randomised tie-breaking on $K_N$ behaves likewise
($N=5$, $K=3$: $0.955$ at the optimum against $0.941$ for the linear payoff).

\section{Conclusion}\label{sec:conclusion}
This paper has measured what a potential-game surrogate misses in priority coordination. For the linear payoff the
missed energy is the number of conflicts, on every conflict graph and for every deterministic protocol, so the best
common-objective model of the agents' choice odds errs by at least $1/(d_{\max}+1)$ for every rationality $\beta>0$, a floor
set by the conflicts per agent rather than by the size of the system. The
missed part is invisible to strict better-response paths and to the equilibrium service order, but it is, to leading
order, the probability current of low-rationality learning, and on the complete conflict graph no rank-based payoff
brings it below $1/N$. Since the counts are available from the layout before any simulation, layouts should be
compared by the pair of energies rather than by their ratio: the star and the unprotected crossing of
Table~\ref{tab:layouts} have the same potential share, yet the crossing has $16/7$ times the entropy production.

At finite $K$ the extremality of affine payoffs can fail off the complete graph and under randomised tie-breaking
(Section~\ref{sec:beyond}); whether it holds on every regular graph for continuous priorities, for which $N$ it fails
at $K=2$, and whether Theorem~\ref{thm:finiteK} extends to orientations with directed cycles remain open. The analysis
has taken the ordering abstraction and an ideal execution layer as given; coupling the decision layer with the
continuous dynamics of the agents is the subject of ongoing work.

%======================================================================
\appendices

\section{Proof of Theorem~\ref{thm:spectrum}}\label{app:A}\label{app:A-spec}

Write $n=N-1$, $\mathsf w_m(q)=[q(1-q)]^m$, and $p^{(m)}_j(q)=P^{(m,m)}_j(1-2q)$, which are orthogonal on $[0,1]$ with weight $\mathsf w_m$. The required integral identities are collected in Appendix~\ref{sec:cont} of the SM.

\emph{Step 1: the two energies.} By the continuum analogue of Lemma~\ref{lem:numop}, $\|Pu\|^2=\sum_m\|g_m\|^2/m$ and $\|u\|^2=\sum_i\|u_i-A_i u_i\|^2$. Conditional on $s_{-i}$, $A_i u_i=\int_0^1f(B_t)\,\mathrm{d}t$, where $B_t=\#\{j\ne i:s_j>t\}$ is a sum of independent Bernoulli$(1-t)$ variables. Its Hoeffding components are products of centred indicators with covariance kernel $\kappa$. Lemmas~\ref{lem:dir} and~\ref{lem:bern} of the SM, together with the conservation law, therefore give $\|g_m\|^2=(N-m)^2J_m(f)$ and the corresponding total form, yielding \eqref{eq:contquad}.

\emph{Step 2: from discrete Chebyshev to Jacobi polynomials.} Let $\mathsf b_f(q)=\E f(\Bin(n,q))$ be the Bernstein polynomial of $f$. Since $\E[\Bin(n,q)^{\underline r}]=n^{\underline r}q^r$, $\mathsf b_{t_p}(q)=P_p(1-2q)$, where $P_p$ is the Legendre polynomial of degree $p$. The Bernstein derivative identity gives $d_m[f]=\mathsf b_f^{(m)}/n^{\underline m}$, and the $m$-th derivative of $P_p$ is a multiple of $P^{(m,m)}_{p-m}$; hence
\begin{equation}\label{eq:A-dm}
d_m[t_p]=c_{p,m}p^{(m)}_{p-m},\qquad
c_{p,m}=\frac{(-1)^m(p+m)!}{p!\,n^{\underline m}},
\end{equation}
for $m\le p$, and $d_m[t_p]=0$ for $m>p$.

\emph{Step 3: the kernel $\kappa^m$.} Let $(\mathsf K_mg)(q)=\int_0^1\kappa(q,q')^mg(q')\,\mathrm{d}q'$. Splitting at $q'=q$ shows that $\mathsf K_m$ maps polynomials of degree at most $j$ to $\mathsf w_m$ times polynomials of degree at most $j$. Thus $\mathsf K_m p^{(m)}_j=\mathsf w_m q_j$ with $\deg q_j\le j$. By self-adjointness, $q_j$ is $\mathsf w_m$-orthogonal to all polynomials of degree $<j$, so $q_j$ is proportional to $p^{(m)}_j$. The proportionality constant, obtained from the Rodrigues formula at $q\to0$ (Lemma~\ref{lem:Km} of the SM), is
\begin{equation}\label{eq:A-Keig}
\mathsf K_m p^{(m)}_j=
\frac{m}{(j+m)(j+m+1)}\,\mathsf w_m p^{(m)}_j .
\end{equation}

\emph{Step 4: part (a).} For $m\le\min(p,p')$, \eqref{eq:A-dm}--\eqref{eq:A-Keig} give
\begin{align*}
 & J_m(t_p,t_{p'}) \\
=& \binom Nm c_{p,m}c_{p',m}\frac{m}{p'(p'+1)}
\int_0^1\mathsf w_m\,p^{(m)}_{p-m}p^{(m)}_{p'-m}\,\mathrm{d}q.
\end{align*}
This vanishes for $p\ne p'$ by Jacobi orthogonality. For $p=p'$,
$\int_0^1\mathsf w_m(p^{(m)}_{p-m})^2\mathrm{d}q=(p!)^2/[(2p+1)(p+m)!(p-m)!]$, giving \eqref{eq:Jdiag}; for $m>p$, it vanishes by Step~2.

\emph{Step 5: part (b).} Substituting \eqref{eq:Jdiag} into \eqref{eq:contquad}, the potential form of $t_p$ is a positive common factor times
$\sum_{m=1}^p (N-m)!(p+m)!/[m!(p-m)!]$. The Chu--Vandermonde identity gives
\[
\sum_{m=0}^p\frac{(N-m)!(p+m)!}{m!(p-m)!}
=\frac{(N+p+1)!(N-p)!}{(N+1)!},
\]
so subtracting the $m=0$ term and dividing by the total form of $t_p$, which is $\frac N{N+1}\sum_kt_p(k)^2$ by Step~1, with
$\sum_{k=0}^{N-1}t_p(k)^2=(N+p)!(N-1-p)!/[(2p+1)(N-1)!^2]$ \cite{koekoek2010}, yields \eqref{eq:lambda}. Direct substitution gives $\lambda_1=(N-1)/N$.

Finally, write $N\lambda_p=D_p(1-\varepsilon_p)$, where $D_p=N(N+1)/[p(p+1)]-1$ and $\varepsilon_p=\binom Np/\binom{N+p+1}p$. The inequality $\lambda_{p+1}<\lambda_p$ is equivalent to
\[
\frac{1-\varepsilon_p}{\varepsilon_p}>
\frac{p(p+1)(N-p-1)}{N(N+1)}.
\]
Since
$\varepsilon_p^{-1}=\prod_{i=0}^{p-1}(1+\frac{p+1}{N-i})
\ge1+\frac{p(p+1)}N$,
the left-hand side is at least $p(p+1)/N$, which is strictly larger than the right-hand side. This proves the claim.
 \hfill$\square$

\section{Proof of Theorem~\ref{thm:finiteK}}\label{app:B}
This appendix contains the proof arguments for Theorem~\ref{thm:finiteK}, following the four stages of
Section~\ref{sec:finiteK}. Lengthy algebraic expansions, the exact symbolic computations
for $K=3,4$ and the rational certificates are in Appendix~\ref{app:finiteK} of the SM, whose steps are cited below.
Throughout, $C=K_N$, the protocol is the index protocol, and $u_i=f(\pos_i)$.

\subsection{The linear payoff is a null vector: proof of Lemma~\ref{lem:B}}\label{app:B-kernel}

The identity $Q_K\one=0$ is equivalent to
$\langle Pu^{\rm lin}-\frac{N-1}{N}u^{\rm lin},u^f\rangle=0$ for every $f$, where $u^{\rm lin}$ and $u^f$ are the flows of the linear game and of $\calG_f(K_N,K)$. By Proposition~\ref{prop:energies}, with $d_i=N-1$, $Pu^{\rm lin}$ is the flow generated by $(N-1)(s_i-\bar s)/K$, while $u^{\rm lin}-Pu^{\rm lin}$ is the flow of the interaction terms in \eqref{eq:rankkernel}. Writing the resulting flow inner products as conditional covariances and using \eqref{eq:cons} gives the equivalent identity
\begin{align}\label{eq:A-B2}
& \sum_i\E[\pos_i]\,\E[f(\pos_i)]-\sum_{r=0}^{N-1}rf(r) \nonumber\\
=& \frac{N+1}{K}\sum_i\Cov\bigl(s_i,f(\pos_i)\bigr).    
\end{align}

Both sides are linear in $f$, so it suffices to verify \eqref{eq:A-B2} for $f(r)=x^r$ for $N$ distinct values of $x$. Conditional on $s_i$, the events that the other agents precede $i$ are independent, with probabilities $(K-s_i)/K$ or $(K-1-s_i)/K$ according to the protocol orientation. The resulting generating functions are products $\mathsf A_a^{\,i}\mathsf B_a^{\,N-1-i}$ with $\mathsf A_a=\mathsf B_{a+1}$; the telescoping calculation in Step~\ref{step:B} of the SM proves \eqref{eq:A-B2}. Hence $Q_K\one=0$.

Finally, Proposition~\ref{prop:energies} on $K_N$, where $d_i=N-1$, gives the Rayleigh quotient
$\rho_\pot=(N-1)/N$.
 \hfill$\square$

\subsection{Stage 1: splitting by interaction order}\label{app:B-orders}
By Lemma~\ref{lem:numop} and \eqref{eq:cons}, $g_T=\sum_{i\in T}[u_i]_T$, and $Q_K=K^{N+1}\sum_{m\ge1}Q^{(m)}_K$ with
\begin{equation}\label{eq:A-Qm}
Q^{(m)}_K=\sum_{|T|=m}\Big[\tfrac{N-1}N\sum_{i\in T}\|[u_i]_T\|^2-\tfrac1m\Big\|\sum_{i\in T}[u_i]_T\Big\|^2\Big].
\end{equation}

\begin{proposition}[Higher orders]\label{prop:A-orders}
For $N\ge3$, $K\ge2$ and every $T$ with $|T|=m$, $\|\sum_{i\in T}[u_i]_T\|^2\le2\sum_{i\in T}\|[u_i]_T\|^2$; hence
$Q^{(m)}_K\succeq0$ for all $m\ge3$.
\end{proposition}

\begin{proof}
Write $X_i=[u_i]_T=(I-A_i)[\Psi^{(i)}D^{(i)}(s_i)]$. By Lemma~\ref{lem:struct} of the SM, for $i\ne k$ the bilinear kernel of $\langle X_i,X_k\rangle$ in $(D^{(i)},D^{(k)})$ is entrywise nonpositive, while the kernel subtracted in $\|X_i\|^2$ is entrywise nonnegative. Let $\tilde X_i$ be obtained by replacing $D^{(i)}$ with $|D^{(i)}|$. These sign properties imply
$\langle X_i,X_k\rangle\le-\langle\tilde X_i,\tilde X_k\rangle$ for $i\ne k$ and $\|\tilde X_i\|\le\|X_i\|$. Hence
\[
\begin{aligned}
\sum_{i\ne k}\langle X_i,X_k\rangle
&\le-\sum_{i\ne k}\langle\tilde X_i,\tilde X_k\rangle
=\sum_i\|\tilde X_i\|^2-\Big\|\sum_i\tilde X_i\Big\|^2\\
&\le\sum_i\|X_i\|^2.
\end{aligned}
\]
Therefore $\|\sum_{i\in T}X_i\|^2\le2\sum_{i\in T}\|X_i\|^2$. Consequently,
\[
Q_K^{(m)}
\ge\big(\frac{N-1}{N}-\frac2m\big)\sum_{|T|=m}\sum_{i\in T}\|[u_i]_T\|^2\ge0,
\]
since $N,m\ge3$. Thus $Q_K^{(m)}\succeq0$ for every $m\ge3$.
\end{proof}

The argument uses no property of $f$. It remains to show $Q^{(1)}_K+Q^{(2)}_K\succeq0$, where $Q^{(1)}_K\preceq0$.

\subsection{Stage 2: reduction to bond variables}\label{app:B-bond}
For $k=0,\dots,N-2$ and a level $a$, define the \emph{bond vector}
$\hat D_k(a)=\E\big[\Delta f\big(\Bin(k,\tfrac{K-a}K)+\Bin(N-2-k,\tfrac{K-1-a}K)\big)\big]$, and let $\hat D(f)=(\hat D_k)_{k=0}^{N-2}$. Let $\mu_k=\frac{N-1-k}N$ and
$\nu_k=\frac{k+1}N$. For a general array $D=(D_k)_{k=0}^{N-2}$ with $D_k\in\R^K$, write $G_k,H_k$ for the
inclusive and exclusive cumulative sums of $D_k$, as in Section~\ref{sec:finiteK}.

\begin{proposition}[Two-body reduction]\label{prop:A-bond}
For every payoff $f$,
\[
Q^{(1)}_K(f)+Q^{(2)}_K(f)=\fB_{N,K}(\hat D(f)),
\]
where $\fB_{N,K}$ is the quadratic form in arrays $D$ defined in Step~\ref{step:bond} of the SM. With
\[
Z_k=\frac{N-2}{2}[\mu_k(\calD-2\Var)(H_k)+\nu_k(\calD-2\Var)(G_k)],
\]
it satisfies
\begin{equation}\label{eq:A-SOS}
K^2\fB_{N,K}(D)
\!=\!\sum_k Z_k+\mathrm{Sq}(D)
-\sum_k(2N\mu_k\nu_k-\tfrac12)\Var(D_k),
\end{equation}
where $\mathrm{Sq}(D)\ge0$ is the sum of variances
\begin{equation}\label{eq:A-Sq}
\begin{aligned}
\mathrm{Sq}(D)=\tfrac1{2N}\sum_{k<l}\big[&
(N{-}1{-}l)\Var(H_k-H_l+D_l)\\[-2pt]
&+(k{+}1)\Var(G_k-G_l+D_k)\\
&+(l{-}k)\Var(H_k-G_l)\big].
\end{aligned}
\end{equation}
Moreover, $\mathrm{Sq}(D)=0$ iff $D\in\operatorname{span}(\one\otimes\one)$, and
$\hat D(f)\in\operatorname{span}(\one\otimes\one)$ iff $f$ is affine.
\end{proposition}

\begin{proof}
Given $s_i=a$, the events that opponents precede $i$ are independent, with probability $(K-a)/K$ for opponents favoured by the protocol and $(K-1-a)/K$ otherwise.

\emph{First order.} Lowering $a$ by one increases each such probability by $1/K$. Since
$\E f(R+\mathrm{Bern}(p+1/K))-\E f(R+\mathrm{Bern}(p))
=K^{-1}\E\Delta f(R)$,
one obtains, up to an additive constant,
$\varphi_j(a):=\E[u_j\mid s_j=a]=-K^{-1}\Omega_j$, where
$\Omega_j=\sum_{k<j}H_k+\sum_{k\ge j}G_k$. Hence
$Q_K^{(1)}=-\frac1{NK^2}\sum_j\Var(\Omega_j)$.

\emph{Second order.} Conditioning also on $s_j=b$ isolates the indicator of whether $j$ precedes $i$:
$\E[u_i\mid a,b]=\chi(a)+\one[j\prec i]\,\hat D_k(a)$, with $k\in\{i-1,i\}$. Thus $[u_i]_{\{i,j\}}$ is a centred Bernoulli indicator times the corresponding bond vector. Substitution into \eqref{eq:A-Qm} gives the diagonal Bernoulli-covariance terms and, by complementarity of the two orientations, cross terms $-K^{-2}\Cov(H_k,G_l)$. Summation by parts and
$n\sum_k\Var(z_k)-\Var(\sum_k z_k)=\sum_{k<l}\Var(z_k-z_l)$
then yield \eqref{eq:A-SOS}--\eqref{eq:A-Sq}; the detailed rearrangement is given in Steps~\ref{step:bond}--\ref{step:kerSq} of the SM.

Since every coefficient in \eqref{eq:A-Sq} is nonnegative, $\mathrm{Sq}=0$ iff all displayed variance terms vanish. Step~\ref{step:kerSq} of the SM shows that this happens iff $D\in\operatorname{span}(\one\otimes\one)$. Finally, by the definition of $\hat D(f)$, this occurs iff all second differences of $f$ vanish, equivalently iff $f$ is affine.
\end{proof}

\subsection{Stage 3: the level inequality}\label{app:B-level}
By Lemma~\ref{lem:level}, $Z_k\ge\frac{N-2}2(K-2)\mu_k\nu_k\Var(D_k)$, so the only negative term of \eqref{eq:A-SOS}
is compensated bond by bond. The proof of Lemma~\ref{lem:level} works in the step coordinates $c_m=D(m)-D(m-1)$. In
these coordinates the difference of the two sides of \eqref{eq:level} is the Hadamard product of the Green's function
$\min(m,m')(K-\max(m,m'))/K$ of the path with a one-pair matrix minus a rank-one term. The inverse of the one-pair
matrix is a path Laplacian with Robin ends, whose row sums show that the second factor is positive semidefinite (and singular), and
the Schur product theorem concludes (Lemma~\ref{lem:Lstar} of the SM). The constant $K-2$ cannot be improved, and for
$K=2$ the inequality gives no margin.

\subsection{Stage 4: the Perron argument and three branches}\label{app:B-perron}
Let $\varrho(\mu):=\mu\nu+\tfrac12c^\star(\mu)$, where $c^\star(\mu)$ is the best constant in \eqref{eq:level} for the
given $\mu$. Retaining part of $\mathrm{Sq}$, positivity of $\fB_{N,K}$ reduces to $NL-2\operatorname{diag}(1/\theta_k)\succeq0$, where $L$ is the Dirichlet Laplacian of the path of bonds and
$\theta_k=(N-2)\varrho(\mu_k)-N\mu_k\nu_k+\frac12$. The vector $v_k=(k+1)(N-1-k)$ satisfies $Lv=2\cdot\one$ exactly,
so by the $M$-matrix criterion it suffices that, for every bond,
\begin{equation}\label{eq:A-T}
(N-2)\,\varrho(\mu_k)-2N\mu_k\nu_k+\tfrac12\ \ge\ 0 .
\end{equation}
Lemma~\ref{lem:level} gives $\varrho(\mu)\ge K\mu\nu/2$, and with $\mu\nu\le\frac14$ condition \eqref{eq:A-T} holds
whenever $(N-2)(K-4)\ge4$. This covers $K\ge5$ up to finitely many $N$. At $K=4$ the bound from Lemma~\ref{lem:level}
never suffices, and one uses the true $\varrho(\mu)$, the smallest root of an explicit cubic; root isolation by Sturm
sequences certifies $\varrho\ge2\mu\nu+5(\mu\nu)^3$ and hence \eqref{eq:A-T} for $N\ge10$. At $K=3$ the level space
modulo constants is two-dimensional; an explicit congruence removes the same-bond coupling and leaves a sum of squares
plus a $Z$-matrix with a positive Perron vector for $N\ge17$, the required polynomial identities in $N$ being established
by exact computer algebra. The finitely many remaining pairs $(K,N)$ are settled by $LDL^\top$ factorisations in exact
rational arithmetic (Table~\ref{tab:A-coverage}). Steps~\ref{step:Sqlow}--\ref{step:certs} of the SM give the details.

\begin{table}[t]\centering
\caption{Coverage of Theorem~\ref{thm:finiteK} for $K\ge3$, $N\ge3$: by the large-$N$ branch of
Appendix~\ref{app:B-perron} or by an exact rational certificate.}\label{tab:A-coverage}
\small
\begin{tabular}{@{}lccc@{}}\toprule
$K$ & all large $N$ & bond-space cert. & image cert.\\\midrule
$3$ & $N\ge17$ & -- & $4\le N\le16^{\ast}$\\
$4$ & $N\ge10$ & $5\le N\le9$ & $N=3,4$\\
$5$ & $N\ge6$ & $N=4,5$ & $N=3$\\
$6,7$ & $N\ge4$ & $N=3$ & --\\
$\ge8$ & $N\ge3$ & -- & --\\\bottomrule
\end{tabular}\\[2pt]
{\footnotesize $^{\ast}$At $(K,N)=(3,3)$ the two-body form vanishes identically on the image, and
$Q_3=\frac83\big(\begin{smallmatrix}1&-1\\-1&1\end{smallmatrix}\big)$ is computed exactly from the definition.}
\end{table}

\subsection{Assembly}\label{app:B-assembly}
Stages 1 to 4 give $Q_K\succeq0$ for all $K\ge3$ and $N\ge3$. The kernel statement $\ker Q_K=\operatorname{span}(\one)$ follows
from Lemma~\ref{lem:B}, the kernel of $\mathrm{Sq}$ and the last sentence of Proposition~\ref{prop:A-bond}
(Step~\ref{step:assembly} of the SM). For $N=2$ the space of increments is one-dimensional and the statement follows from
Lemma~\ref{lem:B}. The counterexample for $K=2$ is verified in exact rational arithmetic (Appendix~\ref{app:repro} of the SM).
\hfill$\square$

\section*{References}
\bibliographystyle{IEEEtran}
\bibliography{sample}

% Generated by IEEEtran.bst, version: 1.14 (2015/08/26)
\begin{thebibliography}{10}
\providecommand{\url}[1]{#1}
\csname url@samestyle\endcsname
\providecommand{\newblock}{\relax}
\providecommand{\bibinfo}[2]{#2}
\providecommand{\BIBentrySTDinterwordspacing}{\spaceskip=0pt\relax}
\providecommand{\BIBentryALTinterwordstretchfactor}{4}
\providecommand{\BIBentryALTinterwordspacing}{\spaceskip=\fontdimen2\font plus
\BIBentryALTinterwordstretchfactor\fontdimen3\font minus \fontdimen4\font\relax}
\providecommand{\BIBforeignlanguage}[2]{{%
\expandafter\ifx\csname l@#1\endcsname\relax
\typeout{** WARNING: IEEEtran.bst: No hyphenation pattern has been}%
\typeout{** loaded for the language `#1'. Using the pattern for}%
\typeout{** the default language instead.}%
\else
\language=\csname l@#1\endcsname
\fi
#2}}
\providecommand{\BIBdecl}{\relax}
\BIBdecl

\bibitem{dresner2008}
K.~Dresner and P.~Stone, ``A multiagent approach to autonomous intersection management,'' \emph{Journal of artificial intelligence research}, vol.~31, pp. 591--656, 2008.

\bibitem{levin2017}
M.~W. Levin and D.~Rey, ``Conflict-point formulation of intersection control for autonomous vehicles,'' \emph{Transportation Research Part C: Emerging Technologies}, vol.~85, pp. 528--547, 2017.

\bibitem{xu2018}
B.~Xu, S.~E. Li, Y.~Bian, S.~Li, X.~J. Ban, J.~Wang, and K.~Li, ``Distributed conflict-free cooperation for multiple connected vehicles at unsignalized intersections,'' \emph{Transportation Research Part C: Emerging Technologies}, vol.~93, pp. 322--334, 2018.

\bibitem{zanardi2021}
A.~Zanardi, E.~Mion, M.~Bruschetta, S.~Bolognani, A.~Censi, and E.~Frazzoli, ``Urban driving games with lexicographic preferences and socially efficient nash equilibria,'' \emph{IEEE Robotics and Automation Letters}, vol.~6, no.~3, pp. 4978--4985, 2021.

\bibitem{li2023}
N.~Li, Y.~Yao, I.~Kolmanovsky, E.~Atkins, and A.~R. Girard, ``Game-theoretic modeling of multi-vehicle interactions at uncontrolled intersections,'' \emph{IEEE Transactions on Intelligent Transportation Systems}, vol.~23, no.~2, pp. 1428--1442, 2022.

\bibitem{TAVV6081}
Z.~Lin, L.~Wu, Z.~Tian, A.~Lomuscio, and J.~Lan, ``Scalable and safe multi-agent coordination with reconstructed level-k monte carlo tree search,'' in \emph{Proceedings of the 25th International Conference on AAMAS}, 2026, pp. 394--402.

\bibitem{11130528}
Z.~Lin, J.~Lan, C.~Anagnostopoulos, Z.~Tian, and D.~Flynn, ``Safety-critical multi-agent mcts for mixed traffic coordination at unsignalized intersections,'' \emph{IEEE Transactions on Intelligent Transportation Systems}, vol.~26, no.~11, pp. 18\,992--19\,006, 2025.

\bibitem{monderer1996}
D.~Monderer and L.~S. Shapley, ``Potential games,'' \emph{Games and economic behavior}, vol.~14, no.~1, pp. 124--143, 1996.

\bibitem{marden2009}
J.~R. Marden, G.~Arslan, and J.~S. Shamma, ``Cooperative control and potential games,'' \emph{IEEE Transactions on Systems, Man, and Cybernetics, Part B (Cybernetics)}, vol.~39, no.~6, pp. 1393--1407, 2009.

\bibitem{blume1993}
L.~E. Blume, ``The statistical mechanics of strategic interaction,'' \emph{Games and Economic Behavior}, vol.~5, no.~3, pp. 387--424, 1993.

\bibitem{marden2012}
J.~R. Marden and J.~S. Shamma, ``Revisiting log-linear learning: Asynchrony, completeness and payoff-based implementation,'' \emph{Games and Economic Behavior}, vol.~75, no.~2, pp. 788--808, 2012.

\bibitem{marden2014efficiency}
J.~R. Marden and T.~Roughgarden, ``Generalized efficiency bounds in distributed resource allocation,'' \emph{IEEE Transactions on Automatic Control}, vol.~59, no.~3, pp. 571--584, 2014.

\bibitem{paccagnan2020utilityI}
D.~Paccagnan, R.~Chandan, and J.~R. Marden, ``Utility design for distributed resource allocation---{Part I}: Characterizing and optimizing the exact price of anarchy,'' \emph{IEEE Transactions on Automatic Control}, vol.~65, no.~11, pp. 4616--4631, 2020.

\bibitem{paccagnan2022utilityII}
D.~Paccagnan and J.~R. Marden, ``Utility design for distributed resource allocation---{Part II}: Applications to submodular, covering, and supermodular problems,'' \emph{IEEE Transactions on Automatic Control}, vol.~67, no.~2, pp. 618--632, 2022.

\bibitem{liu2023}
M.~Liu, I.~Kolmanovsky, H.~E. Tseng, S.~Huang, D.~Filev, and A.~Girard, ``Potential game-based decision-making for autonomous driving,'' \emph{IEEE Transactions on Intelligent Transportation Systems}, vol.~24, no.~8, pp. 8014--8027, 2023.

\bibitem{liu2025}
M.~Liu, H.~E. Tseng, D.~Filev, A.~Girard, and I.~Kolmanovsky, ``Game projection and robustness for game-theoretic autonomous driving,'' \emph{IEEE Transactions on Intelligent Transportation Systems}, vol.~26, no.~3, pp. 3446--3457, 2025.

\bibitem{candogan2011}
O.~Candogan, I.~Menache, A.~Ozdaglar, and P.~A. Parrilo, ``Flows and decompositions of games: Harmonic and potential games,'' \emph{Mathematics of Operations Research}, vol.~36, no.~3, pp. 474--503, 2011.

\bibitem{candogan2013}
O.~Candogan, A.~Ozdaglar, and P.~A. Parrilo, ``Near-potential games: Geometry and dynamics,'' \emph{ACM Transactions on Economics and Computation}, vol.~1, no.~2, pp. 11:1--11:32, 2013.

\bibitem{candogan2013geb}
------, ``Dynamics in near-potential games,'' \emph{Games and Economic Behavior}, vol.~82, pp. 66--90, 2013.

\bibitem{bichler2025}
M.~Bichler, D.~Legacci, P.~Mertikopoulos, M.~Oberlechner, and B.~S.~R. Pradelski, ``Characterizing the convergence of game dynamics via potentialness,'' \emph{Transactions on Machine Learning Research}, 2025.

\bibitem{sandholm2010}
W.~H. Sandholm, ``Decompositions and potentials for normal form games,'' \emph{Games and Economic Behavior}, vol.~70, no.~2, pp. 446--456, 2010.

\bibitem{hwang2011}
S.-H. Hwang and L.~Rey-Bellet, ``Strategic decompositions of normal form games: Zero-sum games and potential games,'' \emph{Games and Economic Behavior}, vol. 122, pp. 370--390, 2020.

\bibitem{arefizadeh2024}
S.~Arefizadeh, A.~Nedi{\'c}, and G.~Dasarathy, ``On characterizations of potential and ordinal potential games,'' 2024, arXiv:2405.06253.

\bibitem{abdou2022}
J.~Abdou, N.~Pnevmatikos, M.~Scarsini, and X.~Venel, ``Decomposition of games: Some strategic considerations,'' \emph{Mathematics of Operations Research}, vol.~47, no.~1, pp. 176--208, 2022.

\bibitem{gilles2026}
R.~P. Gilles, ``Second-order potentials for finite games: Existence, characterisation, and game decomposition,'' 2026, arXiv:2608.01967.

\bibitem{cheng2014finite}
D.~Cheng, ``On finite potential games,'' \emph{Automatica}, vol.~50, no.~7, pp. 1793--1801, 2014.

\bibitem{cheng2015networked}
D.~Cheng, F.~He, H.~Qi, and T.~Xu, ``Modeling, analysis and control of networked evolutionary games,'' \emph{IEEE Transactions on Automatic Control}, vol.~60, no.~9, pp. 2402--2415, 2015.

\bibitem{cheng2016}
D.~Cheng, T.~Liu, K.~Zhang, and H.~Qi, ``On decomposed subspaces of finite games,'' \emph{IEEE Transactions on Automatic Control}, vol.~61, no.~11, pp. 3651--3656, 2016.

\bibitem{hao2022survey}
Y.~Hao and J.-F. Zhang, ``A survey on decomposition of finite strategic-form games,'' in \emph{Complex Systems: Spanning Control and Computational Cybernetics: Foundations}, ser. Studies in Systems, Decision and Control, P.~Shi, J.~Stefanovski, and J.~Kacprzyk, Eds.\hskip 1em plus 0.5em minus 0.4em\relax Cham: Springer, 2022, vol. 414, pp. 363--384.

\bibitem{zhang2017}
K.~Zhang, ``Polynomial representation for orthogonal projections onto subspaces of finite games,'' in \emph{Proc. 36th Chinese Control Conference}, 2017, pp. 11\,267--11\,272.

\bibitem{arditti2020sep}
L.~Arditti, G.~Como, and F.~Fagnani, ``On the separability of functions and games,'' \emph{IEEE Transactions on Control of Network Systems}, vol.~11, no.~2, pp. 831--841, 2024.

\bibitem{guo2024}
X.~Guo, X.~Li, and Y.~Zhang, ``An $\alpha$-potential game framework for {$N$}-player dynamic games,'' 2024, arXiv:2403.16962.

\bibitem{rokade2025}
K.~Rokade, A.~Jain, F.~Parise, V.~Krishnamurthy, and {\'E}.~Tardos, ``Asymmetric network games: $\alpha$-potential function and learning,'' 2025, arXiv:2508.06619.

\bibitem{legacci2024}
D.~Legacci, P.~Mertikopoulos, and B.~Pradelski, ``A geometric decomposition of finite games: Convergence vs. recurrence under exponential weights,'' in \emph{Proceedings of the 41st International Conference on Machine Learning}, ser. Proceedings of Machine Learning Research, vol. 235.\hskip 1em plus 0.5em minus 0.4em\relax Vienna, Austria: PMLR, 2024, pp. 27\,137--27\,173.

\bibitem{legacci2024noregret}
D.~Legacci, P.~Mertikopoulos, C.~H. Papadimitriou, G.~Piliouras, and B.~S.~R. Pradelski, ``No-regret learning in harmonic games: Extrapolation in the face of conflicting interests,'' in \emph{Advances in Neural Information Processing Systems}, vol.~37, 2024.

\bibitem{boundedLLL2026}
Z.~Wang, E.~Akyol, and M.~M. Vasconcelos, ``Learning to coordinate over networks with bounded rationality,'' 2026, arXiv:2604.07751.

\bibitem{maes2010}
C.~Maes and K.~Neto{\v{c}}n{\'y}, ``Rigorous meaning of {M}c{L}ennan ensembles,'' \emph{Journal of Mathematical Physics}, vol.~51, no.~1, p. 015219, 2010.

\bibitem{szabo2015}
G.~Szab{\'o} and A.~Szolnoki, ``Congestion phenomena caused by matching pennies in evolutionary games,'' \emph{Physical Review E}, vol.~91, p. 032110, 2015.

\bibitem{szabo2016}
G.~Szab{\'o} and I.~Borsos, ``Evolutionary potential games on lattices,'' \emph{Physics Reports}, vol. 624, pp. 1--60, 2016.

\bibitem{brandt2009ranking}
F.~Brandt, F.~Fischer, P.~Harrenstein, and Y.~Shoham, ``Ranking games,'' \emph{Artificial Intelligence}, vol. 173, no.~2, pp. 221--239, 2009.

\bibitem{harks2011}
T.~Harks, M.~Klimm, and R.~H. M{\"o}hring, ``Characterizing the existence of potential functions in weighted congestion games,'' \emph{Theory of Computing Systems}, vol.~49, no.~1, pp. 46--70, 2011.

\bibitem{helmers1982}
R.~Helmers, \emph{Edgeworth Expansions for Linear Combinations of Order Statistics}, ser. Mathematical Centre Tracts.\hskip 1em plus 0.5em minus 0.4em\relax Amsterdam: Mathematisch Centrum, 1982, no. 105.

\bibitem{koekoek2010}
R.~Koekoek, P.~A. Lesky, and R.~F. Swarttouw, \emph{Hypergeometric Orthogonal Polynomials and Their $q$-Analogues}.\hskip 1em plus 0.5em minus 0.4em\relax Berlin: Springer, 2010.

\bibitem{strang2022}
A.~Strang, K.~C. Abbott, and P.~J. Thomas, ``The network {HHD}: Quantifying cyclic competition in trait-performance models of tournaments,'' \emph{SIAM Review}, vol.~64, no.~2, pp. 360--391, 2022.

\bibitem{stanley2012}
R.~P. Stanley, \emph{Enumerative Combinatorics, Volume 1}, 2nd~ed.\hskip 1em plus 0.5em minus 0.4em\relax Cambridge University Press, 2012.

\bibitem{gantmacher2002}
F.~R. Gantmacher and M.~G. Krein, \emph{Oscillation Matrices and Kernels and Small Vibrations of Mechanical Systems}.\hskip 1em plus 0.5em minus 0.4em\relax Providence, RI: AMS Chelsea Publishing, 2002, revised ed.

\end{thebibliography}

\vspace{-7 mm}
\begin{IEEEbiography}[{\includegraphics[width=1in,height=1.25in,clip,keepaspectratio]{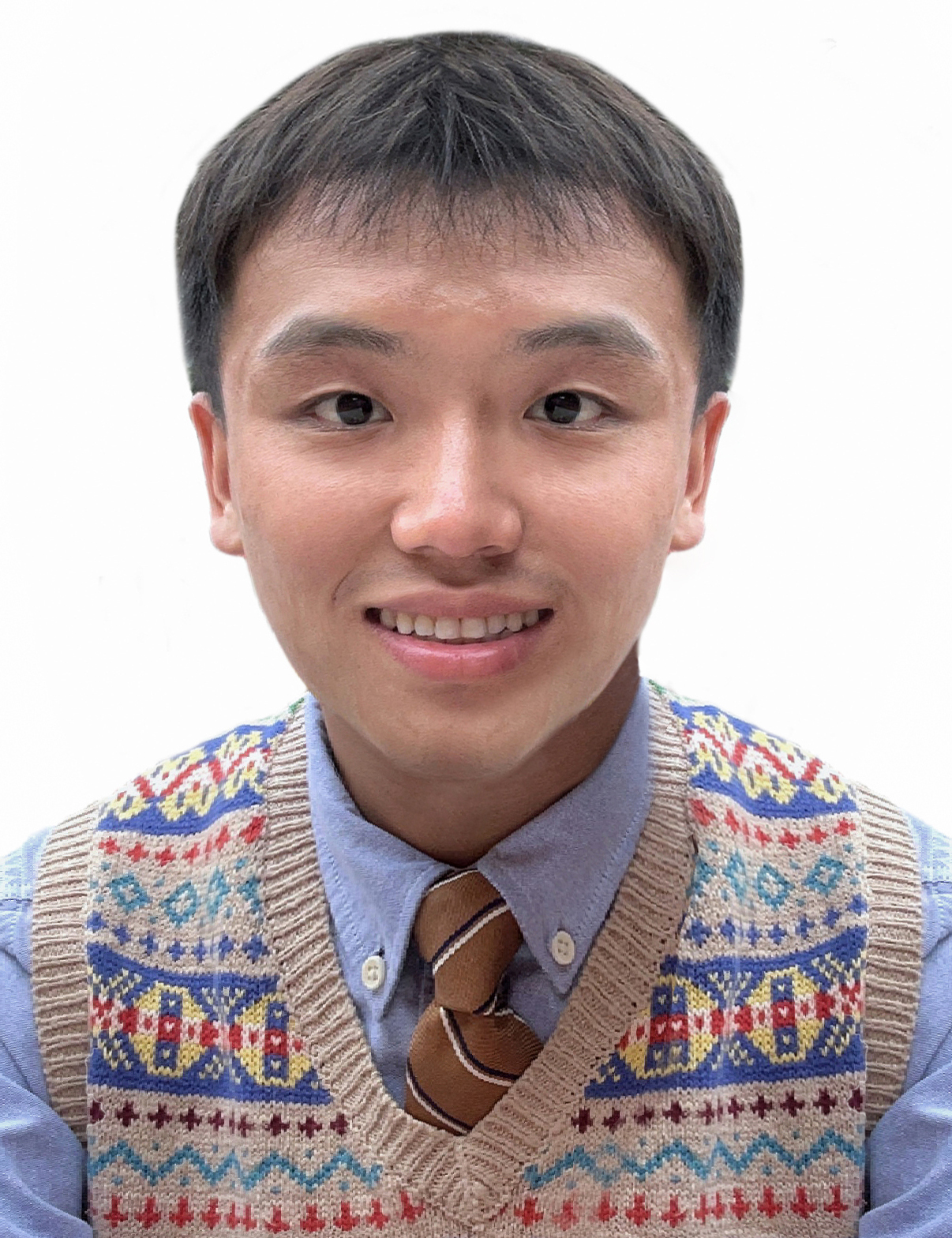}}]{Zhihao Lin}
received the Ph.D. degree in autonomous systems and connectivity from the University of Glasgow, Glasgow, U.K., in 2026. He is currently a JSPS Postdoctoral Fellow at the Graduate School of Information Science and Technology, The University of Osaka, Osaka, Japan. His research interests include multi-agent coordination, game theory and learning in games,
operator-theoretic methods for data-driven control, and safe autonomous driving.
\end{IEEEbiography}

\vspace{-7 mm}
\begin{IEEEbiography}[{\includegraphics[width=1in,height=1.25in,clip,keepaspectratio]{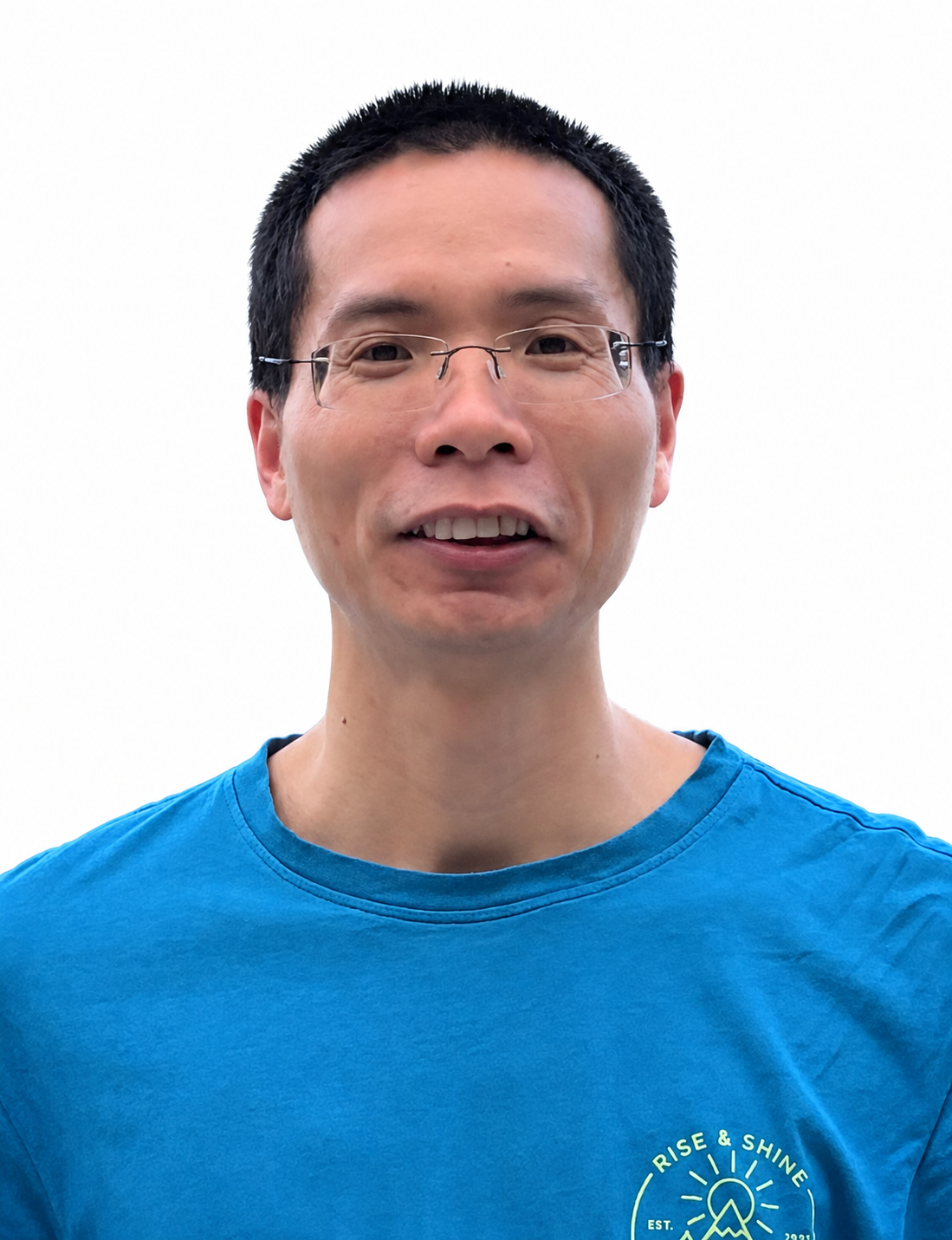}}]{Jianglin Lan}
received the Ph.D. degree in Engineering from the University of Hull in 2017. He is currently a Lecturer at the University of Glasgow and an Honorary Research Fellow with Imperial College London. From 2022 to 2025, he held a Leverhulme Early Career Fellowship at University of Glasgow. He is an Associate Editor for \emph{IEEE Transactions on Intelligent Transportation Systems}, an Editor for \emph{International Journal of Adaptive Control and Signal Processing}, and a member of the IFAC SAFEPROCESS Technical Committee. His research interests include safe AI, and control theory.
\end{IEEEbiography}

\vspace{-7 mm}
\begin{IEEEbiography}[{\includegraphics[width=1in,height=1.25in,clip,keepaspectratio]{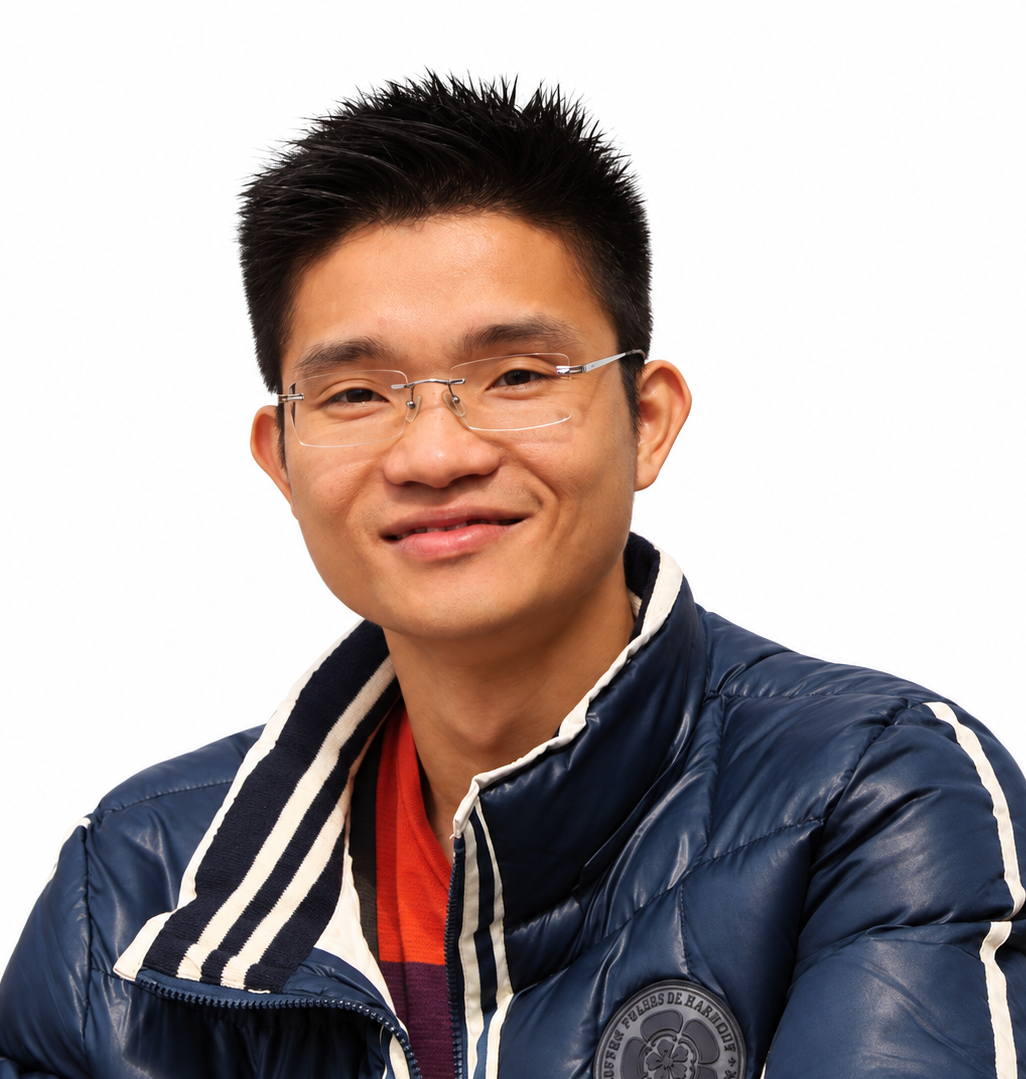}}]{Anh-Tu Nguyen}
(Senior Member, IEEE) received the Ph.D. degree in automatic control from the Universit\'e Polytechnique
Hauts-de-France, Valenciennes, France, in 2013.

He is currently an Associate Professor with INSA Hauts-de-France, Universit\'e Polytechnique Hauts-de-France,
Valenciennes, France. His research interests include robust control and estimation, cybernetic control systems, and
mechatronic applications.

Dr. Nguyen serves as an Associate Editor for several international journals, including the
\textsc{IEEE/ASME Transactions on Mechatronics}, the \textsc{IEEE Transactions on Automation Science and Engineering},
the \textsc{IEEE Transactions on Intelligent Transportation Systems}, \emph{Control Engineering Practice},
\emph{Engineering Applications of Artificial Intelligence}, and \emph{ISA Transactions}.
\end{IEEEbiography}

\vspace{-7 mm}
\begin{IEEEbiography}[{\includegraphics[width=1in,height=1.25in,clip,keepaspectratio]{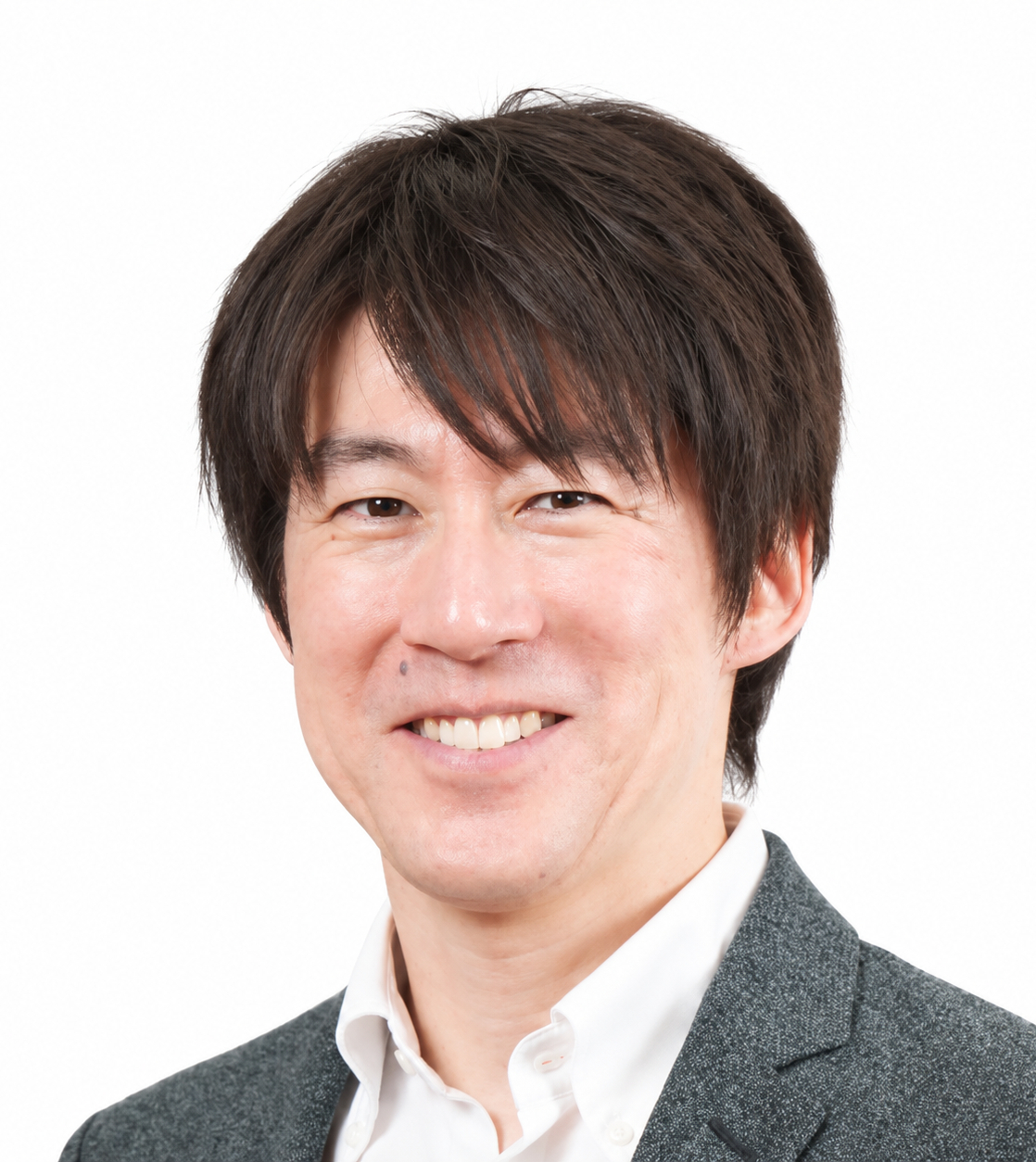}}]{Yoshinobu Kawahara}
received the B.S., M.S., and Ph.D. degrees in engineering
from the University of Tokyo, Tokyo, Japan, in 2003,
2005, and 2008, respectively.

He is currently a Professor with the Graduate School
of Information Science and Technology, The University
of Osaka, Osaka, Japan, and the Team Director of the
Dynamical Systems Learning Team at the RIKEN Center
for Advanced Intelligence Project, Japan. He is also
a Project Manager of the JST Moonshot R\&D Program
(Goal 3). He was previously a Professor with the
Institute of Mathematics for Industry, Kyushu
University, Fukuoka, Japan, and has held visiting
positions at the Max Planck Institute for Intelligent
Systems, T\"ubingen, Germany, and the University of
Washington, Seattle, WA, USA. His research interests
lie at the intersection of machine learning and
dynamical systems, including operator-theoretic methods
and their applications to scientific and engineering
problems.

Dr. Kawahara received the Young Scientists' Prize,
the Commendation for Science and Technology by the
Minister of Education, Culture, Sports, Science and
Technology of Japan, in 2020, and Best Paper Awards
from the Japanese Society for Artificial Intelligence
(2011 and 2020) and the Japanese Neural Network
Society (2020). He currently serves as an Action Editor
for \emph{Transactions on Machine Learning Research}
and \emph{Neural Networks}, and has served as an
Area Chair or Senior Area Chair for conferences
including NeurIPS, ICML, and ICLR.
\end{IEEEbiography}

% ---------------------------------------------------------------
% Supplementary material. Typeset with the main text (one compile,
% no xr/aux); split the PDF at this page for submission.
% ---------------------------------------------------------------
\clearpage
\onecolumn
\setcounter{page}{1}
\renewcommand{\thesection}{S-\Alph{section}}\setcounter{section}{0}\renewcommand{\thesectiondis}{S-\Alph{section}}
\renewcommand{\thesubsection}{\thesection.\arabic{subsection}}
\renewcommand{\theequation}{S\arabic{equation}}\setcounter{equation}{0}
\renewcommand{\thetable}{S\arabic{table}}\setcounter{table}{0}
\renewcommand{\thefigure}{S\arabic{figure}}\setcounter{figure}{0}\renewcommand{\theHfigure}{S\arabic{figure}}
\renewcommand{\theHsection}{S\Alph{section}}\renewcommand{\theHsubsection}{S\Alph{section}.\arabic{subsection}}
\renewcommand{\theHequation}{S\arabic{equation}}\renewcommand{\theHtable}{S\arabic{table}}
% steps of the proof of Theorem~\ref{thm:finiteK}: numbered by a counter, referenced by label
\newcounter{pstep}\setcounter{pstep}{-1}
\makeatletter
\newcommand{\stepsection}[2]{\stepcounter{pstep}\subsection{Step~\thepstep: #2}\edef\@currentlabel{\thepstep}\label{step:#1}}
\makeatother
\markboth{\hskip25pc SUPPLEMENTARY MATERIAL}{Lin \MakeLowercase{\textit{et al.}}: Supplementary material}
\begin{center}
{\Large\bfseries Supplementary Material for}\\[4pt]
{\large\itshape Priority Coordination Games:\\ Hodge Decomposition and a Sharp Design Limit}\\[10pt]
\end{center}
\bigskip
\noindent
This document gives the details behind the proofs in the appendices of the main text, the proofs of the further
results stated there, and documents the computer-assisted steps. References of the form ``Theorem~\ref{thm:finiteK}'' or
``\eqref{eq:rayleigh}'' point to the main text; auxiliary results proved here are numbered S1, S2, \dots, and equations
(S1), (S2), \dots. Appendix~\ref{app:finiteK} comes first because it proves the main technical result,
Theorem~\ref{thm:finiteK}; Appendix~\ref{sec:cont} gives a self-contained version of the proof of
Theorem~\ref{thm:spectrum} in Appendix~\ref{app:A} of the main text, and uses one lemma of Appendix~\ref{app:finiteK}.
\begin{itemize}
\item Appendix~\ref{app:finiteK} proves Theorem~\ref{thm:finiteK}
(Finite quantisation), in fourteen steps following the four-stage
outline of Section~\ref{sec:finiteK}; its last subsection
(Section~\ref{sec:recip}) proves Proposition~\ref{prop:recip}
(reciprocity and expansion in $K^{-2}$), independent of the rest.
\item Appendix~\ref{sec:pattern} proves the two statements about the harmonic component of the linear game made after
Proposition~\ref{prop:hflow} of the main text: the tie-pattern law (Theorem~\ref{thm:pattern-SM}) and the effect of
randomised tie-breaking (Proposition~\ref{prop:randtie-SM}).
\item Appendix~\ref{sec:cont} proves Theorem~\ref{thm:spectrum} (Hodge spectrum) and the identities
\eqref{eq:contquad}.
\item Appendix~\ref{sec:priced} proves the statement about priced priorities made after Corollary~\ref{cor:rigid} of
the main text (Proposition~\ref{prop:priced-SM}) and records the corresponding exact computation on the intersection.
\item Appendix~\ref{app:repro} lists the scripts behind every computer-assisted step, describes an end-to-end check
that certifies Theorem~\ref{thm:finiteK} for small $(K,N)$ without using any step of its proof, and says which
computations are \emph{not} part of any proof.
\end{itemize}
Proofs given in full in the main text (Lemma~\ref{lem:numop}, Proposition~\ref{prop:pairwise} and its corollary,
Propositions~\ref{prop:energies}, \ref{prop:hflow}, \ref{prop:fip}, \ref{prop:ne}
and~\ref{prop:learning}) are not repeated.

\begin{figure}[H]\centering
\includegraphics[width=\textwidth]{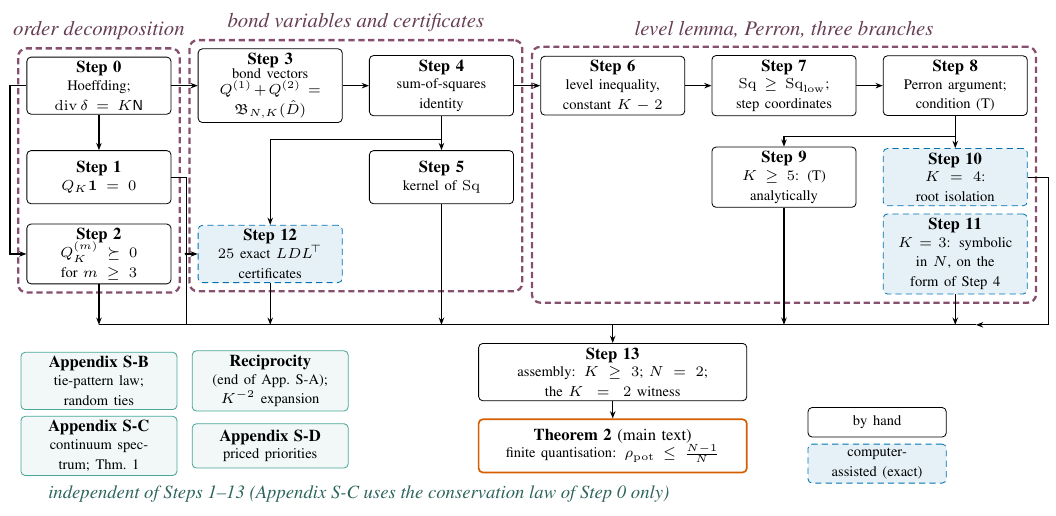}
\vspace{-8mm}
\caption{Logical structure of the supplementary material. Boxes are the fourteen steps of the proof of
Theorem~\ref{thm:finiteK} (Appendix~\ref{app:finiteK}); arrows point from a step to the steps that use it, and the
horizontal collector above Step~\ref{step:assembly} gathers everything the assembly needs: the higher orders
(Step~\ref{step:orders}), the kernel of $\mathrm{Sq}$ (Step~\ref{step:kerSq}), the three branches
(Steps~\ref{step:K5}--\ref{step:K3}) and the certificates (Step~\ref{step:certs}), which together cover every pair
$(K,N)$ with $K\ge3$; Step~\ref{step:B} feeds both the kernel statement and the image certificates. Dashed boxes are
the three computer-assisted steps; all computations are exact. Section~\ref{sec:recip} and
Appendices~\ref{sec:pattern}--\ref{sec:priced} are independent of the proof and of one another, except that
Appendix~\ref{sec:cont} uses the conservation law of Step~\ref{step:hoeff}.}\label{fig:dependency}
\end{figure}

\paragraph{What depends on what.} Figure~\ref{fig:dependency} shows the logical structure of this document. The four
appendices are independent of one another and depend on the main text only through Lemma~\ref{lem:numop} and the
definitions of Section~\ref{sec:setup}, except that Appendix~\ref{sec:cont} uses Lemma~\ref{lem:cons} of
Step~\ref{step:hoeff}; Section~\ref{sec:recip} is not used in the proof of Theorem~\ref{thm:finiteK}.

\paragraph{Conventions.} Notation is that of the main text. Throughout, $[N]=\{0,\dots,N-1\}$; functions on the
profile space $S$ carry the uniform probability measure and $\|g\|^2=\E[g^2]$; flows carry the counting inner
product $\langle x,y\rangle=\sum_ex(e)y(e)$; $A_i$ is the average over the $i$-th coordinate, $g_T$ the Hoeffding
component of $g$ on $T\subseteq[N]$, and $[u_i]_T$ that of the payoff $u_i$. The bond weights of Appendix~\ref{app:finiteK} are
$\mu_k=\frac{N-1-k}N$ and $\nu_k=\frac{k+1}N$, as in the main text; $\beta$ is reserved for the rationality parameter
and appears only in Appendix~\ref{sec:priced}. In Appendix~\ref{app:finiteK}
``Step~$k$'' refers to one of the fourteen steps of the proof of Theorem~\ref{thm:finiteK}, which are numbered in
reading order; parts of a single proof are labelled (1), (2), \dots\ and are local to that proof. Three kinds of
statement occur and are marked as such: results proved by hand; polynomial identities in $N$ established by computer
algebra, whose scripts are named at the point of use; and finitely many exact rational certificates. Nothing in any
proof depends on a floating-point computation.

\section{Proof of Theorem~\ref{thm:finiteK} (finite quantisation)}\label{app:finiteK}
Throughout this appendix the conflict graph is complete, $C=K_N$, the protocol is the index protocol, and
$u_i=f(\pos_i)$ for an arbitrary payoff-from-rank function $f$ with increments $F_j=f(j-1)-f(j)$; the factor $K^{N}$
between the expectation norm of functions and the counting norm of flows appears explicitly wherever one is
converted into the other. The proof is organised in fourteen steps, numbered in the order in which they are read;
Appendix~\ref{app:B} of the main text follows the four stages of Section~\ref{sec:finiteK} and cites the steps each
consists of.

\paragraph{Roadmap.} Table~\ref{tab:roadmap} lists the steps, what each proves, where its result is used, and how it
is verified. A reader interested only in the logical structure can read Steps~\ref{step:hoeff}--\ref{step:kerSq}
and~\ref{step:assembly} and take Table~\ref{tab:cov} on trust; a reader auditing the computer-assisted parts needs
Steps~\ref{step:K4}--\ref{step:certs} and Appendix~\ref{app:repro}.

\begin{table}[H]\centering\small
\caption{Roadmap of the proof of Theorem~\ref{thm:finiteK}. ``exact'' means rational or symbolic arithmetic.}
\label{tab:roadmap}
\begin{tabular}{@{}>{\raggedright\arraybackslash}p{0.06\linewidth}>{\raggedright\arraybackslash}p{0.40\linewidth}>{\raggedright\arraybackslash}p{0.22\linewidth}>{\raggedright\arraybackslash}p{0.26\linewidth}@{}}\toprule
Step & proves & used in & verified by\\\midrule
\ref{step:hoeff} & Laplacian $\dv\delta=K\numop$, conservation, order decomposition $Q_K=K^{N+1}\sum_mQ^{(m)}_K$ (Lemmas~\ref{lem:lap}--\ref{lem:orderdec}) & every later step & by hand\\
\ref{step:B} & $Q_K\one=0$: the linear payoff is an exact null vector (Lemma~\ref{lem:B}) & kernel statement, Step~\ref{step:assembly} & by hand\\
\ref{step:orders} & $\|\sum_{i\in T}[u_i]_T\|^2\le2\sum_i\|[u_i]_T\|^2$, hence $Q^{(m)}_K\succeq0$ for $m\ge3$ (Proposition~\ref{prop:orders}, Lemma~\ref{lem:struct}) & reduces the theorem to $Q^{(1)}+Q^{(2)}\succeq0$ & by hand\\
\ref{step:bond} & bond vectors $\hat D_k(f)$; first- and second-order components in bond form (Propositions~\ref{prop:first}, \ref{prop:second}); $Q^{(1)}+Q^{(2)}=\fB_{N,K}(\hat D(f))$ (Proposition~\ref{thm:C}); injectivity on affine $f$ (Lemma~\ref{lem:transfer}) & Steps~\ref{step:SOS}--\ref{step:certs} work on $\fB$ alone & by hand\\
\ref{step:SOS} & Hahn--Poincar\'e and spread identities; the sum-of-squares identity $K^2\fB=\sum_kZ_k+\mathrm{Sq}-\mathrm{Neg}$ (Theorem~\ref{thm:SOS}) & Steps~\ref{step:kerSq}--\ref{step:K3} & by hand\\
\ref{step:kerSq} & $\mathrm{Sq}(D)=0$ iff $D\in\operatorname{span}(\one\otimes\one)$ (Lemma~\ref{lem:kerSq}) & kernel statement, Step~\ref{step:assembly} & by hand\\
\ref{step:Lstar} & sharp level inequality with constant $K-2$ (Lemma~\ref{lem:Lstar-SM}); the level pencil $\varrho(\mu)$ (Proposition~\ref{cor:rho}) & Steps~\ref{step:Sqlow}--\ref{step:K4} & by hand (Schur product)\\
\ref{step:Sqlow} & lower bound $\mathrm{Sq}\ge\mathrm{Sq}_{\rm low}$ in the bond vectors alone; step coordinates (Propositions~\ref{prop:Sqlow}, \ref{prop:hadamard}) & Step~\ref{step:perron} & by hand\\
\ref{step:perron} & Perron argument: $\fB\succeq0$ if condition (T) holds on every bond (Lemmas~\ref{lem:path}, \ref{lem:Zmatrix}, Proposition~\ref{prop:T}) & Steps~\ref{step:K5}--\ref{step:K4} & by hand ($M$-matrix criterion)\\
\ref{step:K5} & $K\ge5$: (T) from Lemma~\ref{lem:Lstar-SM} whenever $(N-2)(K-4)\ge4$ (Theorem~\ref{thm:K5}) & assembly & by hand\\
\ref{step:K4} & $K=4$: root isolation of the cubic for $\varrho(\mu)$, $\varrho\ge2\mu\nu+5(\mu\nu)^3$, hence (T) for $N\ge10$ (Lemma~\ref{lem:K4root}, Theorem~\ref{thm:K4}) & assembly & Sturm sequences; \texttt{audit\_\allowbreak K4.py} (exact)\\
\ref{step:K3} & $K=3$: block representation, shear, $C-\gamma I\succeq0$, Perron certificate for $N\ge17$ (Lemmas~\ref{lem:K3u}--\ref{lem:K3M}, Theorem~\ref{thm:K3}) & assembly & computer algebra, symbolic in $N$: \texttt{k3\_\allowbreak representation\_\allowbreak proof.py}, \texttt{s1\_\allowbreak closed.py}, \texttt{verify\_\allowbreak M\_\allowbreak symbolic.py} (exact)\\
\ref{step:certs} & the remaining pairs $(K,N)$ by exact $LDL^\top$ certificates on the bond space or on its image (Table~\ref{tab:cov}) & assembly & $25$ stored certificates; \texttt{verify\_\allowbreak certificates.py} (exact)\\
\ref{step:assembly} & assembly; $N=2$; the counterexample at $K=2$, $N=7$ & Theorem~\ref{thm:finiteK} & by hand; \texttt{k2\_\allowbreak witness.py} (exact)\\\bottomrule
\end{tabular}
\end{table}

\stepsection{hoeff}{Hoeffding decomposition, the number operator and conservation}\label{sec:hoeff}

This step fixes conventions and records, in the form used below, the facts behind Lemma~\ref{lem:numop} of the main
text. Throughout, $S=\{0,\dots,K-1\}^N$ carries the uniform probability measure and $\|g\|^2=\E[g^2]$ for functions
on~$S$; flows carry the counting inner product $\langle x,y\rangle=\sum_ex(e)y(e)$. For a function $g$ on~$S$ and a
coordinate~$i$, $A_i$ denotes averaging over~$s_i$. The Hoeffding
(ANOVA) decomposition is $g=\sum_{T\subseteq[N]}g_T$ with
$g_T=\prod_{j\in T}(I-A_j)\prod_{j\notin T}A_j\,g$; $g_T$ depends only on
$s_T$, has zero mean in each of its variables, and the $g_T$ are
mutually orthogonal. We write $g_m=\sum_{|T|=m}g_T$. The
\emph{number operator} $\numop=\sum_i(I-A_i)$ acts on $g_T$ as
multiplication by~$|T|$. On the complete graph we abbreviate
$\pos_i=\pos_{K_N}(i,s)$.

\begin{lemmaS}[Laplacian of the game graph]\label{lem:lap}
Let $\delta$ be the gradient and let $\delta^*$ be its adjoint. With the
divergence convention $(\dv x)(s)=\sum_{e\,\mathrm{into}\,s}x(e)-\sum_{e\,\mathrm{out\ of}\,s}x(e)$
one has $\langle\delta\phi,x\rangle=\langle\phi,\dv x\rangle$, i.e.\ $\delta^*=\dv$. Then
$\delta^*\delta=K\numop$; in particular $\delta^*\delta g_T=K|T|\,g_T$.
For the game flow $u$ of payoffs $(u_i)$,
\begin{equation}\label{eq:divS}
(\dv u)(s)=K\sum_i\big[u_i(s)-\bar u_i(s_{-i})\big],\qquad
\bar u_i:=A_iu_i,
\end{equation}
and the potential energy is
\begin{equation}\label{eq:specS}
\|\delta\phi^\star\|^2=K^{N}\sum_{m=1}^N\frac{\|(\dv u)_m\|^2}{mK}
=K^{N+1}\sum_{m=1}^N\frac{\|g_m\|^2}{m},\qquad g:=\sum_i(u_i-A_iu_i),
\end{equation}
which is the first identity of \eqref{eq:twoenergies} in the main text.
\end{lemmaS}
\begin{proof}
This is Lemma~\ref{lem:numop} of the main text, whose proof identifies $\dv\delta=K\numop$ and $\dv u=Kg$; the sum
over the $K^N$ profiles in $\langle\phi^\star,\dv\delta\phi^\star\rangle$ supplies the factor $K^N$ in
\eqref{eq:specS}, and diagonalising $\numop$ on the Hoeffding components gives the two forms.
\end{proof}

\begin{lemmaS}[conservation]\label{lem:cons}
For every conflict graph and every payoff-from-rank function~$f$,
$\sum_iu_i(s)$ is independent of~$s$ on the complete graph, and
$\sum_i\pos_C(i,s)=|E|$ on every graph. Consequently, on the complete
graph, $\sum_i[u_i]_T=0$ for every nonempty~$T$, hence
$g_T=\sum_{i\in T}[u_i]_T=-\sum_{i\notin T}[\bar u_i]_T$.
\end{lemmaS}
\begin{proof}
On $K_N$ the map $i\mapsto\pos_i$ is a permutation of
$\{0,\dots,N-1\}$ (the order by $(s_i,-i)$ is strict), so
$\sum_if(\pos_i)=\sum_rf(r)$. On a general graph each conflict edge
contributes exactly one unit to the position of whichever endpoint is
served second. A constant has no nonconstant Hoeffding components, which
gives $\sum_i[u_i]_T=0$; finally $[u_i]_T$ with $i\notin T$ equals
$[\bar u_i]_T$ because $[u_i]_T$ contains the factor $A_i$.
\end{proof}

\begin{lemmaS}[order decomposition]\label{lem:orderdec}
For every $f$,
\begin{equation}\label{eq:P2}
\|u\|^2=K^{N+1}\sum_i\sum_{T\ni i}\|[u_i]_T\|^2,\qquad
\|Pu\|^2=K^{N+1}\sum_{T\ne\emptyset}\frac1{|T|}\Big\|\sum_{i\in T}[u_i]_T\Big\|^2 .
\end{equation}
Consequently $F^\top Q_KF=K^{N+1}\sum_{m\ge1}Q^{(m)}_K(f)$ with
\begin{equation}\label{eq:Qm-SM}
Q^{(m)}_K(f)=\sum_{|T|=m}\Big[\frac{N-1}N\sum_{i\in T}\|[u_i]_T\|^2-\frac1m\Big\|\sum_{i\in T}[u_i]_T\Big\|^2\Big],
\end{equation}
which is the order decomposition used in stage~1) of the outline in the main text.
\end{lemmaS}
\begin{proof}
Since $I-A_i$ annihilates the Hoeffding components on sets not containing $i$ and fixes the others,
$u_i-A_iu_i=\sum_{T\ni i}[u_i]_T$, an orthogonal sum; this gives the first identity by Lemma~\ref{lem:numop}. Summing
over $i$ and collecting by index set, $g=\sum_i(u_i-A_iu_i)$ has components $g_T=\sum_{i\in T}[u_i]_T$, and the second
identity is again Lemma~\ref{lem:numop}. The last statement is the definition
$Q_K=\frac{N-1}NM_{\mathrm{tot}}-M_\pot$ with $F^\top M_{\mathrm{tot}}F=\|u\|^2$ and $F^\top M_\pot F=\|Pu\|^2$
(flow norms).
\end{proof}

\stepsection{B}{proof of Lemma~\ref{lem:B} (exact eigenvector)}\label{sec:lemA}

With the order decomposition in place, the first thing to establish is that the linear payoff is an exact null
vector of $Q_K$ at every finite $K$. The result is used twice: for the kernel statement of
Theorem~\ref{thm:finiteK} (Step~\ref{step:assembly}), and inside the certificates of Step~\ref{step:certs}, which need to know that the bond
form annihilates the linear payoff.

\begin{mainlem}[\ref{lem:B} of the main text]
For every $N\ge2$ and $K\ge2$, $Q_K\one=0$, i.e.\
$M_\pot\one=\frac{N-1}NM_{\mathrm{tot}}\one$.
\end{mainlem}

Players are indexed $0,\dots,N-1$, so player $i$ has exactly $i$
opponents of smaller index. Write $u^{\rm lin}$ for the flow of the linear
game and $u^f$ for that of $\mathcal G_f(K_N,K)$.

\emph{(1) Reduction.}
By definition of the Gram matrices, $Q_K\one=0$ is equivalent to
$\langle Pu^{\rm lin}-\frac{N-1}Nu^{\rm lin},\,u^f\rangle=0$ for all $f$.
By Proposition~\ref{prop:energies} with $d_i\equiv N-1$, $Pu^{\rm lin}=\delta\phi$ with
$\phi=\frac{N-1}K\sum_is_i$, which on the edges of player~$i$ is the
flow of the payoff $(N-1)\phi_1(s_i)$, $\phi_1(a):=\frac{a-\bar s}K$,
$\bar s=\frac{K-1}2$. Writing $-\pos_i=\sum_{j\ne i}\psi_{ij}(s_i,s_j)$ with
$\psi_{ij}(a,b)=-\one[b>a]$ ($j>i$) or $-\one[b\ge a]$ ($j<i$), the
Hoeffding decomposition of each kernel is
$\psi_{ij}=\bar\psi_{ij}+\phi_1(a)-\phi_1(b)+\Gamma_{ij}(a,b)$ (both kernels have
the same first-order parts $\phi_1(a)$ and $-\phi_1(b)$; this is checked
directly: $\E_b\psi_{ij}(a,b)$ equals $-\frac{K-1-a}K$ or
$-\frac{K-a}K$, both equal to $\frac{a-\bar s}K$ up to a constant).
Hence $u^{\rm lin}$ is the flow of $(N-1)\phi_1(s_i)+\Gamma_i$ with
$\Gamma_i:=\sum_{j\ne i}\Gamma_{ij}(s_i,s_j)$, and the harmonic part
$h=u^{\rm lin}-\delta\phi$ is the flow of $\Gamma_i$. For two payoff
families $(w_i)$, $(v_i)$ the inner product of their flows is
\begin{equation}\label{eq:flowip}
\langle\text{flow}(w),\text{flow}(v)\rangle
=\sum_i\sum_{s_{-i}}\sum_{a<b}\big(w_i(b)-w_i(a)\big)\big(v_i(b)-v_i(a)\big)
=K^{N+1}\sum_i\E\big[\Cov(w_i,v_i\mid s_{-i})\big],
\end{equation}
since $\sum_{a<b}(x_b-x_a)(y_b-y_a)=K^2\Cov(x,y)$ for the uniform law on
$\{0,\dots,K-1\}$. Because $\Gamma_i$ and $\phi_1(s_i)$ are centred in
$s_i$ for every $s_{-i}$, the conditional covariances reduce to plain
expectations, and the claim becomes
\begin{equation}\label{eq:B1}
\sum_i\E[\Gamma_i\,v_i]=\sum_i\E[\phi_1(s_i)\,v_i]\qquad\text{for all }f,
\end{equation}
where $v_i=f(\pos_i)$.

\emph{(2) Using conservation.}
Let $\Xi_i=\sum_{j\ne i}(-\phi_1(s_j))$ be the second-variable first-order
part, so that $-\pos_i=-\E\pos_i+\allowbreak(N-1)\phi_1(s_i)+\Xi_i+\Gamma_i$. By Lemma~\ref{lem:cons},
$\sum_i\E[\Xi_iv_i]=\allowbreak-\sum_j\E\big[\phi_1(s_j)\sum_{i\ne j}v_i\big]
=-\sum_j\E\big[\phi_1(s_j)(\mathrm{const}-v_j)\big]=\sum_j\E[\phi_1(s_j)v_j]$.
Therefore
$\sum_i\E[\Gamma_iv_i]=\allowbreak\sum_i\E[(\E\pos_i-\pos_i)v_i]-N\sum_i\E[\phi_1(s_i)v_i]$,
and \eqref{eq:B1} is equivalent to
\begin{equation}\label{eq:B2}
\sum_i\E[\pos_i]\,\E[f(\pos_i)]-\sum_{r}r f(r)
=\frac{N+1}K\sum_i\Cov\big(s_i,f(\pos_i)\big),
\end{equation}
where we used $\sum_i\E[\pos_if(\pos_i)]=\sum_rrf(r)$ (Lemma~\ref{lem:cons}).

\emph{(3) Generating functions.}
Both sides of \eqref{eq:B2} are linear in $f$, so it suffices to take
$f(r)=x^r$. Put $y=x-1$. Given $s_i$, write $\hat a=K-1-s_i$; opponent $j<i$
is served before $i$ with probability $\frac{\hat a+1}K$ and opponent $j>i$
with probability $\frac{\hat a}K$, independently across~$j$. Let
$\mathsf B_{\hat a}=1+\frac{\hat a}Ky$ and $\mathsf A_{\hat a}=1+\frac{\hat a+1}Ky=\mathsf B_{\hat a+1}$. Then
\[
\Pi_i(x):=\E[x^{\pos_i}]=\frac1K\sum_{\hat a=0}^{K-1}\mathsf A_{\hat a}^{\,i}\mathsf B_{\hat a}^{\,N-1-i},\qquad
\mu_i:=\E[\pos_i]=\frac{(N-1)\frac{K-1}2+i}K .
\]
With $\sum_{i=0}^{N-1}\mathsf A^i\mathsf B^{N-1-i}=\frac{\mathsf A^N-\mathsf B^N}{\mathsf A-\mathsf B}$, $\mathsf A_{\hat a}-\mathsf B_{\hat a}=\frac yK$,
telescoping ($\mathsf A_{\hat a}=\mathsf B_{\hat a+1}$, $\mathsf B_0=1$, $\mathsf B_K=1+y$), $\mathsf S:=\sum_{\hat a=0}^K\mathsf B_{\hat a}^N$
and $G(x):=\sum_{r=0}^{N-1}x^r=\frac{(1+y)^N-1}y$:
\begin{align*}
\sum_i\Pi_i&=\frac1y\sum_{\hat a=0}^{K-1}(\mathsf A_{\hat a}^N-\mathsf B_{\hat a}^N)=G,\\
\sum_ii\,\Pi_i&=\frac1y\sum_{\hat a=0}^{K-1}(N\mathsf A_{\hat a}^N+\mathsf B_{\hat a}^N)-\frac{K\big((1+y)^{N+1}-1\big)}{y^2}\\
 &=\frac{(N+1)\mathsf S-N-(1+y)^N}{y}-\frac{K\big((1+y)^{N+1}-1\big)}{y^2},\\
Z(x)&:=\sum_i\E[s_ix^{\pos_i}]=\frac1y\sum_{\hat a=0}^{K-1}(K-1-\hat a)(\mathsf A_{\hat a}^N-\mathsf B_{\hat a}^N)\\
 &=\frac{\mathsf S-K-(1+y)^N}{y}.
\end{align*}
(The second line uses $\sum_ii\,\mathsf A^i\mathsf B^{N-1-i}=\mathsf A\partial_{\mathsf A}\frac{\mathsf A^N-\mathsf B^N}{\mathsf A-\mathsf B}$;
the third is a summation by parts: with $\mathsf A_{\hat a}=\mathsf B_{\hat a+1}$,
$\sum_{\hat a=0}^{K-1}(K-1-\hat a)(\mathsf B_{\hat a+1}^N-\mathsf B_{\hat a}^N)=\sum_{\hat a=1}^{K-1}\mathsf B_{\hat a}^N-(K-1)\mathsf B_0^N
=\mathsf S-\mathsf B_0^N-\mathsf B_K^N-(K-1)$, and $\mathsf B_0^N=1$, $\mathsf B_K^N=(1+y)^N$.) The generating-function form of
\eqref{eq:B2} is $\sum_i\mu_i\Pi_i-xG'=\frac{N+1}K\big(Z-\frac{K-1}2G\big)$.
Substituting, $\mathsf S$ and $K$ cancel and the difference of the two sides equals
\begin{align*}
NG+\frac{N+1}y-\frac{(1+y)^{N+1}-1}{y^2}-xG'
&=\frac{Ny(1+y)^N+y-(1+y)^{N+1}+1}{y^2}\\
&\quad-\frac{Ny(1+y)^N-(1+y)^{N+1}+1+y}{y^2}=0 .
\end{align*}
This proves \eqref{eq:B2}, hence Lemma~\ref{lem:B}. \qed

\begin{remarkS}
The quantisation $K$ disappears in part (3): the identity only uses that,
conditionally on $s_i$, the opponents are served before $i$ independently,
and that positions form a permutation. 
\end{remarkS}

\stepsection{orders}{proof of Proposition~\ref{prop:orders} (higher orders)}\label{sec:orders}
The order decomposition of Lemma~\ref{lem:orderdec} splits $Q_K$ into a sum over Hoeffding orders $m$. This step
disposes of every order $m\ge3$ at once, by an argument that uses nothing about $f$; what remains afterwards is the
two-body form $Q^{(1)}_K+Q^{(2)}_K$, which occupies Steps~\ref{step:bond}--\ref{step:certs}.

\begin{propS}[higher orders]\label{prop:orders}
For all $N\ge3$, $K\ge2$ and $m\ge3$, $Q^{(m)}_K\succeq0$, where
\[
Q^{(m)}_K(f)=\frac{N-1}N\sum_{|T|=m}\sum_{i\in T}\|[u_i]_T\|^2-\frac1m\sum_{|T|=m}\Big\|\sum_{i\in T}[u_i]_T\Big\|^2 .
\]
More precisely, for every $T$ with $|T|=m$,
$\big\|\sum_{i\in T}[u_i]_T\big\|^2\le2\sum_{i\in T}\|[u_i]_T\|^2$, and the
constant~$2$ cannot be improved at $K=2$.
\end{propS}

\begin{proof}
Fix $T$ and write $X_i=[u_i]_T=(I-A_i)F_i$ with $F_i=\Psi^{(i)}D^{(i)}$ as in Lemma~\ref{lem:struct} below, whose proof does not use this proposition. Three facts.

\emph{(i) Cross terms.} For $i\ne k$ in $T$, $F_k$ is centred in $s_i$ (it
contains the factor $\theta(s_i;\cdot)-p$), so $A_iF_k=0$ and, $A_i$ being
an orthogonal projection, $\langle A_iF_i,F_k\rangle=\langle A_iF_i,A_iF_k\rangle=0$;
similarly the other cross projections vanish, whence
$\langle X_i,X_k\rangle=\langle F_i,F_k\rangle$.

\emph{(ii) The cross kernel is entrywise nonpositive.} For $i<k$,
conditioning on $s_i=a$, $s_k=b$ and averaging over the remaining
coordinates,
\begin{gather*}
\langle F_i,F_k\rangle=\frac1{K^2}\sum_{a,b}D^{(i)}(a)D^{(k)}(b)\,\Omega_{ik}(a,b),\\
\Omega_{ik}(a,b)=-\big(\one[b>a]-p(a)\big)\big(\one[b>a]-q(b)\big)\prod_{j\ne i,k}\dk\big(t^{(i)}_j(a),t^{(k)}_j(b)\big),
\end{gather*}
with $q(b)=b/K$ and $\dk(t,t')=p(\max(t,t'))-p(t)p(t')\ge0$. Indeed the
factor of $\Psi^{(i)}$ involving $s_k$ is $\one[b>a]-p(a)$ (as $k>i$), and
the factor of $\Psi^{(k)}$ involving $s_i$ is
$\one[a\ge b]-p(b-1)=-(\one[b>a]-q(b))$ (as $i<k$, $p(b-1)=1-q(b)$): the two
indicators are complementary, so the product is $-(1-p)(1-q)\le0$ if $b>a$
and $-pq\le0$ if $b\le a$. The remaining factors are covariances of two
increasing indicators of the same variable, hence nonnegative.

\emph{(iii) Diagonal terms.} $\|X_i\|^2=\|F_i\|^2-\|A_iF_i\|^2$, where
$\|F_i\|^2$ depends on $D^{(i)}$ only through $D^{(i)}(a)^2$, and
$\|A_iF_i\|^2=\frac1{K^2}\sum_{a,a'}D^{(i)}(a)D^{(i)}(a')\Omega^{\rm d}_i(a,a')$ with
$\Omega^{\rm d}_i(a,a')=\prod_{j\ne i}\dk(t^{(i)}_j(a),t^{(i)}_j(a'))\ge0$.

Let $\tilde D^{(i)}=|D^{(i)}|$ and $\tilde F_i,\tilde X_i$ the corresponding
vectors. By (i)--(ii),
\begin{align*}
\sum_{i\ne k}\langle X_i,X_k\rangle
&\le\frac1{K^2}\sum_{i\ne k}\sum_{a,b}|D^{(i)}(a)|\,|D^{(k)}(b)|\,|\Omega_{ik}(a,b)|
=-\sum_{i\ne k}\langle\tilde F_i,\tilde F_k\rangle\\
&=\sum_i\|\tilde X_i\|^2-\Big\|\sum_i\tilde X_i\Big\|^2\le\sum_i\|\tilde X_i\|^2 .
\end{align*}
By (iii), $\|\tilde X_i\|^2\le\|X_i\|^2$. Therefore
\[
\Big\|\sum_iX_i\Big\|^2=\sum_i\|X_i\|^2+\sum_{i\ne k}\langle X_i,X_k\rangle
\le2\sum_i\|X_i\|^2 .
\]
Summing over $T$,
\[
\sum_{|T|=m}\|g_T\|^2\le2\sum_{|T|=m}\sum_{i\in T}\|[u_i]_T\|^2,\qquad\text{so}\qquad
Q^{(m)}_K\ge\Big(\frac{N-1}N-\frac2m\Big)\sum_{|T|=m}\sum_{i\in T}\|[u_i]_T\|^2\ge0
\]
as soon as $m\ge\frac{2N}{N-1}$, i.e.\ $m\ge3$ for $N\ge3$. For $K=2$ the
interior members of $T$ have $X_i=0$ (their $\Psi^{(i)}$ vanishes at both
levels) and the two extreme members are parallel; for $f$ chosen so that they have equal norm and the same
orientation the inequality is an equality, so the constant~$2$ cannot be improved.
\end{proof}

In Proposition~\ref{prop:orders}, $Q^{(m)}_K$ is viewed as a quadratic form in the increments~$F$ of~$f$; since
$F^\top Q_KF=K^{N+1}\sum_mQ^{(m)}_K(f)$ by Lemma~\ref{lem:orderdec}, positive-semidefiniteness statements about the
$Q^{(m)}_K$ transfer to $Q_K$ unchanged.

\begin{lemmaS}[structure of the components]\label{lem:struct}
Let $T=\{i_1<\dots<i_m\}$ and $i\in T$. Put $t^{(i)}_j(a)=a$ for $j>i$ and
$a-1$ for $j<i$, $\theta(x;t)=\one[x>t]$, $p(t)=\frac{K-1-t}K$ (with
$p(-1)=1$). Then
\[
[u_i]_T=(I-A_i)F_i,\qquad
F_i(s_T)=\Psi^{(i)}(s_T)\,D^{(i)}(s_i),\qquad
\Psi^{(i)}=\prod_{j\in T\setminus i}\big(\theta(s_j;t^{(i)}_j(s_i))-p(t^{(i)}_j(s_i))\big),
\]
with $D^{(i)}(a)=\E\big[\Delta^{m-1}f(B^{(a)}_{\notin T})\big]$, where
$B^{(a)}_{\notin T}$ is the number of players outside $T$ served before~$i$
when $s_i=a$.
\end{lemmaS}
\begin{proof}
We have $\pos_i=\sum_{j\in T\setminus i}\theta(s_j;t^{(i)}_j(s_i))+B_{\notin T}$.
Averaging over the coordinates outside $T$ gives
$\E f\big(\sum_{j\in T\setminus i}\theta_j+B\big)$. For $j\in T\setminus i$
and any function $\zeta$, $\zeta(\theta_j)=\zeta(0)+\theta_j\Delta\zeta(0)$ because
$\theta_j\in\{0,1\}$, so $(I-A_j)\zeta(\theta_j)=(\theta_j-p_j)\Delta\zeta(0)$;
applying this successively to all $j\in T\setminus i$ produces the product
$\Psi^{(i)}$ and the $(m-1)$-fold difference, and $(I-A_i)$ is the last
factor.
\end{proof}
%=====================================================================================================
\stepsection{bond}{bond vectors and the two-body reduction (first part of Proposition~\ref{prop:bond})}\label{sec:bond}
The two-body form $Q^{(1)}_K+Q^{(2)}_K$ is a quadratic form in the $N-1$ increments of $f$. This step rewrites it as a
quadratic form $\fB_{N,K}$ in $N-1$ \emph{bond vectors} $\hat D_k(f)\in\R^K$, one for each pair of adjacent positions,
whose entries are averaged first differences of $f$. The point of the change of variables is that $\fB_{N,K}$ is
defined on all of $(\R^K)^{N-1}$, not only on the image $\{\hat D(f)\}$, and positivity on the larger space is what
Steps~\ref{step:SOS}--\ref{step:K3} prove; Lemma~\ref{lem:transfer} at the end of this step brings the kernel back to $f$.

\paragraph{Notation for the rest of the appendix.} Priority levels are $a\in\{0,\dots,K-1\}$ and carry the uniform probability
measure; $\Var$ and $\Cov$ of functions of the level refer to it. \emph{Bonds} are indexed by $k\in\{0,\dots,N-2\}$. For
$D\in\R^K$ we write
\[
(J_1D)(a)=\sum_{x\le a}D(x),\qquad (J_0D)(a)=\sum_{x<a}D(x),
\]
for the inclusive and exclusive cumulative sums, written without the operators in the main text; given bond vectors
$D_0,\dots,D_{N-2}\in\R^K$ we put $G_k=J_1D_k$, $H_k=J_0D_k$, so that $G_k-H_k=D_k$, and define the \emph{vertex
functions}
\begin{equation}\label{eq:Omega}
\Omega_j:=\sum_{k<j}H_k+\sum_{k\ge j}G_k,\qquad j=0,\dots,N-1 ,
\end{equation}
which satisfy $\Omega_k-\Omega_{k+1}=D_k$. Further $\mu_k=\frac{N-1-k}N$, $\nu_k=\frac{k+1}N$ (so that $\nu_k+\mu_l=1-\frac{l-k}N$ for $k\le l$; in particular $\mu_k+\nu_k=1$), and
\begin{equation}\label{eq:Hahn}
\calD(g)=\sum_{y=1}^{K-1}\frac{y(K-y)}K\big(g(y)-g(y-1)\big)^2,\qquad P(g):=\calD(g)-\Var(g),
\end{equation}
is the Hahn Dirichlet form. For a payoff-from-rank function $f$, with $\Delta f(r)=f(r+1)-f(r)$, the \emph{bond vectors
of $f$} are
\begin{equation}\label{eq:Dhat}
\hat D_k(a)=\E\Big[\Delta f\Big(\Bin\big(k,\tfrac{K-a}K\big)+\Bin\big(N-2-k,\tfrac{K-1-a}K\big)\Big)\Big],
\end{equation}
the two binomials being independent. The \emph{bond form} is the quadratic form on $(\R^K)^{N-1}$
\begin{equation}\label{eq:formula1}
K^2\fB_{N,K}(D):=\underbrace{\frac{N-2}{2}\sum_k\big[\mu_kP(H_k)+\nu_kP(G_k)\big]}_{\mathrm I}
+\underbrace{\sum_{k\le l}\Cov(H_k,G_l)}_{\mathrm{II}}
-\underbrace{\frac1N\sum_{j=0}^{N-1}\Var(\Omega_j)}_{\mathrm{III}} .
\end{equation}

\begin{lemmaS}[what an opponent sees]\label{lem:groups}
Conditionally on $s_i=a$, the events ``$j$ is served before $i$'', $j\ne i$, are independent; their probability is
$\frac{K-a}K$ for $j<i$ (the event $s_j\ge a$) and $\frac{K-1-a}K$ for $j>i$ (the event $s_j>a$). Moreover, for any
random variable $R$ with values in $\{0,\dots,N-2\}$, independent of the Bernoulli variables below, and any
$p\in[0,1-\frac1K]$,
\begin{equation}\label{eq:oneflip}
\E f\big(R+\mathrm{Bern}(p+\tfrac1K)\big)-\E f\big(R+\mathrm{Bern}(p)\big)=\tfrac1K\,\E\,\Delta f(R).
\end{equation}
\end{lemmaS}
\begin{proof}
The first statement is the definition \eqref{eq:pos} of the index protocol together with the independence of the
levels. For the second, $\E f(R+\mathrm{Bern}(p))=\E f(R)+p\,\E\Delta f(R)$ is affine in $p$.
\end{proof}

\begin{propS}[first order]\label{prop:first}
Let $\varphi_j(a):=\E[u_j\mid s_j=a]$ and let $\Omega_j$ be built from the bond vectors \eqref{eq:Dhat} of $f$. Then
$\varphi_j=-\frac1K\Omega_j+\mathrm{const}$, and
\begin{equation}\label{eq:Q1}
Q^{(1)}_K(f)=-\frac1N\sum_{j=0}^{N-1}\Var(\varphi_j)=-\frac1{NK^2}\sum_{j=0}^{N-1}\Var(\Omega_j).
\end{equation}
\end{propS}
\begin{proof}
By Lemma~\ref{lem:groups}, given $s_j=a$ the position of $j$ is a sum of $j$ independent Bernoulli variables of
parameter $\frac{K-a}K$ (the opponents of smaller index) and $N-1-j$ of parameter $\frac{K-1-a}K$. Passing from level
$a$ to level $a-1$ raises every parameter by $\frac1K$; we do this one opponent at a time, in two stages, applying
\eqref{eq:oneflip} at each flip with $R$ the contribution of the other $N-2$ opponents.
\emph{Stage (i):} raise the $N-1-j$ opponents of larger index from $\frac{K-1-a}K$ to $\frac{K-a}K$. When $m$ of the
other $N-2$ opponents are at $\frac{K-a}K$ and $N-2-m$ at $\frac{K-1-a}K$, the increment is $\frac1K\hat D_m(a)$ by
\eqref{eq:Dhat}; here $m$ runs from $j$ to $N-2$.
\emph{Stage (ii):} all opponents are now at $\frac{K-a}K=\frac{K-1-(a-1)}K$; raise the $j$ opponents of smaller index
to $\frac{K-(a-1)}K$. When $m$ of the others have been raised the increment is $\frac1K\hat D_m(a-1)$, $m=0,\dots,j-1$.
Hence
\[
\varphi_j(a-1)-\varphi_j(a)=\frac1K\Big[\sum_{m\ge j}\hat D_m(a)+\sum_{m<j}\hat D_m(a-1)\Big]
=\frac1K\big[\Omega_j(a)-\Omega_j(a-1)\big],
\]
because $G_m(a)-G_m(a-1)=\hat D_m(a)$ and $H_m(a)-H_m(a-1)=\hat D_m(a-1)$. This gives the first claim. For
\eqref{eq:Q1}, $[u_j]_{\{j\}}=\varphi_j-\E\varphi_j$ and, in \eqref{eq:Qm-SM} with $m=1$,
$\frac{N-1}N\|[u_j]_{\{j\}}\|^2-\|[u_j]_{\{j\}}\|^2=-\frac1N\Var(\varphi_j)$.
\end{proof}

\begin{propS}[second order]\label{prop:second}
Let $i<j$, $X_i:=[u_i]_{\{i,j\}}$ and $X_j:=[u_j]_{\{i,j\}}$, and let $H_k,G_k$ be built from the bond vectors of $f$.
Then
\begin{equation}\label{eq:second}
\|X_i\|^2=\frac{P(H_i)}{K^2},\qquad \|X_j\|^2=\frac{P(G_{j-1})}{K^2},\qquad
\langle X_i,X_j\rangle=-\frac{\Cov(H_i,G_{j-1})}{K^2}.
\end{equation}
\end{propS}
\begin{proof}
Write $a=s_i$, $b=s_j$, let $A_a,A_b$ be the averaging operators in these coordinates and $C_a=I-A_a$, $C_b=I-A_b$; they
are commuting orthogonal projections. Put $p(t)=\frac{K-1-t}K$, so that $\E_b\one[b>t]=p(t)$ for $-1\le t\le K-1$.

\emph{The component $X_i$.} Given $a$ and $b$, $\pos_i=\one[b>a]+R$, where $R$ is the contribution of the other $N-2$
opponents: $i$ of smaller index and $N-2-i$ of larger index, so by Lemma~\ref{lem:groups}
$\E[f(R+1)-f(R)\mid a]=\hat D_i(a)$ and
$\E[u_i\mid a,b]=\chi(a)+\one[b>a]\hat D_i(a)$ with $\chi(a)=\E[f(R)\mid a]$. Now $X_i=C_aC_b\E[u_i\mid a,b]$; $C_b$
annihilates $\chi(a)$, and $C_b[\one[b>a]\hat D_i(a)]=(\one[b>a]-p(a))\hat D_i(a)=:F(a,b)$. Thus $X_i=C_aF$ with $F$
centred in $b$, and $\|X_i\|^2=\|F\|^2-\|A_aF\|^2$. First,
\[
\|F\|^2=\frac1K\sum_ap(a)(1-p(a))\hat D_i(a)^2=\frac1{K^2}\sum_{a=0}^{K-2}\frac{(a+1)(K-1-a)}K\hat D_i(a)^2=\frac{\calD(H_i)}{K^2},
\]
since $H_i(y)-H_i(y-1)=\hat D_i(y-1)$. Second, $\sum_a\one[b>a]\hat D_i(a)=H_i(b)$ and
$\sum_ap(a)\hat D_i(a)=\E_bH_i(b)$, so $(A_aF)(b)=\frac1K(H_i(b)-\E H_i)$ and $\|A_aF\|^2=\Var(H_i)/K^2$. This proves
the first identity.

\emph{The component $X_j$.} Symmetrically, given $a,b$ one has $\pos_j=\one[a\ge b]+R'$, where the other opponents of $j$
are $j-1$ of smaller index and $N-1-j$ of larger index, so that the coefficient is $\hat D_{j-1}(b)$, and
$X_j=C_b\tilde F$ with $\tilde F(a,b)=(\one[a\ge b]-p(b-1))\hat D_{j-1}(b)$, centred in $a$. Here
$p(b-1)(1-p(b-1))=\frac{b(K-b)}{K^2}$ and $G_{j-1}(y)-G_{j-1}(y-1)=\hat D_{j-1}(y)$ give
$\|\tilde F\|^2=\calD(G_{j-1})/K^2$, while $\sum_b\one[a\ge b]\hat D_{j-1}(b)=G_{j-1}(a)$ gives
$\|A_b\tilde F\|^2=\Var(G_{j-1})/K^2$.

\emph{The cross term.} Since $C_a$ is self-adjoint, $C_aC_b=C_bC_a$, $C_a\tilde F=\tilde F$ and $C_bF=F$,
\[
\langle C_aF,C_b\tilde F\rangle=\langle F,C_bC_a\tilde F\rangle=\langle F,C_b\tilde F\rangle=\langle C_bF,\tilde F\rangle
=\langle F,\tilde F\rangle .
\]
The two indicators are complementary, $\one[a\ge b]=1-\one[b>a]$ and $p(b-1)=1-\E_U\one[U\ge b]$, whence
\[
\big(\one[b>a]-p(a)\big)\big(\one[a\ge b]-p(b-1)\big)=-\big(\one[b>a]-p(a)\big)\big(\one[b>a]-\tfrac bK\big).
\]
Checking the cases $b>a$ and $b\le a$ separately, the right-hand side equals
$-\dk(a,b-1)$, where $\dk(t,t'):=\Cov_U(\one[U>t],\one[U>t'])=p(\max(t,t'))-p(t)p(t')\ge0$ for $U$ uniform on the levels.
Therefore
\[
\langle X_i,X_j\rangle=-\frac1{K^2}\sum_{a,b}\hat D_i(a)\hat D_{j-1}(b)\,\Cov_U\big(\one[U>a],\one[U\ge b]\big)
=-\frac1{K^2}\Cov\big(H_i(U),G_{j-1}(U)\big).\qedhere
\]
\end{proof}

\begin{propS}[two-body reduction]\label{thm:C}\label{prop:bond}
For every $N\ge2$, $K\ge2$ and every payoff-from-rank function $f$,
$Q^{(1)}_K(f)+Q^{(2)}_K(f)=\fB_{N,K}\big(\hat D(f)\big)$.
\end{propS}
\begin{proof}
For $m=2$, \eqref{eq:Qm-SM} reads
$Q^{(2)}_K=\sum_{i<j}\big[\frac{N-2}{2N}(\|X_i\|^2+\|X_j\|^2)-\langle X_i,X_j\rangle\big]$, since
$\frac{N-1}N-\frac12=\frac{N-2}{2N}$. By \eqref{eq:second}, bond $k$ occurs through $P(H_k)$ in the $N-1-k$ pairs $(k,j)$,
$j>k$, and through $P(G_k)$ in the $k+1$ pairs $(i,k+1)$, $i\le k$; this gives $K^{-2}\,\mathrm I$. The map
$(i,j)\mapsto(k,l)=(i,j-1)$ is a bijection from $\{i<j\}$ onto $\{k\le l\}$, so the cross terms give $K^{-2}\,\mathrm{II}$.
Finally $Q^{(1)}_K=-K^{-2}\,\mathrm{III}$ by \eqref{eq:Q1}.
\end{proof}

\begin{lemmaS}[kernel transfer]\label{lem:transfer}
$\hat D(f)\in\operatorname{span}(\one\otimes\one)$, i.e.\ $\hat D_k(a)=c$ for all $k$ and $a$, if and only if $f$ is
affine.
\end{lemmaS}
\begin{proof}
If $\Delta f\equiv c$ then $\hat D_k\equiv c$ by \eqref{eq:Dhat}. Conversely, at the top level $a=K-1$ the second binomial
in \eqref{eq:Dhat} is degenerate at $0$, so $\hat D_k(K-1)=\E\,\Delta f(\Bin(k,\frac1K))$, $k=0,\dots,N-2$. This is a
triangular linear system in $\Delta f(0),\dots,\Delta f(N-2)$ with diagonal entries $K^{-k}>0$; its unique solution with
right-hand side $c\,\one$ is $\Delta f\equiv c$.
\end{proof}

Proposition~\ref{thm:C} uses only exchangeability within the two index groups, the one-flip identity and the orthogonality of
the Hoeffding projections. It holds for every $K\ge2$, including $K=2$, where Theorem~\ref{thm:finiteK} fails: the
hypothesis $K\ge3$ enters only through Lemma~\ref{lem:Lstar-SM} below.

%=====================================================================================================
\stepsection{SOS}{the sum-of-squares identity}\label{sec:SOS}
The bond form \eqref{eq:formula1} has one negative term, $\mathrm{III}$, of the same order as the two positive ones. This
step rearranges the three terms into a sum of manifestly nonnegative pieces plus a single negative term of a
different shape, $-\sum_kc_k\Var(D_k)$, which involves the bond vectors only through their variances; every later step
is a way of paying for that term. Two elementary identities are needed first.

\begin{lemmaS}[Hahn--Poincar\'e]\label{lem:hahn}
For every function $g$ of the level, $\calD(g)=\sum_{m=1}^{K-1}m(m+1)\Var_m(g)$, where $\Var_m(g)$ is the squared norm of
the component of $g$ along the discrete Chebyshev polynomial of degree $m$ on $\{0,\dots,K-1\}$. In particular
$\calD(g)\ge2\Var(g)$, with equality iff $g$ is affine.
\end{lemmaS}
\begin{proof}
Summation by parts gives $\calD(g)=\E[g\cdot Lg]$ with
$(Lg)(y)=-\big[(y+1)(K-1-y)(g(y+1)-g(y))-y(K-y)(g(y)-g(y-1))\big]$, the second-order difference operator of the Hahn
polynomials with parameters $(0,0)$, whose eigenfunctions are the discrete Chebyshev polynomials with eigenvalues
$m(m+1)$ \cite[Sec.~9.5]{koekoek2010}. The components of degree $m\ge2$ have eigenvalue at least $6>2$.
\end{proof}

\begin{lemmaS}[spread identity]\label{lem:spread}
For functions $z_0,\dots,z_{n-1}$ of the level,
$n\sum_k\Var(z_k)-\Var\big(\sum_kz_k\big)=\sum_{k<l}\Var(z_k-z_l)$.
\end{lemmaS}
\begin{proof}
Both sides equal $(n-1)\sum_k\Var(z_k)-2\sum_{k<l}\Cov(z_k,z_l)$.
\end{proof}

\begin{thmS}[sum-of-squares identity]\label{thm:SOS}
For all bond vectors $D_0,\dots,D_{N-2}\in\R^K$,
\begin{equation}\label{eq:SOS-SM}
K^2\fB_{N,K}(D)=\sum_kZ_k+\mathrm{Sq}(D)-\sum_kc_k\Var(D_k),\qquad c_k:=2N\mu_k\nu_k-\tfrac12 ,
\end{equation}
where $Z_k:=\frac{N-2}2\big[\mu_k(\calD-2\Var)(H_k)+\nu_k(\calD-2\Var)(G_k)\big]\ge0$ and
\begin{equation}\label{eq:Sq-SM}
\mathrm{Sq}(D):=\frac1{2N}\sum_{k<l}\Big[(N{-}1{-}l)\Var(H_k-H_l+D_l)+(k{+}1)\Var(G_k-G_l+D_k)+(l{-}k)\Var(H_k-G_l)\Big]\ge0 .
\end{equation}
\end{thmS}
\begin{proof}
\emph{1) Vertex variables.} For $j\in\{0,\dots,N-1\}$ and a bond $k$ put $Z_{jk}:=H_k$ if $k<j$ and $Z_{jk}:=G_k$ if
$k\ge j$, so that $\Omega_j=\sum_kZ_{jk}$, and let $\E_j$ denote the average over $j$. For fixed $k$, $Z_{jk}=H_k$ for the
$N-1-k$ values $j>k$ and $Z_{jk}=G_k$ for the $k+1$ values $j\le k$. Hence, in \eqref{eq:formula1},
\[
\mathrm I=\tfrac{N-2}2\,\E_j\sum_kP(Z_{jk}),\qquad \mathrm{III}=\E_j\Var\Big(\sum_kZ_{jk}\Big),\qquad
\sum_kZ_k=\tfrac{N-2}2\,\E_j\sum_k(\calD-2\Var)(Z_{jk}).
\]
\emph{2) Spreading.} With $n=N-1$ bonds, Lemma~\ref{lem:spread} gives, for each $j$,
\[
\tfrac{N-2}2\sum_kP(Z_{jk})=\tfrac{N-2}2\sum_k(\calD-2\Var)(Z_{jk})+\tfrac12\sum_{k<l}\Var(Z_{jk}-Z_{jl})
-\tfrac12\sum_k\Var(Z_{jk})+\tfrac12\Var(\Omega_j);
\]
indeed both sides equal $\frac{N-2}2\sum_k\calD(Z_{jk})-\frac{N-2}2\sum_k\Var(Z_{jk})$. Averaging over $j$ and inserting
into \eqref{eq:formula1},
\begin{equation}\label{eq:SOSmid}
K^2\fB=\sum_kZ_k+\tfrac12\E_j\sum_{k<l}\Var(Z_{jk}-Z_{jl})+\Xi,\qquad
\Xi:=\mathrm{II}-\tfrac12\E_j\Big[\Var(\Omega_j)+\sum_k\Var(Z_{jk})\Big].
\end{equation}
\emph{3) The pair term.} For $k<l$ the pair $(Z_{jk},Z_{jl})$ equals $(G_k,G_l)$ for the $k+1$ values $j\le k$,
$(H_k,G_l)$ for the $l-k$ values $k<j\le l$, and $(H_k,H_l)$ for the $N-1-l$ values $j>l$. Therefore
\begin{equation}\label{eq:pairterm}
\tfrac12\E_j\sum_{k<l}\Var(Z_{jk}-Z_{jl})=\frac1{2N}\sum_{k<l}\Big[(k{+}1)\Var(G_k-G_l)+(l{-}k)\Var(H_k-G_l)+(N{-}1{-}l)\Var(H_k-H_l)\Big].
\end{equation}
\emph{4) The residual.} Since $\frac12[\Var(\Omega_j)+\sum_k\Var(Z_{jk})]=\sum_k\Var(Z_{jk})+\sum_{k<l}\Cov(Z_{jk},Z_{jl})$,
the same count gives
\[
\Xi=\sum_{k<l}\Big[\big(1-\tfrac{l-k}N\big)\Cov(H_k,G_l)-\nu_k\Cov(G_k,G_l)-\mu_l\Cov(H_k,H_l)\Big]
+\sum_k\Big[\Cov(H_k,G_k)-\mu_k\Var(H_k)-\nu_k\Var(G_k)\Big].
\]
As $1-\frac{l-k}N=\nu_k+\mu_l$ and $G=H+D$, the first bracket is $\mu_l\Cov(H_k,D_l)-\nu_k\Cov(G_l,D_k)$; as
$\mu_k+\nu_k=1$, the second is $\mu_k\Cov(H_k,D_k)-\nu_k\Cov(G_k,D_k)=(\mu_k-\nu_k)\Cov(H_k,D_k)-\nu_k\Var(D_k)$.
Thus
\begin{equation}\label{eq:Xi}
\Xi=\sum_{k<l}\big[\mu_l\Cov(H_k,D_l)-\nu_k\Cov(G_l,D_k)\big]+\sum_k\big[(\mu_k-\nu_k)\Cov(H_k,D_k)-\nu_k\Var(D_k)\big].
\end{equation}
\emph{5) Absorption.} Expanding the shifted variances and using
$2\Cov(H,D)=\Var(G)-\Var(H)-\Var(D)$,
\begin{align}
\tfrac12\Var(H_k-H_l)+\Cov(H_k,D_l)&=\tfrac12\Var(H_k-H_l+D_l)+\tfrac12\big[\Var G_l-\Var H_l\big]-\Var(D_l),\label{eq:abs1}\\
\tfrac12\Var(G_k-G_l)-\Cov(G_l,D_k)&=\tfrac12\Var(G_k-G_l+D_k)-\tfrac12\big[\Var G_k-\Var H_k\big]-\Var(D_k).\label{eq:abs2}
\end{align}
The weights match: the $(H,H)$ pair carries $\frac{N-1-l}{2N}=\frac{\mu_l}2$ in \eqref{eq:pairterm} and $\mu_l$ in
\eqref{eq:Xi}, the $(G,G)$ pair $\frac{\nu_k}2$ and $\nu_k$. Applying \eqref{eq:abs1}--\eqref{eq:abs2} converts
\eqref{eq:pairterm} plus the off-diagonal part of \eqref{eq:Xi} into $\mathrm{Sq}(D)$ plus multiples of
$\frac12(\Var G_k-\Var H_k)$ and of $\Var(D_k)$, which we now collect together with the diagonal part of \eqref{eq:Xi},
written as $(\mu_k-\nu_k)\cdot\frac12(\Var G_k-\Var H_k)-\frac{\mu_k-\nu_k}2\Var(D_k)-\nu_k\Var(D_k)$.

\emph{Coefficient of $\frac12(\Var G_k-\Var H_k)$.} Bond $k$ is the second index of $k$ pairs, each contributing
$+\mu_k$ by \eqref{eq:abs1}, and the first index of $N-2-k$ pairs, each contributing $-\nu_k$ by \eqref{eq:abs2}; the
diagonal contributes $\mu_k-\nu_k$. The total is
$\frac1N\big[k(N-1-k)-(k+1)(N-2-k)+(N-2-2k)\big]=0$.

\emph{Coefficient of $\Var(D_k)$.} The same three sources give $-k\mu_k$, $-(N-2-k)\nu_k$ and
$-\frac{\mu_k-\nu_k}2-\nu_k=-\frac12$, and
$\frac{k(N-1-k)+(k+1)(N-2-k)}N+\frac12=\frac{2(N-1-k)(k+1)}N-\frac12=c_k$.

\emph{6) Signs.} $Z_k\ge0$ by Lemma~\ref{lem:hahn}, and $\mathrm{Sq}\ge0$ because its weights are nonnegative.
\end{proof}

%=====================================================================================================
\stepsection{kerSq}{the kernel of \texorpdfstring{$\mathrm{Sq}$}{Sq}}\label{sec:kerSq}
The term $\mathrm{Sq}$ of Theorem~\ref{thm:SOS} is what pins down the kernel of $\fB_{N,K}$: once positivity is known,
$\fB(D)=0$ forces $\mathrm{Sq}(D)=0$, and the following lemma turns that into $D\in\operatorname{span}(\one\otimes\one)$.
It is used in the kernel part of Proposition~\ref{prop:T} and in the assembly (Step~\ref{step:assembly}).

\begin{lemmaS}[kernel of $\mathrm{Sq}$]\label{lem:kerSq}
Let $N\ge3$, $K\ge2$. If $\mathrm{Sq}(D)=0$ then $D\in\operatorname{span}(\one\otimes\one)$: there is $c\in\R$ with
$D_k=c\,\one$ for every bond $k$.
\end{lemmaS}
\begin{proof}
On the range $0\le k<l\le N-2$ the three weights in \eqref{eq:Sq-SM} are at least $1$. Hence $\mathrm{Sq}(D)=0$ forces,
for every pair $k<l$, each of $H_k-G_l$, $H_k-H_l+D_l$ and $G_k-G_l+D_k$ to be constant in the level. Using
$H(a+1)-H(a)=D(a)$ and $G(a+1)-G(a)=D(a+1)$ and taking first differences at $a=0,\dots,K-2$:
\[
D_k(a)=D_l(a+1),\qquad D_k(a)-2D_l(a)+D_l(a+1)=0,\qquad 2D_k(a+1)-D_l(a+1)-D_k(a)=0 .
\]
Substituting the first relation into the second gives $D_l(a+1)=D_l(a)$, so $D_l$ is constant whenever $l$ occurs as a
second index; substituting it into the third gives $D_k(a+1)=D_k(a)$, so $D_k$ is constant whenever $k$ occurs as a first
index. Since $N\ge3$ there are at least two bonds, and every bond is a first index ($k\le N-3$) or a second index
($l\ge1$); so $D_k=c_k\one$ for all $k$, and the first relation gives $c_k=c_l$.
\end{proof}
All three groups of squares are needed: the first alone does not constrain $D_0(K-1)$ nor $D_{N-2}(0)$; the second removes
the latter freedom and the third the former.

%=====================================================================================================
\stepsection{Lstar}{the sharp level inequality}\label{sec:Lstar}
The terms $Z_k$ of Theorem~\ref{thm:SOS} are nonnegative by Lemma~\ref{lem:hahn}, but that alone does not pay for
$-c_k\Var(D_k)$. This step shows that $Z_k$ dominates $\Var(D_k)$ with an explicit constant, and that the constant
carries the factor $K-2$: this is the one place where the hypothesis $K\ge3$ enters.

\begin{lemmaS}[sharp level inequality]\label{lem:Lstar-SM}\label{lem:Lstar}
For every $K\ge3$, $\mu\in[0,1]$, $\nu=1-\mu$ and $D\in\R^K$,
\begin{equation}\label{eq:Lstar-SM}
\mu(\calD-2\Var)(J_0D)+\nu(\calD-2\Var)(J_1D)\ \ge\ (K-2)\,\mu\nu\,\Var(D),
\end{equation}
and the constant $K-2$ cannot be improved. Moreover, in the step coordinates $c_m=D(m)-D(m-1)$, $m=1,\dots,K-1$, the
difference of the two sides of \eqref{eq:Lstar-SM} equals $\frac2{K^2}\,c^\top\big[(Kg)\circ\Phi_\mu\big]c$, where
$g(m,m')=\min(m,m')(K-\max(m,m'))/K$ is the Green's function of the Dirichlet Laplacian of the path
$\{1,\dots,K-1\}$ and $\Phi_\mu=\widetilde M_\mu-\frac{K\mu\nu}2\one\one^\top$ is positive semidefinite, with
\begin{equation}\label{eq:Mtilde}
\widetilde M_\mu(m,m')=\tilde a_{\min(m,m')}\tilde b_{\max(m,m')},\qquad
\tilde a_m=\frac{m-1+2\mu}2,\quad\tilde b_m=\frac{K-m-1+2\nu}2 .
\end{equation}
\end{lemmaS}
\begin{proof}
\emph{1) Pair representation.} Let $\Phi=J_1D$ and, for $1\le y\le y'\le K-1$,
$\gamma(y,y'):=\frac yK\big(1-\frac{y'}K\big)$. Since
$\Var(\Phi)=\frac1{K^2}\sum_{a<b}\big(\sum_{a<y\le b}D(y)\big)^2$ under the uniform law on the levels, the coefficient
of $D(y)D(y')$, $y\le y'$, in $\Var(\Phi)$ is $\frac1{K^2}\#\{(a,b):a<y,\ b\ge y'\}=\gamma(y,y')$ (counted twice if
$y<y'$). Writing
$D(y)D(y')=\frac12(D(y)^2+D(y')^2)-\frac12(D(y)-D(y'))^2$,
\[
\Var(\Phi)=\sum_{y\ge1}D(y)^2\Big[\gamma(y,y)+\sum_{y'\ne y}\gamma(\min,\max)\Big]-\sum_{1\le y<y'}\gamma(y,y')\big(D(y)-D(y')\big)^2 .
\]
A direct summation gives $\sum_{y'\ne y}\gamma(\min,\max)=\frac{K-2}2\gamma(y,y)$, while $\calD(\Phi)=K\sum_y\gamma(y,y)D(y)^2$.
Hence the diagonal terms cancel in $\calD-2\Var$ and
\[
(\calD-2\Var)(J_1D)=\frac2{K^2}\sum_{1\le y<y'\le K-1}y\,(K-y')\,\big(D(y)-D(y')\big)^2 .
\]
The same computation for $J_0D$, which is $J_1D$ shifted by one level, gives
$(\calD-2\Var)(J_0D)=\frac2{K^2}\sum_{0\le y<y'\le K-2}(y+1)(K-1-y')(D(y)-D(y'))^2$.

\emph{2) The pencil.} For $0\le y<y'\le K-1$ (both weights vanish outside the ranges above),
\[
\mu(y+1)(K-1-y')+\nu\,y(K-y')=(y+\mu)(K-\mu-y')-\mu\nu ,
\]
an identity in $(y,y',\mu)$. Since $\Var(D)=\frac1{K^2}\sum_{y<y'}(D(y)-D(y'))^2$, \eqref{eq:Lstar-SM} is equivalent to
\begin{equation}\label{eq:pencil}
\sum_{0\le y<y'\le K-1}\Big[(y+\mu)(K-\mu-y')-\tfrac K2\mu\nu\Big]\big(D(y)-D(y')\big)^2\ \ge0 .
\end{equation}
\emph{3) Step coordinates.} With $D(y')-D(y)=\sum_{y<m\le y'}c_m$, the left-hand side of \eqref{eq:pencil} is
$c^\top\mathcal Kc$ with, for $m\le m'$,
\[
\mathcal K(m,m')=\sum_{y<m,\ y'\ge m'}\Big[(y+\mu)(K-\mu-y')-\tfrac K2\mu\nu\Big]
=m(K-m')\Big[\tilde a_m\tilde b_{m'}-\tfrac K2\mu\nu\Big],
\]
because $\sum_{y<m}(y+\mu)=m\tilde a_m$ and $\sum_{y'\ge m'}(K-\mu-y')=(K-m')\tilde b_{m'}$. Thus
$\mathcal K=(Kg)\circ\Phi_\mu$.

\emph{4) $\Phi_\mu\succeq0$.} Let $0<\mu<1$. The sequences $\tilde a$ and $\tilde b$ are positive, $\tilde a$
increasing and $\tilde b$ decreasing, with constant Wronskian
$\tilde a_{m+1}\tilde b_m-\tilde a_m\tilde b_{m+1}=\frac12(\tilde a_m+\tilde b_m)=\frac K4$. The inverse of such a one-pair
(Green) matrix is tridiagonal \cite[Ch.~II, pp.~78--79 and 103--104]{gantmacher2002}: $\widetilde T:=\widetilde M_\mu^{-1}$ has off-diagonal
entries $-\frac4K$, interior diagonal entries $\frac8K$, and corner entries $\frac4K(1+\frac1{2\mu})$ and
$\frac4K(1+\frac1{2\nu})$. That is, $\widetilde T=\frac4K\big[L_{\rm path}+\operatorname{diag}(\frac1{2\mu},0,\dots,0,\frac1{2\nu})\big]\succ0$,
with $L_{\rm path}$ the Neumann Laplacian of the path; in particular
$\one^\top\widetilde T\one=\frac4K\big(\frac1{2\mu}+\frac1{2\nu}\big)=\frac2{K\mu\nu}$. For a positive definite
$M$, $c\,\one\one^\top\succeq0$ iff $c\,\one^\top M^{-1}\one\le1$; with $c=\frac K2\mu\nu$ equality holds, so
$\Phi_\mu\succeq0$, with one-dimensional kernel. The cases $\mu\in\{0,1\}$ follow by continuity.

\emph{5) Schur product.} $g\succeq0$ and $\Phi_\mu\succeq0$, so $\mathcal K=(Kg)\circ\Phi_\mu\succeq0$ by the Schur
product theorem. Sharpness of the constant $K-2$: take $D=e_0$ and let $\mu\to0$.
\end{proof}

\begin{propS}[the level pencil]\label{cor:rho}
For $0<\mu<1$ let $\varrho(\mu)$ be the smallest generalised eigenvalue of the pencil
$\big((Kg)\circ\widetilde M_\mu,\,Kg\big)$, and let $c^\star(\mu)$ be the largest constant that can replace
$(K-2)\mu\nu$ in \eqref{eq:Lstar-SM}. Then $\varrho(\mu)=\mu\nu+\tfrac12c^\star(\mu)$, which is the definition
used in the main text, and $\varrho(\mu)\ge\frac K2\mu\nu$.
\end{propS}
\begin{proof}
By part 3) of the preceding proof and $\Var(D)=K^{-2}c^\top(Kg)c$, the left-hand side of \eqref{eq:Lstar-SM} equals
$\frac2{K^2}c^\top[(Kg)\circ\widetilde M_\mu]c-2\mu\nu\Var(D)$, so $\mathrm{LHS}\ge\lambda\Var(D)$ holds for all
$D$ iff $(Kg)\circ\widetilde M_\mu\succeq(\mu\nu+\frac\lambda2)Kg$; hence $c^\star=2\varrho-2\mu\nu$. Finally
$(Kg)\circ\widetilde M_\mu-\frac K2\mu\nu\,(Kg)=(Kg)\circ\Phi_\mu\succeq0$, and $Kg\succ0$.
\end{proof}

%=====================================================================================================
\stepsection{Sqlow}{a lower bound in the bond vectors alone, and step coordinates}\label{sec:Sqlow}
The identity \eqref{eq:SOS-SM} has a single negative term, $-\sum_kc_k\Var(D_k)$, which involves the bond vectors only
through their covariances. We bound the two positive terms from below by quantities of the same kind.

\begin{lemmaS}[minimising out a common part]\label{lem:minout}
Let $w_1,\dots,w_r>0$ with $\sum_iw_i=W$, and let $Y_1,\dots,Y_r$ be functions of the level. For every function $X$ of
the level,
$\sum_iw_i\Var(X+Y_i)\ \ge\ \sum_iw_i\Var(Y_i)-\frac1W\Var\big(\sum_iw_iY_i\big)$.
\end{lemmaS}
\begin{proof}
$\sum_iw_i\Var(X+Y_i)=W\Var(X)+2\Cov(X,\bar Y)+\sum_iw_i\Var(Y_i)$ with $\bar Y=\sum_iw_iY_i$, and
$W\Var(X)+2\Cov(X,\bar Y)=W\Var(X+\bar Y/W)-\Var(\bar Y)/W\ge-\Var(\bar Y)/W$.
\end{proof}

\begin{propS}[the bond matrix]\label{prop:Sqlow}
Let $g_b(k,l):=\frac{(\min(k,l)+1)(N-1-\max(k,l))}N$, $0\le k,l\le N-2$, and
\begin{equation}\label{eq:Pbond}
\mathsf P:=2(N-1)\operatorname{diag}(\mu_k\nu_k)-\tfrac2N\,g_b .
\end{equation}
Then $\mathrm{Sq}(D)\ \ge\ \mathrm{Sq}_{\rm low}(D):=\sum_{k,l}\mathsf P_{kl}\Cov(D_k,D_l)$ for all bond vectors.
\end{propS}
\begin{proof}
Fix $k<l$ and put $X=H_k-H_l$, $a=N-1-l$, $b=k+1$, $c=l-k$, so that $a+b+c=N$. Since $G=H+D$, the three functions in
\eqref{eq:Sq-SM} are $X+D_l$, $X+2D_k-D_l$ and $X-D_l$, with weights $a,b,c$. By Lemma~\ref{lem:minout}, and because
$aD_l+b(2D_k-D_l)-cD_l=2bD_k+(2a-N)D_l$, the pair $(k,l)$ contributes to $2N\,\mathrm{Sq}$ at least
\[
(a+c)\Var(D_l)+b\Var(2D_k-D_l)-\tfrac1N\Var\big(2bD_k+(2a-N)D_l\big)
=\tfrac{4a(N-a)}N\Var(D_l)+\tfrac{4b(N-b)}N\Var(D_k)-\tfrac{8ab}N\Cov(D_k,D_l),
\]
where we used $N-\frac{(2a-N)^2}N=\frac{4a(N-a)}N$ and $-4b-\frac{4b(2a-N)}N=-\frac{8ab}N$. Now $a=N\mu_l$,
$N-a=N\nu_l$, $b=N\nu_k$, $N-b=N\mu_k$; dividing by $2N$, the pair contributes at least
$2\mu_l\nu_l\Var(D_l)+2\mu_k\nu_k\Var(D_k)-4\nu_k\mu_l\Cov(D_k,D_l)$. Each bond belongs to $N-2$ pairs,
which gives the diagonal $2(N-2)\mu_k\nu_k=2(N-1)\mu_k\nu_k-\frac2Ng_b(k,k)$, and for $k<l$ the symmetric
off-diagonal entry is $-2\nu_k\mu_l=-\frac2Ng_b(k,l)$.
\end{proof}

\begin{propS}[step coordinates]\label{prop:hadamard}
Let $c_k\in\R^{K-1}$, $c_k(m)=D_k(m)-D_k(m-1)$, and let $g$, $\widetilde M_\mu$ be as in Lemma~\ref{lem:Lstar-SM}. Then
$\Cov(D_k,D_l)=K^{-2}c_k^\top(Kg)c_l$, and
\begin{equation}\label{eq:F}
K^2\Big[\sum_kZ_k+\mathrm{Sq}_{\rm low}-\sum_kc_k\Var(D_k)\Big]
=(N-2)\sum_kc_k^\top\big[(Kg)\circ\widetilde M_{\mu_k}\big]c_k+\sum_{k,l}\Lambda_{kl}\,c_k^\top(Kg)\,c_l,
\end{equation}
with $\Lambda:=\mathsf P-\operatorname{diag}\big((N-2)\mu_k\nu_k+c_k\big)$, where in the last expression $c_k=2N\mu_k\nu_k-\frac12$
is the scalar of \eqref{eq:SOS-SM}.
\end{propS}
\begin{proof}
$\Cov(D,D')=\frac1{K^2}\sum_{y<y'}(D(y')-D(y))(D'(y')-D'(y))$, and the coefficient of $c(m)c'(m')$, $m\le m'$, in this sum
is $\#\{y<m\}\cdot\#\{y'\ge m'\}=m(K-m')=Kg(m,m')$. By Lemma~\ref{lem:Lstar-SM},
$\mu(\calD-2\Var)(J_0D)+\nu(\calD-2\Var)(J_1D)=\frac2{K^2}c^\top[(Kg)\circ\Phi_\mu]c+(K-2)\mu\nu\Var(D)$, and
$(Kg)\circ\Phi_\mu=(Kg)\circ\widetilde M_\mu-\frac{K\mu\nu}2Kg$; hence
\[
K^2Z_k=(N-2)\,c_k^\top\big[(Kg)\circ\widetilde M_{\mu_k}\big]c_k-(N-2)\mu_k\nu_k\,c_k^\top(Kg)c_k .
\]
Add $K^2\mathrm{Sq}_{\rm low}=\sum_{k,l}\mathsf P_{kl}c_k^\top(Kg)c_l$ and $-K^2c_k\Var(D_k)=-c_k\,c_k^\top(Kg)c_k$.
\end{proof}

%=====================================================================================================
\stepsection{perron}{a Perron argument in the space of bonds}\label{sec:perron}
Proposition~\ref{prop:hadamard} has reduced positivity of $\fB_{N,K}$ to positivity of a block matrix on the space
of bonds whose blocks are Hadamard products with the fixed positive definite matrix $Kg$. This step shows that the
whole matrix is positive as soon as a scalar condition, (T) below, holds on every bond; the condition is then checked
branch by branch in Steps~\ref{step:K5}--\ref{step:K4}. Let $L$ be the Dirichlet Laplacian of the path with vertices $0,\dots,N-2$ (diagonal $2$, off-diagonal $-1$), and recall
$\varrho(\mu)$ from Proposition~\ref{cor:rho}.

\begin{lemmaS}\label{lem:path}
$g_b=L^{-1}$, and the positive vector $v_k=(k+1)(N-1-k)=N^2\mu_k\nu_k$ satisfies $Lv=2\cdot\one$.
\end{lemmaS}
\begin{proof}
$g_b$ is the Green's function of the Dirichlet problem on the path; $v$ is quadratic in $k$ with second difference $-2$
and vanishes at $k=-1$ and $k=N-1$.
\end{proof}

\begin{lemmaS}[symmetric $Z$-matrices]\label{lem:Zmatrix}
Let $M$ be a real symmetric matrix with $M_{ij}\le0$ for $i\ne j$, and let $v$ have positive entries. For every $x$,
\begin{equation}\label{eq:Zid}
x^\top Mx=\sum_i\frac{(Mv)_i}{v_i}x_i^2+\frac12\sum_{i\ne j}(-M_{ij})\,v_iv_j\Big(\frac{x_i}{v_i}-\frac{x_j}{v_j}\Big)^2 .
\end{equation}
Hence $Mv\ge0$ implies $M\succeq0$; if moreover $M$ is irreducible and $Mv\ne0$, then $M\succ0$.
\end{lemmaS}
\begin{proof}
Expand the square and use $\sum_jM_{ij}v_j=(Mv)_i$. If $x^\top Mx=0$ with $M$ irreducible, the second sum forces
$x=t\,v$, and then the first gives $t^2v^\top Mv=0$, so $t=0$ when $Mv\ne0$.
\end{proof}

\begin{propS}[condition (T)]\label{prop:T}
Let $N\ge3$, $K\ge3$, and put $\theta_k:=(N-2)\varrho(\mu_k)-N\mu_k\nu_k+\frac12$. If, for every bond $k$,
\begin{equation}\label{eq:T-SM}
(N-2)\,\varrho(\mu_k)-2N\mu_k\nu_k+\tfrac12\ \ge\ 0,\tag{T}
\end{equation}
then $\fB_{N,K}\succeq0$. If in addition \eqref{eq:T-SM} is strict for at least one bond, then
$\ker\fB_{N,K}=\operatorname{span}(\one\otimes\one)$.
\end{propS}
\begin{proof}
By Theorem~\ref{thm:SOS} and Propositions~\ref{prop:Sqlow} and~\ref{prop:hadamard},
$K^2\fB=(\mathrm{Sq}-\mathrm{Sq}_{\rm low})+K^{-2}\mathcal F(c)$ with $\mathcal F$ the right-hand side of \eqref{eq:F}, and
$\mathrm{Sq}-\mathrm{Sq}_{\rm low}\ge0$.
By definition of $\varrho$, $(Kg)\circ\widetilde M_{\mu_k}\succeq\varrho(\mu_k)\,Kg$, so
$\mathcal F(c)\ge\sum_{k,l}\mathsf A_{kl}\,c_k^\top(Kg)c_l$ with
$\mathsf A:=\operatorname{diag}\big((N-2)\varrho(\mu_k)\big)+\Lambda$. Inserting \eqref{eq:Pbond},
\[
\mathsf A=\operatorname{diag}(\theta_k)-\tfrac2N\,g_b=\operatorname{diag}(\theta_k)-\tfrac2N\,L^{-1}.
\]
Since $Kg\succ0$, the form $\sum_{k,l}\mathsf A_{kl}c_k^\top(Kg)c_l=c^\top(\mathsf A\otimes Kg)c$ is nonnegative iff
$\mathsf A\succeq0$. Condition \eqref{eq:T-SM} says $\theta_k\ge N\mu_k\nu_k>0$. For $\Theta=\operatorname{diag}(\theta_k)\succ0$,
$\Theta\succeq\frac2NL^{-1}$ is equivalent to $\frac N2L\succeq\Theta^{-1}$, i.e.\ to $\mathsf M:=NL-2\Theta^{-1}\succeq0$.
$\mathsf M$ is a symmetric $Z$-matrix and, by Lemma~\ref{lem:path},
$(\mathsf Mv)_k=2N-\frac{2v_k}{\theta_k}=\frac{2N}{\theta_k}\big(\theta_k-N\mu_k\nu_k\big)\ge0$,
so $\mathsf M\succeq0$ by Lemma~\ref{lem:Zmatrix}, and $\fB\succeq0$.

For the kernel, let $\fB(D)=0$. If \eqref{eq:T-SM} is strict for some bond, then $\mathsf Mv\ne0$; $L$ is irreducible, so
$\mathsf M\succ0$, hence $\mathsf A\succ0$ and $\mathsf A\otimes Kg\succ0$. As both $\mathrm{Sq}-\mathrm{Sq}_{\rm low}$ and
$\mathcal F$ are nonnegative, $\mathcal F(c)=0$, so $c=0$: every bond vector is constant in the level, $D_k=t_k\one$. Then
$H_k(a)=t_ka$, $G_k(a)=t_k(a+1)$, $\mathrm{Sq}_{\rm low}(D)=0$ and
$\mathrm{Sq}(D)=\frac1{2N}\sum_{k<l}N(t_k-t_l)^2\Var(a)$, which vanishes only if all $t_k$ coincide. Conversely
$\fB(\one\otimes\one)=0$ by Proposition~\ref{thm:C} and Lemma~\ref{lem:B} applied to the linear payoff, for which all
$Q^{(m)}_K$, $m\ge3$, vanish.
\end{proof}

%=====================================================================================================
\stepsection{K5}{the branch \texorpdfstring{$K\ge5$}{K>=5}}\label{sec:K5}
For $K\ge5$ the lower bound $\varrho(\mu)\ge\frac K2\mu\nu$ of Proposition~\ref{cor:rho} verifies (T) for all but
finitely many $N$.
\begin{thmS}\label{thm:K5}
Let $K\ge5$ and $N\ge3$ with $(N-2)(K-4)\ge4$. Then $\fB_{N,K}\succeq0$ with kernel $\operatorname{span}(\one\otimes\one)$.
Explicitly this covers $N\ge6$ for $K=5$, $N\ge4$ for $K=6,7$, and $N\ge3$ for $K\ge8$.
\end{thmS}
\begin{proof}
Write $x=\mu_k\nu_k\in(0,\frac14]$. By Proposition~\ref{cor:rho}, $\varrho(\mu_k)\ge\frac K2x$, so the left-hand side of
\eqref{eq:T-SM} is at least $x\,s+\frac12$ with $s:=\frac{(N-2)K}2-2N$. If $s\ge0$ this is positive. If $s<0$ it is at
least $\frac s4+\frac12=\frac18\big[(N-2)(K-4)-4\big]\ge0$, with strict inequality unless $x=\frac14$. Since
$\mu_k\nu_k=\frac14$ for at most one bond and there are $N-1\ge2$ bonds, \eqref{eq:T-SM} is strict for some bond, and
Proposition~\ref{prop:T} applies.
\end{proof}
\emph{Where the argument stops.} At $K=4$ the bracket is $(N-2)\cdot0-4<0$ for every $N$: the bound $\varrho\ge\frac K2\mu\nu$ is never sufficient, and the
true value of $\varrho$ is needed. At $K=3$ even that fails (Step~\ref{step:K3}). The three-branch case distinction is therefore forced by the problem, not by the method.

%=====================================================================================================
\stepsection{K4}{the branch \texorpdfstring{$K=4$}{K=4} (computer-assisted root isolation)}\label{sec:K4}
For $K=4$ the pencil of Proposition~\ref{cor:rho} is $3\times3$ and
$\det\big((Kg)\circ\widetilde M_\mu-\varrho\,Kg\big)=-4\,\mathcal P(\varrho;\mu)$ with
\begin{align*}
\mathcal P(\varrho;\mu)={}&4\varrho^3-10\varrho^2+6\varrho+12\mu^4\varrho-24\mu^3\varrho+12\mu^2\varrho^2-13\mu^2\varrho-12\mu\varrho^2+25\mu\varrho\\
&+4\mu^6-12\mu^5-3\mu^4+26\mu^3-3\mu^2-12\mu .
\end{align*}

\begin{lemmaS}\label{lem:K4root}
For every $\mu\in(0,1)$, $\varrho(\mu)$ is the unique root of $\mathcal P(\cdot;\mu)$ in $(0,\frac7{10})$, and
\begin{equation}\label{eq:rho0}
\varrho(\mu)\ >\ \varrho_0(\mu):=2\mu\nu+5(\mu\nu)^3 .
\end{equation}
\end{lemmaS}
\begin{proof}
All three roots are positive, because both matrices of the pencil are positive definite. The following facts about
polynomials in $\mu$ with rational coefficients are established by exact root counting (Sturm sequences; script
\texttt{audit\_K4.py}): (i) $\mathcal P(0;\mu)=\mu(1-\mu)q_0(\mu)$ with $\deg q_0=4$ and $q_0<0$ on $[0,1]$;
(ii) $\mathcal P(\frac7{10};\mu)$, of degree $6$, has no root in $[0,1]$ and equals $\frac{51}{500}>0$ at $\mu=\frac12$;
(iii) $\mathcal P(\varrho_0(\mu);\mu)=[\mu(1-\mu)]^2q_1(\mu)$ with $\deg q_1=14$ and $q_1<0$ on $[0,1]$. Moreover
the three roots sum to $\frac52+3\mu(1-\mu)\ge\frac52>3\cdot\frac7{10}$, so they are not all below $\frac7{10}$. With the
positive leading coefficient, (i) and (ii) give exactly one root in $(0,\frac7{10})$ and two above; it is $\varrho(\mu)$.
Since $\varrho_0\le\frac{37}{64}<\frac7{10}$ and $\mathcal P<0$ on $(0,\varrho(\mu))$, (iii) gives $\varrho_0(\mu)<\varrho(\mu)$.
\end{proof}
The bound \eqref{eq:rho0} is nearly tight ($\min_\mu[\varrho-\varrho_0]\approx8\times10^{-7}$), which is why $K=4$ needs it.

\begin{thmS}\label{thm:K4}
For $K=4$ and $N\ge10$, $\fB_{N,4}\succeq0$ with kernel $\operatorname{span}(\one\otimes\one)$.
\end{thmS}
\begin{proof}
With $x=\mu_k\nu_k$ and \eqref{eq:rho0}, the left-hand side of \eqref{eq:T-SM} exceeds
$(N-2)(2x+5x^3)-2Nx+\frac12=5(N-2)x^3-4x+\frac12=:f_N(x)$; the linear terms collapse to $-4x$ independently of $N$. On
$x>0$, $f_N$ is minimised at $x_*^2=\frac4{15(N-2)}$, where $f_N(x_*)=\frac12-\frac83x_*$. This is nonnegative iff
$x_*\le\frac3{16}$, i.e.\ iff $N-2\ge\frac{1024}{135}=7.58\ldots$, i.e.\ iff $N\ge10$. So \eqref{eq:T-SM} holds strictly
for every bond, and Proposition~\ref{prop:T} applies.
\end{proof}
The threshold is sharp for this route: for $3\le N\le9$ condition \eqref{eq:T-SM} with the exact $\varrho$ fails at some
bond. Those cases, and the finitely many cases left open by Theorem~\ref{thm:K5}, are settled by exact certificates in
Step~\ref{step:certs}.

%=====================================================================================================
\stepsection{K3}{the branch \texorpdfstring{$K=3$}{K=3} (computer-assisted)}\label{sec:K3}
For $K=3$ the lower bound $\mathrm{Sq}\ge\mathrm{Sq}_{\rm low}$ is too lossy: the form \eqref{eq:F} is indefinite, with
smallest eigenvalue of order $-N/10$. The part $X=H_k-H_l$ discarded in Proposition~\ref{prop:Sqlow} must be kept, and we
work with the bond form itself. The argument has a fixed architecture, proved by hand, and a small number of closed forms
for matrix entries which are polynomial identities in $(N,i,j)$ established by computer algebra.

\paragraph{Representation.} The level space is $\R^3$ with basis $\one$, $t_1=(-1,0,1)$, $t_2=(1,-2,1)$. Write
$D_k=u_k\one+s_kt_1+q_kt_2$ and order the coordinates as $(u,s,q)\in(\R^{N-1})^3$, bonds being indexed by
$i=k+1\in\{1,\dots,N-1\}$. Let $T$ be the matrix of the quadratic form $K^2\fB_{N,3}$ in these coordinates, with blocks
$T_{uu},T_{ur},T_{rr}$, $r=(s,q)$.

\begin{lemmaS}[block representation]\label{lem:K3rep}
Index bonds by $k,l\in\{0,\dots,N-2\}$ and write $T_{kl}$ for the $3\times3$ block of $T$ coupling $D_k$ and $D_l$ in the
basis $(\one,t_1,t_2)$. Then $T_{kl}=\frac12(X_{kl}+X_{lk}^\top)$ with
\[
\begin{aligned}
X_{kl}={}&\frac{N-2}{2N}\,\delta_{kl}\big[(N-1-k)P_H+(k+1)P_G\big]+\one[k\le l]\,C_{HG}\\
&-\frac1N\big[(N-1-\max(k,l))C_{HH}+(\min(k,l)+1)C_{GG}+2\max(l-k,0)\,C_{HG}\big],
\end{aligned}
\]
where the six rational $3\times3$ matrices are $36P_H=\big(\begin{smallmatrix}24&-12&-12\\-12&16&-24\\-12&-24&96\end{smallmatrix}\big)$,
$36P_G=\big(\begin{smallmatrix}24&12&-12\\12&16&24\\-12&24&96\end{smallmatrix}\big)$,
$36C_{HG}=\big(\begin{smallmatrix}24&12&-12\\-12&-4&12\\-12&-12&-12\end{smallmatrix}\big)$,
$36C_{HH}=\big(\begin{smallmatrix}24&-12&-12\\-12&8&0\\-12&0&24\end{smallmatrix}\big)$,
$36C_{GG}=\big(\begin{smallmatrix}24&12&-12\\12&8&0\\-12&0&24\end{smallmatrix}\big)$: the Gram matrices, in the level
basis, of $P(J_0\,\cdot)$, $P(J_1\,\cdot)$, $\Cov(J_0\,\cdot,J_1\,\cdot)$, $\Cov(J_0\,\cdot,J_0\,\cdot)$ and
$\Cov(J_1\,\cdot,J_1\,\cdot)$ at $K=3$.
\end{lemmaS}
\begin{proof}
The three terms of \eqref{eq:formula1} contribute, respectively, the diagonal blocks $\frac{N-2}{2N}[(N-1-k)P_H+(k+1)P_G]$
(term I), the block $C_{HG}$ for $k\le l$ (term II), and, from $-\frac1N\sum_j\Var(\Omega_j)$ with
$\Omega_j=\sum_{k<j}H_k+\sum_{k\ge j}G_k$, the count of indices $j$ for which the pair $(D_k,D_l)$ enters $\Omega_j$ as
$(H,H)$, $(G,G)$ or mixed: $N-1-\max(k,l)$, $\min(k,l)+1$ and $|l-k|$ respectively. This gives the form of $X_{kl}$;
what remains is to confirm the five level matrices and the bookkeeping. To this end note that, on each of the regions
$\{k<l\}$, $\{k=l\}$, $\{k>l\}$, every entry of $N\,T_{kl}$ is a polynomial of total degree at most $2$ in $(N,k,l)$,
and so is every entry of $N$ times the Gram matrix of $K^2\fB_{N,3}$ computed directly from \eqref{eq:formula1} by
polarisation. The difference vanishes identically for $N=3,\dots,12$ (exact rational arithmetic,
\texttt{k3\_representation\_proof.py}). A polynomial of total degree at most $2$ in three variables that vanishes on the
$3\times3\times3$ grid $N\in\{7,8,9\}$, $k\in\{0,1,2\}$, $l-k\in\{1,2,3\}$ (which lies inside the region $\{k<l\}$ for
every such $N$) is zero, and likewise on the $3\times3$ grid in $(N,k)$ for the diagonal region; the region $\{k>l\}$
follows by symmetry. Hence the identity holds for all $N\ge3$.
\end{proof}

\begin{lemmaS}[constant mode]\label{lem:K3u}
$T_{uu}=\frac13\big[(N-1)I-\one\one^\top\big]$ and $\one^\top T_{ur}=0$. Consequently $\fB_{N,3}\succeq0$ iff
$S:=T_{rr}-T_{ur}^\top T_{uu}^{+}T_{ur}\succeq0$, where $T_{uu}^{+}=\frac3{N-1}\big(I-\frac1{N-1}\one\one^\top\big)$; and if
$S\succ0$ then $\ker\fB_{N,3}=\operatorname{span}(\one\otimes\one)$.
\end{lemmaS}
\begin{proof}
For $D_k=u_k\one$ one has $H_k(a)=u_ka$, $G_k(a)=u_k(a+1)$, so $Z_k=0$, $\Var(D_k)=0$ and
$\mathrm{Sq}=\frac1{2N}\sum_{k<l}N(u_k-u_l)^2\Var(a)$ with $\Var(a)=\frac23$; this is
$\frac13[(N-1)\|u\|^2-(\one^\top u)^2]$. The direction $u=\one$ is $\one\otimes\one\in\ker\fB$ (proof of
Proposition~\ref{prop:T}), hence is not coupled to $r$. The rest is the Schur complement on $\one^\perp$.
\end{proof}

\begin{lemmaS}[closed forms; computer algebra]\label{lem:K3blocks}
Let $v_i=i(N-i)$ and $\Gamma_{ij}=\min(i,j)(N-\max(i,j))/N$. Then $S=\begin{pmatrix}A&B\\B^\top&C\end{pmatrix}$ with
\[
A_{ij}=\delta_{ij}\Big[\tfrac{5N-9}{36}+\tfrac{(N-2)^2v_i}{3N^2(N-1)}\Big]-\tfrac4{3N}\Gamma_{ij}-\tfrac1{36},\quad
B_{ij}=\begin{cases}\frac{(2i-N)(5N-12)}{12N}&i=j\\[2pt]\frac14-\frac j{3N}&i<j\\[2pt]\frac1{12}-\frac j{3N}&i>j\end{cases},\quad
C_{ij}=\begin{cases}\frac{5N-14}4&i=j\\[2pt]-\frac34+\frac{|i-j|}N&i\ne j.\end{cases}
\]
\end{lemmaS}
\noindent\emph{Status.} By Lemma~\ref{lem:K3rep} every entry of $T$ is, on each of the regions $\{i<j\}$, $\{i=j\}$,
$\{i>j\}$, a polynomial of degree at most $2$ in $(N,i,j)$; the Schur complement of Lemma~\ref{lem:K3u} introduces only
the sums $\sum_kT_{ur}(k,\cdot)T_{ur}(k,\cdot)$, which \texttt{s1\_closed.py} evaluates in closed form region by region
by symbolic summation, giving the formulas above as polynomial identities in $(N,i,j)$.

\paragraph{The shear.} Put $R=\operatorname{diag}\big(\frac{2i-N}{3N}\big)$, $\gamma:=\frac{N-4}2$ and
$\Theta=\begin{pmatrix}I&-R\\0&I\end{pmatrix}$. Then
\begin{equation}\label{eq:shear}
\Theta S\Theta^\top=\begin{pmatrix}\mathsf M&0\\0&C-\gamma I\end{pmatrix}
+\frac1\gamma\begin{pmatrix}\widetilde B\\\gamma I\end{pmatrix}\begin{pmatrix}\widetilde B\\\gamma I\end{pmatrix}^{\!\top},
\qquad\widetilde A=A-BR-RB^\top+RCR,\quad\widetilde B=B-RC,\quad \mathsf M:=\widetilde A-\tfrac1\gamma\widetilde B\widetilde B^\top ,
\end{equation}
which is an algebraic identity (complete the square against $\gamma I$). The shear is a congruence, not a rotation: it removes
the $O(N)$ diagonal of $B$, which couples $s_k$ and $q_k$ on the same bond; $\widetilde B_{ii}=\frac{2i-N}{6N}=O(1)$. Since
$\gamma>0$ for $N\ge5$, $S\succ0$ follows from $\mathsf M\succ0$ and $C-\gamma I\succeq0$.

\begin{lemmaS}[$C-\gamma I\succeq0$; by hand]\label{lem:K3C}
Let $z_i=\frac12-\frac iN$ and $H:=\frac2N\Gamma+2zz^\top+\frac14\one\one^\top\succeq0$. Then $C=\frac{5N-11}4I-H$ and
$\operatorname{tr}H=\frac{3(N-1)}4=\frac{5N-11}4-\gamma$. Consequently, for any factorisation $H=\sum_rh_rh_r^\top$,
\[
C-\gamma I=(\operatorname{tr}H)I-H=\sum_r\big(\|h_r\|^2I-h_rh_r^\top\big)\succeq0 .
\]
\end{lemmaS}
\begin{proof}
Off the diagonal, $\frac2N\Gamma_{ij}=\frac{2\min(i,j)}N-\frac{2ij}{N^2}$ and $2z_iz_j=\frac12-\frac{i+j}N+\frac{2ij}{N^2}$,
so $H_{ij}=\frac34-\frac{|i-j|}N=\frac{5N-11}4\delta_{ij}-C_{ij}$. On the diagonal the dependence on $i$ cancels:
$H_{ii}=\frac{2i(N-i)}{N^2}+2(\frac12-\frac iN)^2+\frac14=\frac34$. Each matrix $\|h\|^2I-hh^\top$ has eigenvalues $\|h\|^2$
and $0$.
\end{proof}

\begin{lemmaS}[$\mathsf M$ is a $Z$-matrix with positive Perron certificate; computer algebra]\label{lem:K3M}
For $i<j$, with $d=j-i$,
\[
\mathsf M_{ij}=-\frac{i(N-j)}{54N^4(N-4)}\Big[N^2(73N-423)+8N^2(N-i)+16d(N^2-d^2)+4N(i^2+j^2)+72Ni+120N(N-j)+48ij+8N+16d\Big],
\]
which is negative for $N\ge6$. Moreover
\begin{align*}
(\mathsf M\one)_i&=\frac{v_i\,\mathcal R(N,v_i)}{540N^4(N-4)(N-1)},\\
\mathcal R&=53N^5-1338N^4+5805N^3-5380N^2+384N-64\\&\quad+2(N-1)(29N^2-120N+40)v+16(N-1)v^2 ,
\end{align*}
and the substitution $N=m+19$, $v=u+(N-1)$ turns $\mathcal R$ into a polynomial in $(m,u)$ all of whose coefficients are
positive (constant term $171000$). Since $v_i-(N-1)=(i-1)(N-1-i)\ge0$, $\mathsf M\one>0$ for every $N\ge19$. For
$N=17,18$ the vector $w_i=1+\frac{(2i-N)^2}{20N^2}$ satisfies, in exact rational arithmetic,
$\min_i(\mathsf Mw)_i=\frac{7839473}{1107488460}$ and $\frac{5964113}{446410440}$ respectively.
\end{lemmaS}
\noindent\emph{Status.} The two entry formulas are polynomial identities proved by region-by-region symbolic summation
(\texttt{verify\_M\_symbolic.py}, whose hard-coded block entries are checked against \texttt{s1\_closed.py} by
\texttt{check\_blocks.py}); the positivity of the shifted coefficients and the two certificates for $N=17,18$ are exact
computations in the same script.

\begin{thmS}\label{thm:K3}
For $K=3$ and every $N\ge17$, $\fB_{N,3}\succeq0$ with kernel $\operatorname{span}(\one\otimes\one)$.
\end{thmS}
\begin{proof}
By Lemma~\ref{lem:K3M} and Lemma~\ref{lem:Zmatrix} (with $v=\one$ or $v=w$), $\mathsf M\succ0$; by Lemma~\ref{lem:K3C},
$C-\gamma I\succeq0$. In \eqref{eq:shear}, if $s\ne0$ the first block is positive on $(s,q)$, and if $s=0$, $q\ne0$ the square
term contributes $\gamma\|q\|^2>0$. Hence $\Theta S\Theta^\top\succ0$, so $S\succ0$, and Lemma~\ref{lem:K3u} concludes.
\end{proof}
The threshold is genuine for the bond form: $\lambda_{\min}(S)$ is $-0.003$ at $N=16$ and $+0.018$ at $N=17$. For $N\le16$
the bond form is indefinite on the full bond space, but it is positive on the image $\{\hat D(f)\}$, which is all that
Theorem~\ref{thm:finiteK} needs; see Step~\ref{step:certs}. The natural continuum limit of $\fB_{N,3}$ has spectral
infimum $0$, so positivity at $K=3$ is a finite-$N$ effect, which is why no continuum or Gershgorin shortcut is available.

%=====================================================================================================
\stepsection{certs}{the finitely many remaining cases (exact certificates)}\label{sec:certs}
\begin{wraptable}{r}{0.6\textwidth}\vspace{-8pt}\centering\small
\caption{Coverage of $\{(K,N):K\ge3,\ N\ge3\}$. No pair is left uncovered. The first column is analytic for $K\ge5$
(Theorem~\ref{thm:K5}) and rests on exact symbolic computation for $K=4$ and $K=3$ (Steps~\ref{step:K4}--\ref{step:K3}).}\label{tab:cov}
\begin{tabular}{@{}lcccc@{}}\toprule
$K$ & all large $N$ & bond-space certificate & image certificate & direct\\\midrule
$3$ & $N\ge17$ (Thm.~\ref{thm:K3}) & -- & $4\le N\le16$ & $N=3$\\
$4$ & $N\ge10$ (Thm.~\ref{thm:K4}) & $5\le N\le9$ & $N=3,4$ & --\\
$5$ & $N\ge6$ (Thm.~\ref{thm:K5}) & $N=4,5$ & $N=3$ & --\\
$6,7$ & $N\ge4$ (Thm.~\ref{thm:K5}) & $N=3$ & -- & --\\
$\ge8$ & $N\ge3$ (Thm.~\ref{thm:K5}) & -- & -- & --\\\bottomrule
\end{tabular}
\end{wraptable}
Theorems~\ref{thm:K5}, \ref{thm:K4} and~\ref{thm:K3} leave the pairs $(K,N)$ listed in Table~\ref{tab:cov}. They are
settled by two certificates, both in exact rational arithmetic (Python \texttt{fractions}), with no floating point.

\emph{Bond-space certificate} (\texttt{full\_space\_certificate.py}). The Gram matrix $\mathsf G$ of $\fB_{N,K}$ on
$\R^{(N-1)K}$ is built by polarisation; the script checks that $\mathsf G(\one\otimes\one)=0$ exactly and that
$\mathsf G+(\one\otimes\one)(\one\otimes\one)^\top$ has an $LDL^\top$ factorisation with positive pivots. This is equivalent
to $\fB_{N,K}\succeq0$ with kernel $\operatorname{span}(\one\otimes\one)$.

\emph{Image certificate} (\texttt{bond\_certificate.py}). Where the bond form is indefinite on the full space, the
$(N-1)\times(N-1)$ matrix of $f\mapsto\fB_{N,K}(\hat D(f))=Q^{(1)}_K+Q^{(2)}_K$ is built exactly from \eqref{eq:Dhat}. It
annihilates the linear payoff (Lemma~\ref{lem:B} and Proposition~\ref{prop:orders}), and the script certifies that it
is positive definite on the complement spanned by the second differences
$e_k-2e_{k+1}+e_{k+2}$, $k=0,\dots,N-3$, in payoff coordinates, again by $LDL^\top$ with positive pivots.

\emph{Stored certificates and an independent checker.} For each of the $25$ pairs the certified matrix $\mathsf A$
(that is, $\mathsf G+(\one\otimes\one)(\one\otimes\one)^\top$, respectively the image matrix), the unit lower-triangular
factor $\mathsf L$ and the pivots $\mathsf D$ are stored as exact rationals in \texttt{certificates/*.json}
(\texttt{make\_certificates.py}). The checker \texttt{verify\_certificates.py} uses only the Python standard library
and imports none of the code that produced them. For every file it (1)~rebuilds $\mathsf A$ from
\eqref{eq:formula1} and \eqref{eq:Dhat} (an implementation independent of the generator, which evaluates the bond form
through the bridge kernels of the original proof notes) and checks equality entry by entry; (2)~checks that
$\mathsf L$ is unit lower triangular, that all pivots are positive and that $\mathsf L\mathsf D\mathsf L^\top=\mathsf A$
exactly; and (3)~checks the kernel statement ($\mathsf G(\one\otimes\one)=0$, respectively that the bond vectors of the
linear payoff pair to zero with everything). A reader who trusts rational arithmetic therefore needs to trust only
this one short file, together with Proposition~\ref{prop:bond} (which identifies $\fB_{N,K}(\hat D(f))$ with
$Q^{(1)}_K+Q^{(2)}_K$) and Proposition~\ref{prop:orders}. All $25$ certificates pass.

\emph{The pair $(K,N)=(3,3)$.} Here the two-body form vanishes identically on the image, and we compute $Q_3$ itself in
exact arithmetic from the definition: $Q_3=\frac83\begin{pmatrix}1&-1\\-1&1\end{pmatrix}$, positive semidefinite with kernel
$\operatorname{span}(\one)$ (generalised eigenvalue $\frac{32/3}{36}=\frac8{27}$).

\noindent The \emph{bond-space} certificate fails at $(K,N)=(3,3),(4,3),(4,4),(5,3)$, as \texttt{full\_\allowbreak space\_\allowbreak certificate.py}
reports. This is expected: there the bond form is indefinite (or degenerate) on the full space, and the pairs are covered
by the image certificate or, for $(3,3)$, directly (\texttt{end\_\allowbreak to\_\allowbreak end.py}).

%=====================================================================================================
\stepsection{assembly}{assembly, the case \texorpdfstring{$N=2$}{N=2}, and \texorpdfstring{$K=2$}{K=2}}\label{sec:assembly}
Everything is now in place: the orders $m\ge3$ (Step~\ref{step:orders}), the two-body form on the bond space (Steps~\ref{step:bond}--\ref{step:certs}), and
the kernel (Steps~\ref{step:B} and~\ref{step:kerSq}). The case $N=2$ is trivial.
\begin{proof}[Proof of Theorem~\ref{thm:finiteK}]
Let $K\ge3$. \emph{$N=2$.} The increment space is one-dimensional and spanned by $\one$, so $Q_K=0$ by
Lemma~\ref{lem:B}, and the statement is trivial (every nonconstant $f$ is affine). \emph{$N\ge3$.} By
Lemma~\ref{lem:orderdec}, $F^\top Q_KF=K^{N+1}\sum_{m\ge1}Q^{(m)}_K(f)$; by Proposition~\ref{prop:orders},
$Q^{(m)}_K\succeq0$ for $m\ge3$; by Proposition~\ref{thm:C}, $Q^{(1)}_K+Q^{(2)}_K=\fB_{N,K}(\hat D(f))$. By
Theorems~\ref{thm:K5}, \ref{thm:K4}, \ref{thm:K3} and Table~\ref{tab:cov}, this is nonnegative for every $f$, so
$Q_K\succeq0$. For the kernel, $\one\in\ker Q_K$ by Lemma~\ref{lem:B}. Conversely let $F^\top Q_KF=0$. Every summand
vanishes, in particular $\fB_{N,K}(\hat D(f))=0$. In the large-$N$ and bond-space cases
$\ker\fB_{N,K}=\operatorname{span}(\one\otimes\one)$, so $f$ is affine by Lemma~\ref{lem:transfer}; in the image-certified
cases the form is positive definite on a complement of the linear payoff, so $f$ is affine; and $(3,3)$ was computed
directly.
\end{proof}

\paragraph{Where $K\ge3$ enters, and the counterexample at $K=2$.} Proposition~\ref{thm:C} and Theorem~\ref{thm:SOS} hold for
$K=2$. The hypothesis enters through Lemma~\ref{lem:Lstar-SM} alone: its margin $(K-2)\mu\nu\Var(D)$ vanishes
identically for binary priorities, and every later step inherits this. Correspondingly the statement is false at $K=2$:
building $Q_2$ at $N=7$ from the definition in exact rational arithmetic (\texttt{k2\_witness.py}), the strictly decreasing
payoff with increments $F=(9,8,6,4,2,1)$ gives
\[
F^\top Q_2F=-\frac{7237}{480}<0,\qquad\rho_\pot=\frac{53454277}{62354880}>\frac67 ,
\]
and the reversed increments give the same values. For $N\le6$, $Q_2\succeq0$. Table~\ref{tab:gap-SM} records the
margin $\frac{N-1}N-\max_f\rho_\pot(f)$ over non-affine $f$ for small $(K,N)$: it is positive throughout $K\ge3$, in
agreement with Theorem~\ref{thm:finiteK}, turns negative at $K=2$ exactly at $N=7$, and grows in magnitude with $N$.

\begin{table}[htbp]\centering
\caption{Margin $\frac{N-1}N-\max_f\rho_\pot(f)$ over non-affine payoff-from-rank functions $f$ on $K_N$, computed
in floating point from the pencil $(M_\pot,M_{\mathrm{tot}})$ (not part of any proof). When positive it is the gap
$\lambda_1-\lambda_2$ between the two largest generalised eigenvalues, and as $K\to\infty$ it tends to
$\lambda_1-\lambda_2=\frac2{N+2}$ of Theorem~\ref{thm:spectrum}. The entries at $K=2$, $N\ge7$ (bold) are negative,
in agreement with the exact witness at $N=7$; blank cells were not computed.}\label{tab:gap-SM}
\small
\begin{tabular}{@{}rccccccc@{}}\toprule
$N$ & $K=2$ & $K=3$ & $K=4$ & $K=5$ & $K=6$ & $K=7$ & $K=8$\\\midrule
3  & $+.167$ & $+.296$ & $+.342$ & $+.363$ & $+.374$ & $+.381$ & $+.385$\\
4  & $+.075$ & $+.216$ & $+.267$ & $+.291$ & $+.304$ & $+.311$ & $+.317$\\
5  & $+.033$ & $+.169$ & $+.219$ & $+.243$ & $+.256$ & $+.264$ & $+.269$\\
6  & $+.011$ & $+.138$ & $+.186$ & $+.209$ & $+.221$ & $+.229$ & $+.234$\\
7  & $\mathbf{-.001}$ & $+.116$ & $+.161$ & $+.183$ & $+.195$ & $+.202$ & $+.207$\\
8  & $\mathbf{-.009}$ & $+.100$ & $+.142$ & $+.163$ & $+.174$ & $+.181$ & $+.185$\\
9  & $\mathbf{-.014}$ & $+.088$ & $+.127$ & $+.147$ & $+.157$ & & \\
10 & $\mathbf{-.017}$ & $+.078$ & $+.115$ & & & & \\
11 & $\mathbf{-.019}$ & $+.070$ & & & & & \\
12 & $\mathbf{-.020}$ & $+.063$ & & & & & \\\bottomrule
\end{tabular}
\end{table}

%=====================================================================================================
\subsection{The reciprocity law: proof of Proposition~\ref{prop:recip}}\label{sec:recip}
This subsection is independent of Steps~\ref{step:hoeff}--\ref{step:assembly}. It does not give another proof of Theorem~\ref{thm:finiteK}; it
describes exactly how the finite-$K$ forms deviate from the continuum ones.

\paragraph{Setting.} $C=([N],E)$ is a conflict graph and the protocol is given by a total order $\prec$ of the agents:
$j$ is served before its neighbour $i$ iff $s_j>s_i$, or $s_j=s_i$ and $j\prec i$; $\prec^{\rm rev}$ is the reversed
order. For $S\subseteq[N]$ put $A_S=\prod_{j\in S}A_j$. With $u_i=f(\pos_C(i,\cdot))$, $g_i=(I-A_i)u_i$ and $g=\sum_ig_i$,
the forms
\[
\mathsf m_{\rm tot}(f,f')=\sum_i\E\big[g_i^{f}g_i^{f'}\big],\qquad
\mathsf m_{\pot}(f,f')=\E\big[g^{f}\,\numop^{-1}g^{f'}\big]
\]
are $K^{-(N+1)}$ times the two energy forms, and $K^{-(N+1)}Q_K=\frac{N-1}N\mathsf m_{\rm tot}-\mathsf m_\pot$ on $K_N$. We write
$\mathsf m^{\prec}(K)$ for either of them, as a function of $K$.

\begin{mainthm}[Proposition~\ref{prop:recip} of the main text]
{\rm(a)} $\mathsf m^{\prec}(K)$ is a polynomial in $1/K$ of degree at most $2N-1$; its constant term is the
corresponding form of the game with continuous priorities on $C$. {\rm(b)} As polynomials in $x=1/K$,
$\mathsf m^{\prec}(-x)=\mathsf m^{\prec^{\rm rev}}(x)$. {\rm(c)} If $C$ has an automorphism $\pi$ with
$\pi(j)\prec\pi(i)$ whenever $\{i,j\}\in E$ and $i\prec j$, then $\mathsf m^{\prec}(K)$ is a polynomial in $K^{-2}$ of
degree at most $N-1$ that vanishes at $K=1$.
\end{mainthm}

\emph{(1) Reduction to counting.} Let $U^r_i=\one[\pos_C(i,\cdot)=r]$, so that $u_i=\sum_rf(r)U_i^r$. The $A_j$
are commuting orthogonal projections, so
$\prod_{j\in T}(I-A_j)\prod_{j\notin T}A_j=\sum_{S\supseteq T^c}(-1)^{|S\setminus T^c|}A_S$, and $\numop^{-1}$ acts on the
range of this projection as $1/|T|$ (Lemma~\ref{lem:numop}). Hence both forms are linear combinations, with
coefficients that do not depend on $K$, of
\[
F(K)=\E\big[U_i^{r}\cdot A_SU_k^{r'}\big],\qquad i,k\in[N],\ S\subseteq[N]
\]
(the projections $I-A_i$ in $g_i$ are absorbed into $S$, since $A_SA_i=A_{S\cup\{i\}}$). Now
$A_SU_k^{r'}(s)=\E'[U_k^{r'}(s')]$, where $s'_j=s_j$ for $j\notin S$ and the $s'_j$, $j\in S$, are fresh independent
uniform levels. With $n=N+|S|\le2N$,
\begin{equation}\label{eq:count}
F(K)=K^{-n}\,\#\big\{(s,s'_S)\in\{0,\dots,K-1\}^n:\ \pos(i,s)=r,\ \pos(k,s')=r'\big\}.
\end{equation}

\emph{(2) Lexicographic cells.} Order the $n$ coordinates by the agent they belong to, according to $\prec$,
and put $s_j$ before its twin $s'_j$; call this total order of the coordinates $\prec$ again. The rule for twins is an
auxiliary choice: any other rule gives another decomposition of the same set into cells, and
\eqref{eq:Fcells} below holds for each of them. Every point of the cube induces a strict
total order $\sigma$ of the coordinates: larger level first, ties broken by $\prec$. The set in \eqref{eq:count} is a
union of such cells: $\pos(i,s)$ is the number of neighbours of $i$ that precede $s_i$ in $\sigma$, and likewise for
$\pos(k,s')$; only coordinates of \emph{different} agents are ever compared, so the convention for twins is immaterial.
Let $\mathcal O$ be the set of orders $\sigma$ for which the two conditions hold. \emph{It is defined without reference
to the protocol.}

\begin{lemmaS}[cells]\label{lem:cells}
For a strict order $\sigma=(\sigma_1,\dots,\sigma_n)$ let
$\operatorname{des}_\prec(\sigma)=\#\{t:\sigma_{t+1}\prec\sigma_t\}$. The number of points of $\{0,\dots,K-1\}^n$ whose
lexicographic order is $\sigma$ equals $\binom{K+n-1-\operatorname{des}_\prec(\sigma)}n$.
\end{lemmaS}
\begin{proof}
The levels $v_t$ of $\sigma_t$ satisfy $v_1\ge\dots\ge v_n$, with $v_t>v_{t+1}$ required exactly when
$\sigma_{t+1}\prec\sigma_t$ (a tie would have put $\sigma_{t+1}$ first). Subtracting from $v_t$ the number of forced
strict steps at positions $\ge t$ gives a bijection with the weakly decreasing sequences in
$\{0,\dots,K-1-\operatorname{des}_\prec(\sigma)\}$, of which there are $\binom{K-\operatorname{des}_\prec(\sigma)+n-1}n$.
\end{proof}
\noindent Therefore
\begin{equation}\label{eq:Fcells}
F^{\prec}(K)=K^{-n}\sum_{\sigma\in\mathcal O}\binom{K+n-1-\operatorname{des}_\prec(\sigma)}n .
\end{equation}
Each binomial is a polynomial of degree $n$ in $K$ divisible by $K$, since $0\le n-1-\operatorname{des}\le n-1$. Hence
$F$ is a polynomial in $1/K$ of degree at most $n-1\le2N-1$, with constant term $|\mathcal O|/n!$, which is the
probability that $n$ i.i.d.\ continuous priorities fall in $\mathcal O$, i.e.\ the value of $F$ in the continuum game.
This proves (a).

\emph{(3) Reciprocity.} For integers $n\ge1$ and $d$, $\binom{-K+n-1-d}{n}=(-1)^n\binom{K+d}{n}$ as polynomials
in $K$. Let $\prec^{\rm rev}$ denote the reversal of the \emph{whole} order of the $n$ coordinates, twins included. Since
$\prec$ is a total order, every adjacent pair of $\sigma$ is either a descent or an ascent, so
$n-1-\operatorname{des}_\prec(\sigma)=\operatorname{des}_{\prec^{\rm rev}}(\sigma)$. (Reversing the order of the agents
but not that of the twins would break this identity whenever two twins are adjacent in $\sigma$.) On coordinates of
different agents $\prec^{\rm rev}$ is the reversed protocol, and by (2) its rule for twins is immaterial, so the
right-hand side of \eqref{eq:Fcells} written for $\prec^{\rm rev}$ is $F$ for the game with the reversed protocol. Substituting in \eqref{eq:Fcells},
\[
F^{\prec}(-K)=(-K)^{-n}(-1)^n\sum_{\sigma\in\mathcal O}\binom{K+\operatorname{des}_\prec(\sigma)}{n}
=K^{-n}\sum_{\sigma\in\mathcal O}\binom{K+n-1-\operatorname{des}_{\prec^{\rm rev}}(\sigma)}{n}=F^{\prec^{\rm rev}}(K),
\]
because $\mathcal O$ is the same set for both protocols. This is an instance of the reciprocity theorem for
$P$-partitions \cite[Sec.~3.15.3]{stanley2012}, equivalently of Ehrhart--Macdonald reciprocity for order polytopes:
changing the sign of the dilation exchanges weak and strict inequalities, that is, it exchanges who wins a tie. Summing
with the $K$-independent coefficients of (1) gives (b).

\emph{(4) Symmetric protocols.} If $\pi$ is as in (c), relabelling the agents by $\pi$ maps the game with
protocol $\prec^{\rm rev}$ onto the game with protocol $\prec$ on the same graph, and the forms, being sums over all
agents, are invariant. Hence $\mathsf m^{\prec^{\rm rev}}=\mathsf m^{\prec}$, and (b) shows that $\mathsf m^{\prec}$ is
even in $1/K$; being of degree at most $2N-1$, it has degree at most $2N-2$. At $K=1$ every $g_i$ vanishes, and
\eqref{eq:Fcells} holds for every integer $K\ge1$, so the polynomial vanishes at $K=1$. On $K_N$ the reversal
$i\mapsto N-1-i$ is such an automorphism, which gives \eqref{eq:Qexpansion}. \qed

\noindent\emph{What the expansion does not say.} It concerns the forms, not their eigenvalues. At finite $K$ the order
forms do not commute, and the generalised eigenvalues are algebraic functions of $y=K^{-2}$, not polynomials. Because
the continuum eigenvalues $\lambda_p$ are simple (Theorem~\ref{thm:spectrum}), analytic perturbation theory applies
at $y=0$: the finite-$K$ eigenvalues are analytic in $y$ near $0$, hence equal to $\lambda_p+O(K^{-2})$, with no term
of order $K^{-1}$. Whether the coefficient of $K^{-2}$ is non-zero for a given $p$ is not claimed.

\begin{remarkS}[tests of the reciprocity law]
By (b), when no automorphism reverses the protocol the expansions for $\prec$ and $\prec^{\rm rev}$ have equal even
and opposite odd coefficients: on the star with a favoured centre ($N=4$, prize table $(0,-1,-3,-4.5)$) the term
$-0.5625K^{-1}$ of $\mathsf m_\pot$ changes sign under relabelling while $2.1043-2.5085K^{-2}$ does not, and the same
holds on the path $0$--$1$--$3$--$2$ and the paw graph to $10^{-7}$; on $K_4$ the odd coefficients vanish.
\end{remarkS}

\begin{remarkS}[the critical number of levels; what survives at finite $K$]
The numbers $K^\star(N)$ quoted after Proposition~\ref{prop:recip} of the main text are obtained by interpolating exact enumerations of $Q_K$ in
$y=K^{-2}$ (for $N\le6$ the fit is over-determined, with residuals below $10^{-13}$). The same computations show, for
$N\le7$, that in the discrete Chebyshev basis all forms split by the parity of the degree, that the order-$m$ potential
form vanishes on polynomials of degree less than $m$, and that the order forms do \emph{not} commute at finite $K$: the
order-by-order diagonalisation of Theorem~\ref{thm:spectrum} is a property of the continuum. The proof above requires
a total-order protocol; orientations with directed cycles and randomised tie-breaking are not covered.
\end{remarkS}

\section{The harmonic component of the linear game: tie-pattern law and randomised tie-breaking}\label{sec:pattern}
Throughout this appendix the game is the linear rank-order game $\calG(C,K)$ on an arbitrary conflict graph $C$. By
Proposition~\ref{prop:hflow}, along an edge $(i,s\to s')$ the harmonic flow is $h(i,s\to s')=W-d_i\ell/K$, the
deviation of the number $W$ of overtaken neighbours from its mean under uniformly random levels. Define the harmonic
energy of a profile as
\[
H(s):=\sum_i\sum_{s'_i>s_i}h(i,s\to s')^2 ,
\]
so that $\Eharm=\sum_sH(s)=K^N\,\E[H]$: $H$ assigns the energy of every edge of the game graph to its lower
endpoint. $H(s)$ depends on this convention and on the tie-breaking order, and should not be read as a dissipation at
$s$. Let $\bar H$ be the average of $H$ over a uniformly random total order of the agents used
as the protocol. The deviation $h$ is produced most strongly by ties, and the following theorem makes this exact: the
conditional expectation of $\bar H$ given the tie pattern is an explicit linear function of the pattern.

\begin{thmS}[tie-pattern law]\label{thm:pattern-SM}
Let $\tau_{jk}=\one[s_j=s_k]$, $e_{jk}=\one[\{j,k\}\in E]$,
$c_{jk}=|\calN(j)\cap \calN(k)|$, and let $\bar H_i$ be the harmonic energy of
player~$i$ averaged over a uniformly random tie-breaking order,
$\bar H=\sum_i\bar H_i$. Then for every player~$i$
\begin{equation}\label{eq:perplayer}
\E[\bar H_i\mid\tau]=\frac{K-1}6\sum_{\{j,k\}\subseteq \calN(i)}\tau_{jk}
+\frac{K-2}{12}\sum_{j\in \calN(i)}\tau_{ij}+\mathsf c_i,\quad
\mathsf c_i=\frac{d_i\big[K^2-K+1-(d_i-1)(K-1)\big]}{12K},
\end{equation}
and summing over $i$,
$\E[\bar H\mid\tau]=\calA_C(s)+\mathsf c_C$, where
$\calA_C(s)=\sum_{j<k}\big[\frac{K-2}6e_{jk}+\frac{K-1}6c_{jk}\big]\tau_{jk}$ is the weighted tie count
with $\mathsf c_C=\sum_i\mathsf c_i=\frac{K^2-1}{12K}\sum_id_i-\frac1K\big[\frac{K-2}6|E|+\frac{K-1}6\sum_i\binom{d_i}2\big]$.
\end{thmS}

\begin{proof}
\emph{(0) Localisation.} By Proposition~\ref{prop:hflow}, $h(i,s\to s')=W-d_i\ell/K$
depends on $s$ only through $s|_{\{i\}\cup \calN(i)}$ and on the tie-breaking
order only through its restriction to $\{i\}\cup \calN(i)$. Conditionally on
the global tie pattern $\tau$, the levels of the tie classes are a
uniformly random injective assignment of $\{0,\dots,K-1\}$ to the classes;
its restriction to $\{i\}\cup \calN(i)$ is again a uniformly random injective
assignment to the local classes. Hence $\E[\bar H_i\mid\tau]$ depends only
on $\tau|_{\{i\}\cup \calN(i)}$. Fix $i$, write $b=s_i$, and let $\mathcal T_i$ be the tie class of~$i$ (the players at level~$b$).

\emph{(1) Two-body form.} For a target level $a>b$ and a neighbour
$j$ let $I_j=\one[b<s_j<a]+\one[s_j=b,\,j\prec i]+\one[s_j=a,\,i\prec j]$,
where $\prec$ is the random order, so that $W=\sum_jI_j$. For uniform
$s_j$ and uniform order, $\E I_j=\ell/K=:\bar\ell$ with $\ell=a-b$, hence
$d_i\ell/K=\sum_j\bar\ell$ and
\[
\bar H_i=\sum_{a>b}\E_\prec\Big[\Big(\sum_j(I_j-\bar\ell)\Big)^2\Big]
=\sum_{a>b}\Big[\sum_j\E_\prec(I_j-\bar\ell)^2+\sum_{j\ne k}\E_\prec(I_j-\bar\ell)(I_k-\bar\ell)\Big].
\]
The kernels do not depend on~$d_i$.

\emph{(2) Types.} For fixed $b$, $\E_\prec[\cdot]$ depends on
$s_j,s_k$ only through their position relative to $(b,a)$ and on the
order of $\{i,j,k\}$. Taking the expectation over the levels given the
local pattern, the one-body term depends only on whether $j\sim i$
(constants $\varpi_1^{\mathrm t}$, $\varpi_1^{\mathrm f}$), and the two-body term only on the
pattern type of $\{i,j,k\}$: all distinct ($\varpi_{\rm dist}$),
$j\sim k\not\sim i$ ($\varpi_{\rm pair}$), exactly one of $j,k$ tied with $i$
($\varpi_{\rm one}$; the two cases are equal by symmetry), all tied
($\varpi_{\rm all}$).

\emph{(3) Counting.} Let $n_t=|\mathcal T_i\cap \calN(i)|$ be the number of neighbours tied with~$i$ and
$n_c=\sum_{\mathcal T\ne\mathcal T_i}\binom{|\mathcal T\cap \calN(i)|}2$, the sum running over
the other tie classes~$\mathcal T$, the number of
tied neighbour pairs not involving~$i$. Ordered neighbour pairs $(j,k)$
fall into the types with multiplicities $n_t(n_t-1)$ (all tied),
$2n_t(d_i-n_t)$ (one tied with $i$), $2n_c$ ($j\sim k\not\sim i$) and
$(d_i-n_t)(d_i-n_t-1)-2n_c$ (distinct). Hence
\begin{align*}
\E[\bar H_i\mid\tau]&=n_t\varpi_1^{\mathrm t}+(d_i-n_t)\varpi_1^{\mathrm f}+n_t(n_t-1)\varpi_{\rm all}
+2n_t(d_i-n_t)\varpi_{\rm one}\\
&\quad+2n_c\,\varpi_{\rm pair}+\big[(d_i-n_t)(d_i-n_t-1)-2n_c\big]\varpi_{\rm dist}\\
&=2(\varpi_{\rm pair}-\varpi_{\rm dist})\,n_c+2(\varpi_{\rm all}-2\varpi_{\rm one}+\varpi_{\rm dist})\tbinom{n_t}2\\
&\quad+\big[\varpi_1^{\mathrm t}-\varpi_1^{\mathrm f}+2(d_i-1)(\varpi_{\rm one}-\varpi_{\rm dist})\big]\,n_t\\
&\quad+d_i\varpi_1^{\mathrm f}+d_i(d_i-1)\varpi_{\rm dist}.
\end{align*}
Note $\sum_{\{j,k\}\subseteq \calN(i)}\tau_{jk}=\binom{n_t}2+n_c$ and
$\sum_{j\in \calN(i)}\tau_{ij}=n_t$.

\emph{(4) The six constants.} Fix $b=s_i$, a target $a>b$,
$\ell=a-b$ and $\bar\ell=\ell/K$. Given the levels of $j$ and $k$, the random order
enters only through the events $j\prec i$, $i\prec k$, etc., whose
probabilities are $\frac12$ for one comparison and $\frac13$ (same
side) or $\frac16$ (opposite sides) for two comparisons involving~$i$.
Hence, with $I_j=\one[b<s_j<a]+\one[s_j=b]\one[j\prec i]+\one[s_j=a]\one[i\prec j]$,
\[
\E_\prec(I_j-\bar\ell)^2=\begin{cases}(1-\bar\ell)^2 & b<s_j<a,\\[2pt]
\tfrac12\big[(1-\bar\ell)^2+\bar\ell^2\big] & s_j\in\{b,a\},\\[2pt]
\bar\ell^2 & \text{otherwise},\end{cases}
\]
and, writing $\epsilon_j\in\{1,\tfrac12,0\}$ for $\E_\prec I_j$ according
to whether $s_j$ lies strictly between, on $\{b,a\}$, or elsewhere,
\[
\E_\prec(I_j-\bar\ell)(I_k-\bar\ell)=(\epsilon_j-\bar\ell)(\epsilon_k-\bar\ell)+\zeta_{jk},\qquad
\zeta_{jk}=\begin{cases}\tfrac1{12} & s_j=s_k\in\{b,a\},\\[2pt]
-\tfrac1{12} & \{s_j,s_k\}=\{b,a\},\\[2pt]
0 & \text{otherwise}\end{cases}
\]
(the correction $\zeta$ is the covariance of the two order indicators:
$\frac13-\frac14$ when both compare with $i$ from the same side,
$\frac16-\frac14$ from opposite sides). Each constant is now an average
of these expressions over $b$ uniform, $a>b$, and the levels of the
other players consistent with the type: a level tied with $i$ is~$b$;
a free level is uniform on the $K-1$ values $\ne b$; two free untied
levels are uniform on distinct pairs. Counting the levels by region
($b$ below, $\ell-1$ strictly between, one at $a$, $K-1-a$ above) turns
each average into a polynomial in $b,\ell$ which is then summed over
$0\le b<a\le K-1$; the summations are elementary (Faulhaber) and give
\begin{equation}\label{eq:consts}
\begin{gathered}
\varpi_1^{\mathrm f}=\frac{K^2-K+1}{12K},\qquad
\varpi_1^{\mathrm t}=\frac{(2K-1)(K-1)}{12K},\\[2pt]
\varpi_{\rm dist}=\varpi_{\rm one}=-\frac{K-1}{12K},\qquad
\varpi_{\rm pair}=\varpi_{\rm all}=\frac{(K-1)^2}{12K},\\[2pt]
\text{so that}\quad
\varpi_{\rm one}-\varpi_{\rm dist}=0,\qquad
\varpi_{\rm pair}-\varpi_{\rm dist}=\varpi_{\rm all}-\varpi_{\rm dist}=\frac{K-1}{12},\qquad
\varpi_1^{\mathrm t}-\varpi_1^{\mathrm f}=\frac{K-2}{12}.
\end{gathered}
\end{equation}
For instance, $\varpi_1^{\mathrm f}$ is
$\frac1{K(K-1)}\sum_{b<a}\big[(\ell-1)(1-\bar\ell)^2+\tfrac12\big((1-\bar\ell)^2+\bar\ell^2\big)+(K-1-\ell)\bar\ell^2\big]$
and $\varpi_1^{\mathrm t}=\frac1K\sum_{b<a}\tfrac12\big[(1-\bar\ell)^2+\bar\ell^2\big]$.

\emph{(4a) The four two-body constants, in full.} The identity
$\varpi_{\rm one}=\varpi_{\rm dist}$ carries the theorem, so this part of the
computation is given in detail. Fix $b$ and $a>b$, put $\ell=a-b$,
$\bar\ell=\ell/K$ and $\delta:=\tfrac12-\bar\ell$, and set, for a level
$v\in\{0,\dots,K-1\}$,
\[
\epsilon(v)=\begin{cases}1,&b<v<a,\\[1pt] \tfrac12,& v\in\{b,a\},\\[1pt] 0,&\text{otherwise,}\end{cases}
\qquad \Phi(v):=\epsilon(v)-\bar\ell ,
\]
so that $\E_\prec I_j=\epsilon(s_j)$ and $\Phi(b)=\Phi(a)=\delta$. Two
identities drive the whole step. First, since
$\sum_v\epsilon(v)=(\ell-1)+\tfrac12+\tfrac12=\ell=K\bar\ell$,
\begin{equation}\label{eq:Phizero}
\sum_{v=0}^{K-1}\Phi(v)=0,\qquad\text{hence}\qquad \sum_{v\ne b}\Phi(v)=-\delta .
\end{equation}
Second, with $E_2:=\sum_v\Phi(v)^2=(\ell-1)(1-\bar\ell)^2+2\delta^2+(K-\ell-1)\bar\ell^2$,
\begin{equation}\label{eq:E2id}
E_2-K\delta^2=\frac{2\ell(K-\ell)}K-\frac{K+2}4 .
\end{equation}
Indeed, writing $m=K-\ell$,
$(\ell-1)m^2+(m-1)\ell^2=K\ell m-(\ell^2+m^2)=(K+2)\ell m-K^2$ and
$(K-2\ell)^2=(m-\ell)^2=K^2-4\ell m$, so that
$K^2\big(E_2-K\delta^2\big)=(K+2)\ell m-K^2-\tfrac{K-2}4\big(K^2-4\ell m\big)=2K\ell m-\tfrac{K+2}4K^2$.

Now average $\E_\prec(I_j-\bar\ell)(I_k-\bar\ell)=\Phi(s_j)\Phi(s_k)+\zeta_{jk}$
over the levels admitted by each type (a level tied with $i$ equals $b$; a
free level is uniform on the $K-1$ values $\ne b$; two free untied levels
are uniform on the ordered distinct pairs among those $K-1$ values). The
$\zeta$ term contributes $+\tfrac1{12}$ exactly when $s_j=s_k\in\{b,a\}$ and
$-\tfrac1{12}$ exactly when $\{s_j,s_k\}=\{b,a\}$; in the ``one tied'' type
this requires $s_k=a$, of probability $\tfrac1{K-1}$, and in the ``all
distinct'' type it cannot occur, since neither level equals~$b$. Using
\eqref{eq:Phizero} and $\sum_{v\ne b}\Phi(v)^2=E_2-\delta^2$,
\begin{equation}\label{eq:types}
\begin{aligned}
\text{all tied}&:&& \delta^2+\tfrac1{12},\\
\text{exactly one tied with }i&:&& \frac{\delta}{K-1}\sum_{v\ne b}\Phi(v)-\frac1{12(K-1)}
   =-\frac{\delta^2+\tfrac1{12}}{K-1},\\
j\sim k\not\sim i&:&& \frac{E_2-\delta^2}{K-1}+\frac1{12(K-1)},\\
\text{all distinct}&:&& \frac{\big(\sum_{v\ne b}\Phi(v)\big)^2-\sum_{v\ne b}\Phi(v)^2}{(K-1)(K-2)}
   =\frac{2\delta^2-E_2}{(K-1)(K-2)} .
\end{aligned}
\end{equation}
Each constant is $\frac1K\sum_{0\le b<a\le K-1}$ of the corresponding line of
\eqref{eq:types}. Since $\delta=\frac{K-2\ell}{2K}$ and $E_2$ depend on
$(b,a)$ only through $\ell$, which occurs with multiplicity $K-\ell$, and
since $\sum_{\ell=1}^{K-1}\ell(K-\ell)^2=\frac{K^2(K^2-1)}{12}$,
\begin{equation}\label{eq:faul}
\sum_{b<a}1=\binom K2,\qquad
\sum_{b<a}\delta^2=\frac{(K-1)(K-2)}{24},\qquad
\sum_{b<a}\!\big(E_2-K\delta^2\big)=\frac{K(K-1)(K-2)}{24},
\end{equation}
the last by \eqref{eq:E2id}, whence $\sum_{b<a}E_2=\frac{K(K-1)(K-2)}{12}$.
Substituting \eqref{eq:faul} into \eqref{eq:types} gives the four values
in \eqref{eq:consts}.

The two coincidences in \eqref{eq:consts} are one identity, not two. Put
\[
\Theta:=\sum_{0\le b<a\le K-1}\Big[E_2-K\delta^2-\frac{K-2}{12}\Big].
\]
Subtracting the lines of \eqref{eq:types} termwise,
\[
\varpi_{\rm one}-\varpi_{\rm dist}=\frac{\Theta}{K(K-1)(K-2)},\qquad
\varpi_{\rm pair}-\varpi_{\rm all}=\frac{\Theta}{K(K-1)},
\]
and by \eqref{eq:faul}, $\Theta=\frac{K(K-1)(K-2)}{24}-\frac{K-2}{12}\binom K2=0$.
Equivalently, by \eqref{eq:E2id} the summand of $\Theta$ is
$\frac{2\ell(K-\ell)}K-\frac{K+1}3$, and
$\frac2K\sum_\ell\ell(K-\ell)^2=\frac{K(K-1)(K+1)}6=\frac{K+1}3\binom K2$.
Hence $\varpi_{\rm one}=\varpi_{\rm dist}$ and $\varpi_{\rm pair}=\varpi_{\rm all}$:
the first removes the $d_i$-dependence of the $n_t$ coefficient in (3),
the second merges $\binom{n_t}2$ and $n_c$ into the single statistic
$\sum_{\{j,k\}\subseteq\calN(i)}\tau_{jk}$.

Substituting \eqref{eq:consts} into (3):
the coefficient of $n_c$ is $\frac{K-1}6$; the coefficient of
$\binom{n_t}2$ is $\frac{K-1}6$ as well, so the two combine into
$\frac{K-1}6\sum_{\{j,k\}\subseteq \calN(i)}\tau_{jk}$; the coefficient of
$n_t$ is $\varpi_1^{\mathrm t}-\varpi_1^{\mathrm f}=\frac{K-2}{12}$, since $\varpi_{\rm one}=\varpi_{\rm dist}$; the constant is $\mathsf c_i=d_i\varpi_1^{\mathrm f}+d_i(d_i-1)\varpi_{\rm dist}$.
The identity $\varpi_{\rm one}=\varpi_{\rm dist}$ is what removes the
$d_i$-dependence of the $n_t$ coefficient. This proves
\eqref{eq:perplayer}.

\emph{(5) Summation.} A tied pair $\{j,k\}$ contributes
$\frac{K-2}{12}$ to $\bar H_j$ and to $\bar H_k$ if $\{j,k\}\in E$, and
$\frac{K-1}6$ to $\bar H_i$ for every common neighbour~$i$. Finally
$\E[\bar H_i]=\mathsf c_i+\frac{K-1}6\binom{d_i}2\frac1K+\frac{K-2}{12}\frac{d_i}K
=\frac{d_i(K^2-1)}{12K}$ gives
$\mathsf c_C=\frac{K^2-1}{12K}\sum_id_i-\frac1K\big[\frac{K-2}6|E|+\frac{K-1}6\sum_i\binom{d_i}2\big]$.
\end{proof}

\begin{corS}[complete graph]\label{cor:law-SM}
On $K_N$, with $A(s)$ the number of tied pairs, $\E[H\mid A=a]=c_1a+c_0$ for the index protocol, where
$c_1=\frac{(K-1)N-K}6$ and $c_0=\frac{N(N-1)[K^2+K-1-(K-1)N]}{12K}$.
\end{corS}
\begin{proof}
On $K_N$ every pair is an edge with $N-2$ common neighbours, so the
weight of every tied pair is
$\frac{K-2}6+(N-2)\frac{K-1}6=\frac{(K-1)N-K}6=c_1$, and Theorem~\ref{thm:pattern-SM}
gives $\E[\bar H\mid\tau]=c_1A(s)+c_0$ with $A=\sum_{j<k}\tau_{jk}$
and $c_0=\mathsf c_{K_N}=\frac{N(N-1)[K^2+K-1-(K-1)N]}{12K}$ (a direct
simplification). Conditioning on $A$ and using the tower property gives
$\E[\bar H\mid A=a]=c_1a+c_0$. It remains to replace $\bar H$ by the
fixed-order energy $H$: for a permutation $\sigma$ of the players,
$H(\sigma s)$ is the energy of the game with tie-breaking order
$\sigma$; since the law of $s$ given $A$ is exchangeable in the players,
$\E[H\mid A]=\E[H\circ\sigma\mid A]$ for every $\sigma$, hence
$\E[H\mid A]=\E[\bar H\mid A]$.
\end{proof}

\begin{figure}[htbp]\centering
\includegraphics[width=0.7\textwidth]{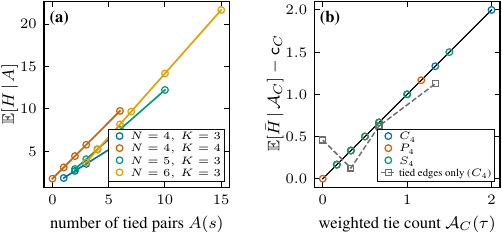}
\caption{The tie-pattern law against exhaustive enumeration (floating point; not part of any proof).
(a)~Complete graph, index protocol: $\E[H\mid A]$ against the number $A$ of tied pairs (points) and the affine law of
Corollary~\ref{cor:law-SM} (lines).
(b)~$C_4$, $P_4$, $S_4$, $K=3$: the order-averaged energy $\E[\bar H\mid\calA_C]$ lies on the line of slope one given by
\eqref{eq:perplayer}; conditioning on tied conflict edges alone ($C_4$, grey, shifted for display) is not
affine.}
\label{fig:law-SM}
\end{figure}

Theorem~\ref{thm:pattern-SM} was also checked against exhaustive enumeration on $P_4$, $S_4$, $C_4$, $C_5$ and random
graphs with $N\le6$, $K\le5$ (Figure~\ref{fig:law-SM}). The full pattern is needed: on $C_4$, $P_4$ and $S_4$ ($K=3$,
averaged over all $24$ orders) the conditional energy equals $\calA_C(\tau)+\mathsf c_C$ to within $4\times10^{-10}$,
whereas conditioning on the number of tied conflict edges alone is not monotone on $C_4$.

\subsection{Randomised tie-breaking}
A deterministic protocol amplifies the harmonic component but does not create it.

\begin{propS}[randomised tie-breaking]\label{prop:randtie-SM}
If ties are broken uniformly at random and payoffs are taken in expectation, $\Epot$ is as in
Proposition~\ref{prop:energies} and $\Eharm$ is multiplied by $(K-2)/(K+1)$. The game is an exact potential game iff
$K=2$.
\end{propS}
\begin{proof}
Now $\psi(x,y)=-\frac12-\frac12\sgn(y-x)$ for every ordered pair, so again $\psi(x,y)+\psi(y,x)\equiv-1$,
$\omega^+\equiv0$, and the own effect is unchanged. Moreover $\Var\psi=\frac{K-1}{4K}$, so
$\|\omega\|^2=\frac{K-1}{4K}-\frac{K^2-1}{6K^2}=\frac{(K-1)(K-2)}{12K^2}$, which is $(K-2)/(K+1)$ times the deterministic
value $\frac{K^2-1}{12K^2}$ of Proposition~\ref{prop:energies}. Apply Proposition~\ref{prop:pairwise}(c),(d).
\end{proof}
What creates the harmonic component is the antisymmetry of the comparison, $\psi(x,y)+\psi(y,x)\equiv-1$; the
deterministic protocol adds to it the resolution of ties, which is why the factor vanishes at $K=2$, where every
randomly resolved tie is a coin flip with no residual interaction.

\section{The continuum game and the proof of Theorem~\ref{thm:spectrum}}\label{sec:cont}

\subsection{Setting}
Let $N$ players choose $s_i\in[0,1]$ i.i.d.\ uniformly,
$\pos_i(s)=\#\{j\ne i:s_j>s_i\}$ and $u_i=f(\pos_i)$. Flows are
functions $x_i(t\to t';s_{-i})$, antisymmetric in $(t,t')$, with
$\|x\|^2=\sum_i\E_{s_{-i}}\big[\frac12\iint x_i(t\to t';s_{-i})^2\,dt\,dt'\big]$;
gradients are $(\delta\phi)_i=\phi(t',s_{-i})-\phi(t,s_{-i})$, game
flows are $u_i(t',s_{-i})-u_i(t,s_{-i})$, and $P$ is the orthogonal
projection onto the closure of $\operatorname{im}\delta$. For
$v,w\in L^2([0,1]^N)$ one has
$\langle\delta_iv,\delta_iw\rangle_i=\E[(v-A_iv)(w-A_iw)]$ because
$\frac12\iint(v(t')-v(t))(w(t')-w(t))\,dt\,dt'=\Cov_{t}(v,w)$.

\begin{propS}[Hodge projection is the inverse number operator]\label{prop:Ninv}
With $g=\sum_i(u_i-A_iu_i)$, $Pu=\delta\phi^\star$ with
$\phi^\star=\numop^{-1}g$, and
\begin{equation}\label{eq:contP}
\|Pu\|^2=\sum_{m\ge1}\frac{\|g_m\|^2}m
=\sum_{T\ne\emptyset}\frac1{|T|}\Big\|\sum_{i\in T}[u_i]_T\Big\|^2,\qquad
\|u\|^2=\sum_i\|u_i-A_iu_i\|^2=\sum_i\sum_{T\ni i}\|[u_i]_T\|^2 .
\end{equation}
\end{propS}
\begin{proof}
$\|u-\delta\phi\|^2=\sum_i\|(I-A_i)(u_i-\phi)\|^2$ is minimised at the
solution of $\sum_i(I-A_i)\phi=\sum_i(I-A_i)u_i$, i.e.\ $\numop\phi=g$;
$\numop$ is invertible on mean-zero functions (it multiplies the order-$m$
component by~$m$), and
$\|\delta\phi^\star\|^2=\langle\numop\phi^\star,\phi^\star\rangle=\langle g,\numop^{-1}g\rangle$.
The second expression for $\|u\|^2$ is Parseval for the Hoeffding
decomposition, since $u_i-A_iu_i$ consists of the components with
$T\ni i$.
\end{proof}

\subsection{The two energies (\texorpdfstring{eq.~\eqref{eq:contquad}}{the two quadratic forms} of the main text)}

\begin{lemmaS}[Dirichlet variance identity]\label{lem:dir}
$\Var(u_i)=\Var_{k\sim U\{0,\dots,N-1\}}f(k)$ and
$\Var(\bar u_i)=\frac1{N+1}\Var(u_i)$; hence
$\|u\|^2=\frac{N^2}{N+1}\Var_kf(k)$.
\end{lemmaS}
\begin{proof}
$\pos_i$ is uniform on $\{0,\dots,N-1\}$ (positions form a uniformly
random permutation). Given $s_{-i}$ with order statistics
$z_1>\dots>z_{N-1}$ and spacings $\varsigma_k=z_k-z_{k+1}$ ($z_0=1$, $z_N=0$),
$\pos_i=k$ iff $z_{k+1}<s_i<z_k$, so $\bar u_i=\sum_kf(k)\varsigma_k$. The
spacings are Dirichlet$(1,\dots,1)$ with $N$ cells:
$\E\varsigma_k=\frac1N$, $\E\varsigma_k\varsigma_l=\frac{1+\delta_{kl}}{N(N+1)}$.
Hence $\E\bar u_i^2=\frac1{N(N+1)}\big[(\sum_kf_k)^2+\sum_kf_k^2\big]$ and
$\Var\bar u_i=\frac1{N+1}\big[\frac1N\sum f_k^2-\frac1{N^2}(\sum f_k)^2\big]
=\frac1{N+1}\Var_kf$. Finally
$\|u\|^2=\sum_i[\Var u_i-\Var\bar u_i]=N\cdot\frac{N}{N+1}\Var_kf$.
\end{proof}

\begin{lemmaS}[Bernoulli-sum decomposition]\label{lem:bern}
Let $B_t=\#\{j\ne i:s_j>t\}$ and $\kappa(t,t')=\min(t,t')-tt'$. For
$T\not\ni i$ with $|T|=m$,
\[
[\bar u_i]_T=\int_0^1\prod_{j\in T}\big(\one[s_j>t]-(1-t)\big)\,d_m[f](1-t)\,dt,\qquad
d_m[f](q)=\E\big[\Delta^mf(\Bin(N-1-m,q))\big],
\]
and this expression does not depend on $i\notin T$. Consequently
\begin{equation}\label{eq:gmJm}
\|g_m\|^2=(N-m)^2J_m,\qquad
J_m:=\binom Nm\iint\kappa(q,q')^m\,d_m[f](q)\,d_m[f](q')\,dq\,dq',
\end{equation}
and $\|Pu\|^2=\sum_m\frac{(N-m)^2}mJ_m$, $\|u\|^2=N\sum_m(N-m)J_m$.
\end{lemmaS}
\begin{proof}
$\bar u_i(s_{-i})=\int_0^1f(B_t)\,dt$, since $\pos_i=B_{s_i}$ and $s_i$ is
uniform. For fixed $t$, $B_t$ is a sum of i.i.d.\ Bernoulli$(1-t)$
variables $\xi_j=\one[s_j>t]$; for a function $F$ of such a sum, the
Hoeffding component on $T$ is
$\prod_{j\in T}(\xi_j-(1-t))\cdot\E[\Delta^{|T|}F(B^{(T^c)}_t)]$
(expand $F(\xi_T+B^{(T^c)})$ in the Walsh basis
$\prod_{j\in T}(\xi_j-(1-t))$; each centred factor acts as a forward
difference). Integrating over $t$ gives the formula, whose right-hand
side involves $i$ only through the size $N-1-m$ of the complement.
By Lemma~\ref{lem:cons}, $g_T=-\sum_{i\notin T}[\bar u_i]_T=-(N-m)\times$ the
common expression; its squared norm is computed using
$\E\big[(\one[s>t]-(1-t))(\one[s>t']-(1-t'))\big]=\kappa(t,t')$ for each of
the $m$ coordinates, which gives \eqref{eq:gmJm} after the substitution
$q=1-t$, under which $\kappa$ is invariant. The formula for
$\|Pu\|^2$ is \eqref{eq:contP}. For $\|u\|^2$: the same computation gives
$\Var(\bar u_i)=\sum_m\binom{N-1}m\iint\kappa^md_md_m$, and Lemma~\ref{lem:dir}
yields $\|u\|^2=N\big[\Var u_i-\Var\bar u_i\big]=N^2\Var\bar u_i
=N^2\sum_m\binom{N-1}m\iint\kappa^md_md_m=N\sum_m(N-m)J_m$, using
$N\binom{N-1}m=(N-m)\binom Nm$.
\end{proof}

\begin{proof}[Proof of Corollary~\ref{cor:Qinf}, with the equality case]
Since $\rho^{(\infty)}_\pot(f)=\|Pu\|^2/\|u\|^2$, it suffices to show that, by Lemma~\ref{lem:bern},
\[
\frac{N-1}N\|u\|^2-\|Pu\|^2=\sum_m(N-m)J_m\Big[(N-1)-\frac{N-m}m\Big]
=\sum_{m\ge1}(N-m)J_m\frac{N(m-1)}m\ \ge0,
\]
since $J_m\ge0$ ($\kappa^m$ is a positive semidefinite kernel). Equality holds iff $J_m=0$ for all
$m\ge2$, i.e.\ $d_m[f]\equiv0$; as $q\mapsto\E[\Delta^mf(\Bin(n,q))]$ is the
Bernstein polynomial of $\Delta^mf$, this holds iff $\Delta^mf\equiv0$
for all $m\ge2$, i.e.\ iff $f$ is affine.
\end{proof}

\subsection{Proof of Theorem~\ref{thm:spectrum}}

\begin{mainthm}[\ref{thm:spectrum} of the main text]
Let $t_p(k)=\sum_{r=0}^p(-1)^r\binom pr\binom{p+r}r\frac{\ff kr}{\ff{(N-1)}r}$, $p=1,\dots,N-1$, be the discrete
Chebyshev polynomials, and $J_m(f,g)=\binom Nm\iint\kappa^m\,d_m[f](q)\,d_m[g](q')\,dq\,dq'$ the order-$m$ form, so that
$\mathcal E_\pot(f,g):=\sum_m\frac{(N-m)^2}mJ_m(f,g)$ and $\mathcal E_{\rm tot}(f,g):=N\sum_m(N-m)J_m(f,g)$ polarise
$\|Pu\|^2$ and $\|u\|^2$.
\begin{enumerate}
\item[(a)] For every $m$ and all $p\ne p'$, $J_m(t_p,t_{p'})=0$; moreover $J_m(t_p,t_p)=0$ for $m>p$ and, for $m\le p$,
\begin{equation}\label{eq:JdiagS}
J_m(t_p,t_p)=\binom Nm\frac{m\,(p+m)!}{p(p+1)(2p+1)\,(p-m)!\,\big[\ff{(N-1)}m\big]^2}.
\end{equation}
\item[(b)] $\displaystyle\lambda_p=\frac{\mathcal E_\pot(t_p,t_p)}{\mathcal E_{\rm tot}(t_p,t_p)}
=\frac{(N-p)(N+p+1)}{N\,p(p+1)}\Big[1-\frac{\binom Np}{\binom{N+p+1}p}\Big]$, and $\lambda_{p+1}<\lambda_p$ for
$1\le p\le N-2$.
\end{enumerate}
\end{mainthm}
Throughout, $q=1-t$, $\kappa=\min(q,q')-qq'$ and $d_m[f](q)=\E[\Delta^mf(\Bin(N-1-m,q))]$ as in Lemma~\ref{lem:bern};
$n:=N-1$; $\tilde P_p(q):=P_p(1-2q)=\sum_{r=0}^p(-1)^r\binom pr\binom{p+r}r q^r$ is the shifted Legendre polynomial,
$\tilde P^{(m)}_j(q):=P^{(m,m)}_j(1-2q)$ the shifted Jacobi (Gegenbauer) polynomial, and $\mathsf w_m(q):=[q(1-q)]^m$.

\paragraph{(i) From Chebyshev to Legendre and Jacobi polynomials.}
Let $\mathsf b_f(q):=\E f(\Bin(n,q))$ be the Bernstein polynomial of $f$. For $f=\ff kr$ one has
$\E[\ff{\Bin(n,q)}r]=\ff nr\,q^r$ and $\Delta^m\ff kr=\ff rm\,\ff k{r-m}$, hence
$d_m[\ff kr](q)=\ff rm\,\ff{(n-m)}{r-m}\,q^{r-m}=\frac1{\ff nm}\frac{d^m}{dq^m}\big(\ff nr\,q^r\big)$, because
$\ff nr=\ff nm\,\ff{(n-m)}{r-m}$. By linearity,
\begin{equation}\label{eq:bernder}
d_m[f]=\frac{1}{\ff nm}\,\mathsf b_f^{(m)}\qquad\text{for every }f,
\end{equation}
which is the classical derivative formula for Bernstein polynomials. Moreover
$\mathsf b_{t_p}(q)=\sum_r(-1)^r\binom pr\binom{p+r}rq^r=\tilde P_p(q)$: the binomial law maps the discrete Chebyshev
polynomials to the Legendre polynomials. Since $\frac{d^m}{dx^m}P_p(x)=\frac{(p+m)!}{2^m\,p!}P^{(m,m)}_{p-m}(x)$ and
$\frac d{dq}=-2\frac d{dx}$ for $x=1-2q$,
\begin{equation}\label{eq:dmtp}
d_m[t_p](q)=c_{p,m}\,\tilde P^{(m)}_{p-m}(q),\qquad c_{p,m}:=\frac{(-1)^m(p+m)!}{p!\,\ff nm}\quad(m\le p),
\end{equation}
and $d_m[t_p]=0$ for $m>p$.

\paragraph{(ii) The kernel $\kappa^m$ carries the Jacobi weight.}
Let $\mathsf K_mg(q):=\int_0^1\kappa(q,q')^mg(q')\,dq'$.
\begin{lemmaS}\label{lem:Km}
{\rm(a)} If $g$ is a polynomial of degree at most $b$, then $\mathsf K_mg=\mathsf w_m\,R$ with $R$ a polynomial of
degree at most $b$. {\rm(b)} For every $j\ge0$,
\[
\mathsf K_m\tilde P^{(m)}_j=\mathsf k_{m,j}\,\mathsf w_m\,\tilde P^{(m)}_j,\qquad\mathsf k_{m,j}=\frac{m}{(j+m)(j+m+1)} .
\]
\end{lemmaS}
\begin{proof}
(a) It suffices to take $g=q^b$. Since $\kappa(q,s)=s(1-q)$ for $s\le q$ and $q(1-s)$ for $s\ge q$,
\[
\mathsf K_mg(q)=(1-q)^m\frac{q^{m+b+1}}{m+b+1}+q^m\,G(q),\qquad G(q):=\int_q^1(1-s)^ms^b\,ds .
\]
$G$ is a polynomial of degree $m+b+1$ with $G^{(i)}(1)=0$ for $i\le m$, so $G=(1-q)^{m+1}G_1$ with $\deg G_1=b$; the
leading coefficient of $G$ is $-(-1)^m/(m+b+1)$, hence that of $G_1$ is $1/(m+b+1)$. Thus
$\mathsf K_mg=\mathsf w_m\big[\frac{q^{b+1}}{m+b+1}+(1-q)G_1\big]$, and in the bracket the coefficients of $q^{b+1}$
cancel.

(b) $\mathsf K_m$ is self-adjoint on $L^2(0,1)$. Write $\mathsf K_m\tilde P^{(m)}_j=\mathsf w_mR$ with $\deg R\le j$. For
a polynomial $S$ of degree $<j$, part (a) gives $\mathsf K_mS=\mathsf w_mR_S$ with $\deg R_S<j$, so
$\int S\,\mathsf w_m R=\int(\mathsf K_mS)\tilde P^{(m)}_j=\int\mathsf w_mR_S\,\tilde P^{(m)}_j=0$, because the
$\tilde P^{(m)}_j$ are orthogonal for the weight $\mathsf w_m$. Hence $R$ is orthogonal, for that weight, to all
polynomials of lower degree, i.e.\ $R=\mathsf k\,\tilde P^{(m)}_j$. To find $\mathsf k$ let $q\to0$: the first term of the
display above is $O(q^{m+1})$, so
$\mathsf K_m\tilde P^{(m)}_j(q)=q^m\int_0^1(1-s)^m\tilde P^{(m)}_j(s)\,ds+O(q^{m+1})$, while
$\mathsf w_m\tilde P^{(m)}_j=q^m\tilde P^{(m)}_j(0)+O(q^{m+1})$ with $\tilde P^{(m)}_j(0)=P^{(m,m)}_j(1)=\binom{j+m}j$.
The Rodrigues formula reads $\mathsf w_m\tilde P^{(m)}_j=\frac1{j!}\frac{d^j}{ds^j}\big[s^{j+m}(1-s)^{j+m}\big]$, and
$j$ integrations by parts (all boundary terms vanish) give
\begin{align*}
\int_0^1(1-s)^m\tilde P^{(m)}_j\,ds&=\frac1{j!}\int_0^1s^{-m}\frac{d^j}{ds^j}\big[s^{j+m}(1-s)^{j+m}\big]ds\\
&=\frac{m(m+1)\cdots(m+j-1)}{j!}\int_0^1(1-s)^{j+m}ds=\frac{\binom{m+j-1}j}{j+m+1}.
\end{align*}
Dividing by $\binom{j+m}j$ gives $\mathsf k_{m,j}=\frac{m}{(j+m)(j+m+1)}$.
\end{proof}

\paragraph{(iii) Every order form is diagonal; the potential form.}
By \eqref{eq:dmtp} and Lemma~\ref{lem:Km}(b), for $m\le\min(p,p')$,
\[
J_m(t_p,t_{p'})=\binom Nm c_{p,m}c_{p',m}\,\mathsf k_{m,p'-m}\int_0^1\mathsf w_m\,\tilde P^{(m)}_{p-m}\tilde P^{(m)}_{p'-m}\,dq,
\]
which vanishes for $p\ne p'$ by Jacobi orthogonality; this is (a), the case $m>\min(p,p')$ being trivial. For $p=p'$,
$\mathsf k_{m,p-m}=\frac m{p(p+1)}$ and the Jacobi norm
$\int_0^1\mathsf w_m\big(\tilde P^{(m)}_{p-m}\big)^2dq=\frac{(p!)^2}{(2p+1)\,(p+m)!\,(p-m)!}$ give \eqref{eq:JdiagS}.
In particular the $t_p$ diagonalise $\sum_mw_mJ_m$ for any weights $w_m$, hence both $\mathcal E_\pot$ and
$\mathcal E_{\rm tot}$ (for the latter this is also clear from
$\mathcal E_{\rm tot}(f,f')=\frac N{N+1}\sum_k(f_k-\bar f)(f'_k-\bar f')$, Lemma~\ref{lem:dir}). Finally, with
$\binom Nm(N-m)^2/[\ff nm]^2=N\,(N-m)!/[m!\,(N-1)!]$,
\[
\mathcal E_\pot(t_p,t_p)=\sum_{m=1}^p\frac{(N-m)^2}mJ_m(t_p,t_p)
=\frac{N}{(N-1)!\,p(p+1)(2p+1)}\sum_{m=1}^p\frac{(N-m)!\,(p+m)!}{m!\,(p-m)!}.
\]
In the same way $N(N-m)J_m(t_p,t_p)\propto(N-m-1)!\,(p+m)!/[(m-1)!\,(p-m)!]$, which is the distribution $\eta_{p,m}$ of
the energy of $t_p$ over the orders $m=1,\dots,p$; the higher the order, the smaller its weight $\frac{N-m}{mN}$ in $\lambda_p$.

\paragraph{(iv) Normalisation and closed form.}
The Hahn-polynomial norm (parameters $(0,0)$) gives
$\sum_{k=0}^{N-1}t_p(k)^2=\frac{(N+p)!\,(N-1-p)!}{(2p+1)\,(N-1)!^2}$, so
$\mathcal E_{\rm tot}(t_p,t_p)=\frac N{N+1}\cdot\frac{(N+p)!(N-1-p)!}{(2p+1)(N-1)!^2}$ and
\[
\lambda_p=\frac{(N+1)!}{N\,p(p+1)\,(N-p-1)!\,(N+p)!}\sum_{m=1}^p\frac{(p+m)!\,(N-m)!}{m!\,(p-m)!}.
\]
Adding the $m=0$ term $N!$, the sum becomes
\[
N!\;{}_2F_1(-p,\,p+1;\,-N;\,1)=\frac{(N+p+1)!\,(N-p)!}{(N+1)!}
\]
by the Chu--Vandermonde identity,\footnote{The lower parameter $-N$ is a
nonpositive integer, so ${}_2F_1$ is meant here as the terminating series
determined by the upper parameter $-p$. This is legitimate: the series
stops after $p+1\le N$ terms, and $(-N)_m\ne0$ for $0\le m\le p\le N-1$,
so every term is finite. Chu--Vandermonde in the form
${}_2F_1(-p,b;c;1)=\frac{(c-b)_p}{(c)_p}$ is an identity between rational
functions of $c$, and $(c)_p$ does not vanish at $c=-N$ when $p\le N$; it
may therefore be evaluated at $c=-N$. Explicitly,
$(-N)_p=(-1)^p\frac{N!}{(N-p)!}$ and $(-N-p-1)_p=(-1)^p\frac{(N+p+1)!}{(N+1)!}$.}
which yields the stated closed form
with
\[
\varepsilon_p:=\frac{\binom Np}{\binom{N+p+1}p}=\prod_{i=0}^{p-1}\frac{N-i}{N+p+1-i}.
\]
For $p=1$, $\varepsilon_1=\frac N{N+2}$ and $\lambda_1=\frac{N-1}N$.

\paragraph{(v) Monotonicity.}
Put $D_p=\frac{N(N+1)}{p(p+1)}-1$, so $N\lambda_p=D_p(1-\varepsilon_p)$,
$D_p-D_{p+1}=\frac{2N(N+1)}{p(p+1)(p+2)}$ and
$\varepsilon_{p+1}=\varepsilon_p\frac{N-p}{N+p+2}$. Then $\lambda_p>\lambda_{p+1}$ reads
$D_p(1-\varepsilon_p)>D_{p+1}\big(1-\varepsilon_p\frac{N-p}{N+p+2}\big)$; since
$D_p-D_{p+1}\frac{N-p}{N+p+2}=(D_p-D_{p+1})+D_{p+1}\frac{2(p+1)}{N+p+2}$, this is
$(D_p-D_{p+1})(1-\varepsilon_p)>\varepsilon_p\,D_{p+1}\frac{2(p+1)}{N+p+2}$, i.e., with
$D_{p+1}=\frac{(N-p-1)(N+p+2)}{(p+1)(p+2)}$,
\[
\frac{1-\varepsilon_p}{\varepsilon_p}>\frac{2(p+1)D_{p+1}}{(D_p-D_{p+1})(N+p+2)}=\frac{p(p+1)(N-p-1)}{N(N+1)} .
\]
Now $\varepsilon_p^{-1}=\prod_{i=0}^{p-1}\frac{N+p+1-i}{N-i}=\prod_{i=0}^{p-1}\Big(1+\frac{p+1}{N-i}\Big)\ge1+\frac{p(p+1)}N$,
so the left-hand side is at least $\frac{p(p+1)}N$, which exceeds the right-hand side because $N-p-1<N+1$. \qed

\begin{remarkS}
Every step is audited in exact arithmetic by \texttt{audits/check\_spectrum.py} (\eqref{eq:bernder}, \eqref{eq:dmtp},
Lemma~\ref{lem:Km}, \eqref{eq:JdiagS} against direct evaluation of the double integrals, and the closed forms of
$\lambda_p$ and of the total form for $3\le N\le9$).
\end{remarkS}

\section{Priced priorities}\label{sec:priced}
This appendix proves the statement made after Corollary~\ref{cor:rigid} of the main text: pricing can change
behaviour drastically while leaving $\Eharm$ fixed, and neither energy predicts the existence of pure equilibria. The
protocol is the index protocol and $C$ is an arbitrary conflict graph.

\begin{propS}[priced priorities]\label{prop:priced-SM}
Let $u_i(s)=-\pos_C(i,s)-c\,s_i$ with $c>0$, and $b_i:=\#\{j\in\calN(i):j<i\}$.
\begin{enumerate}
\item[(a)] $\Eharm$ is as in Proposition~\ref{prop:energies} and $\Epot=c_K\sum_i(d_i-cK)^2$. On a $d$-regular graph
the game is purely harmonic at $c=d/K$, where marginal price equals marginal expected benefit, and uniform mixing is
then a Nash equilibrium.
\item[(b)] If $0<c<1$ and $E\neq\emptyset$, there is no pure Nash equilibrium.
\item[(c)] The all-zero profile is a Nash equilibrium iff $c\ge\max_ib_i$.
\end{enumerate}
\end{propS}
\begin{proof}
(a) Pricing replaces the own effect $\mathsf a_i(x)=d_i(x-\bar s)/K$ of Proposition~\ref{prop:energies} by
$(d_i/K-c)(x-\bar s)$ and nothing else; apply Proposition~\ref{prop:pairwise}(c). If all other agents mix uniformly,
$\E[u_i\mid s_i=x]=\mathrm{const}+x(d_i/K-c)$, constant in $x$ at $c=d/K$.

(b) Let $s$ be a pure equilibrium. For an agent with $s_i\ge1$ let $L_i\in\mathbb{Z}_{\ge0}$ be the number of positions
lost by lowering one level; no profitable deviation means $L_i\ge c>0$, so $L_i\ge1$. For an agent with $s_i\le K-2$
let $W_i$ be the number of neighbours overtaken by raising one level; then $W_i\le c<1$, so $W_i=0$. In particular
isolated agents sit at level $0$. If all agents are at level $0$, take an edge $\{i,j\}$ with $i<j$: raising $j$ to
level $1$ overtakes $i$, so $W_j\ge1$, a contradiction. Otherwise let $s_{\max}\ge1$ be the highest occupied level and
$m$ the largest-index agent at it. Since $L_m\ge1$, some neighbour $j$ is served before $m$ after $m$ lowers but not
before; this requires either $s_j=s_{\max}$ with $j>m$, excluded by the choice of $m$, or $s_j=s_{\max}-1$ with $j<m$.
In the latter case $j$, raising one level, ties $m$ and is favoured by the protocol, so $W_j\ge1$, a contradiction.

(c) At $s=0$ every agent is tied with all neighbours, so $u_i=-b_i$. Raising to level $t\ge1$ puts $i$ ahead of every
neighbour at payoff $-ct$, best at $t=1$. The deviation is strictly profitable iff $c<b_i$.
\end{proof}

Thus pure equilibria exist at $c=0$, vanish on all of $(0,1)$, and reappear for larger $c$, while $\Eharm$ does not
move. Part~(b) uses the index protocol: on a triangle with the cyclic protocol and $K=2$, the profile in which everybody
bids $1$ is an equilibrium for every $c<1$.

\paragraph{Exact computation on the intersection.} Scanning $c$ over $91$ values in $[0,4.5]$ on the eight-vehicle
intersection ($K=3$), with the methods of Section~\ref{sec:numerics} of the main text, the potential share follows
(a) to machine precision, from $4/5$ at $c=0$ through exactly $0$ at $c=d/K=4/3$ back to $0.96$ at $c=4.5$, while
the projected harmonic energy is constant to $10^{-14}$ and the entropy production at $\beta\max|\Delta u|=0.02$ stays
within $3\times10^{-5}$ of the prediction of Proposition~\ref{prop:learning}. Pure equilibria number three at $c=0$,
zero on $(0,1)$ as in (b), and one from $c=\max_ib_i=4$ on as in (c); in between they exist only at $c=1$ and $c=3$,
a feature of this graph (a star whose centre has the largest index has one for every $c\in(1,\max_ib_i)$).

\section{Reproducibility and independent audit}\label{app:repro}
The proof of Theorem~\ref{thm:finiteK} is computer-assisted at three places only, namely the root isolation of
Step~\ref{step:K4}, the symbolic identities in $N$ of Step~\ref{step:K3}, and the certificates of
Step~\ref{step:certs}. Every other step, and all of Appendices~\ref{sec:pattern}--\ref{sec:priced}, is proved by
hand. Table~\ref{tab:scripts} lists the scripts of those three steps; they use exact arithmetic only (Python
\texttt{fractions} or SymPy), and \texttt{run\_checks.sh} runs them all in well under a minute. The floating-point
computations of Section~\ref{sec:numerics} of the main text and of Appendix~\ref{sec:priced} enter no proof and are
kept outside the package.

\begin{table}[H]\centering\small
\caption{Scripts on which the proof of Theorem~\ref{thm:finiteK} depends, and the two independent checks.}\label{tab:scripts}
\begin{tabular}{@{}>{\raggedright\arraybackslash}p{0.33\linewidth}>{\raggedright\arraybackslash}p{0.09\linewidth}>{\raggedright\arraybackslash}p{0.52\linewidth}@{}}\toprule
script & step & establishes\\\midrule
\texttt{verify\_certificates.py} + \texttt{certificates/} ($25$ files) & \ref{step:certs} & the pairs of Table~\ref{tab:cov}: rebuilds each certified matrix from Theorem~\ref{thm:SOS} and Step~\ref{step:bond}, checks $LDL^\top=A$ with positive pivots and the kernel; standard library only, imports nothing\\
\texttt{audit\_K4.py} & \ref{step:K4} & the cubic of Lemma~\ref{lem:K4root}, exact root counting for $\varrho\ge\varrho_0$, and (T) for all bonds iff $N\ge10$ (Theorem~\ref{thm:K4}); also audits the hand-proved closed forms of Steps~\ref{step:Sqlow}--\ref{step:perron}\\
\texttt{k3\_representation\_proof.py}, \texttt{s1\_closed.py}, \texttt{check\_blocks.py}, \texttt{verify\_M\_symbolic.py} & \ref{step:K3} & Lemmas~\ref{lem:K3rep}, \ref{lem:K3blocks} and~\ref{lem:K3M}: the block representation, the closed forms of its entries, the row sums $\mathsf M\one$, and the certificates for $N=17,18$\\
\texttt{k2\_witness.py} & \ref{step:assembly} & the $K=2$, $N=7$ witness from the definition of the Hodge projection; $Q_2\one=0$; the sign change between $N=6$ and $N=7$\\
\texttt{end\_to\_end.py} & -- & Theorem~\ref{thm:finiteK} for $K=3$, $N\le7$; $K=4$, $N\le5$; $K=5$, $N\le4$; $K=6,7$, $N=3$, using no step of the proof (see below)\\\bottomrule
\end{tabular}
\end{table}

\paragraph{An end-to-end check that bypasses the proof.} \texttt{end\_to\_end.py} builds
$Q_K=\frac{N-1}NM_{\rm tot}-M_\pot$ directly from \eqref{eq:rayleigh} (all $K^N$ profiles, the flow of every
increment indicator, and $M_\pot$ from the Hoeffding decomposition of the divergence, Lemma~\ref{lem:lap}) in exact
rational arithmetic, and certifies $Q_K\one=0$ and $Q_K+\one\one^\top=LDL^\top$ with positive pivots. This is
Theorem~\ref{thm:finiteK} for that pair, obtained without the order decomposition, the bond variables or the
sum-of-squares identity; it is feasible while $K^N$ is a few thousand, and at $K=2$ it reports positivity for $N\le6$
and failure at $N=7$. Any other pair: \texttt{python3 end\_to\_end.py K,N}. The package also contains the certificate
generator (\texttt{bond\_certificate.py}, \texttt{full\_space\_certificate.py}, \texttt{make\_certificates.py}), which
the checker does not import, and \texttt{audits/check\_spectrum.py}, a symbolic audit of the hand proof of
Theorem~\ref{thm:spectrum}.

\end{document}